%% file: main.tex
\documentclass{article}
\usepackage{graphicx} 

\input{headers}

\newtheorem{Example}[theorem]{Example}

\newcommand{\hesam}[1]{\textcolor{blue}{#1}}

\usepackage{tikz, calc}
\tikzstyle{vertex} = [fill,shape=circle,node distance=80pt]

\tikzstyle{graphvertex} = [node distance=80pt] 

\tikzstyle{edge} = [fill,opacity=.5,fill opacity=.5,line cap=round, line join=round, line width=50pt]
\tikzstyle{elabel} =  [fill,shape=circle,node distance=30pt]

\pgfdeclarelayer{background}
\pgfsetlayers{background,main}

\title{Group Testing with General Correlation Using
Hypergraphs}
\author{
  {Hesam Nikpey, Saswati Sarkar, Shirin Saeedi Bidokhti}\\
  University of Pennsylvania\\
                    Email: \{Hesam, Swati, Saeedi\}@seas.upenn.edu
}

\begin{document}

\maketitle

\begin{abstract}

    Group testing, a problem with diverse applications across multiple disciplines,  traditionally assumes independence across nodes' states. Recent research, however, focuses on real-world scenarios that often involve correlations among nodes, challenging the simplifying assumptions made in existing models. In this work, we consider a comprehensive model for arbitrary statistical correlation among nodes' states. To capture and leverage these correlations effectively, we model the problem by hypergraphs, inspired by \cite{gonen2022group}, augmented by a probability mass function on the hyper-edges.

    Using this model, we first design a novel greedy adaptive algorithm capable of conducting informative tests and dynamically updating the distribution. {Performance analysis provides upper bounds on the number of tests required, which depend solely on the entropy of the underlying probability distribution and the average number of infections. We demonstrate that the algorithm recovers or improves upon all previously known results for group testing settings with correlation. Additionally, we provide families of graphs where the algorithm is order-wise optimal and give examples where the algorithm or its analysis is not tight.
 }

    {We then generalize the proposed framework of group testing with general correlation in two directions, namely semi-non-adaptive group testing and noisy group testing. In both settings, we provide novel theoretical bounds on the number of tests required.
    }

\end{abstract}

\section{Introduction}\label{sec:intro}

{Group testing is a classical problem that has been extensively studied in the fields of computer science \cite{du2000combinatorial,de2005optimal,sobel1959group, porat2008explicit} and information theory \cite{aldridge2019group,aldridge2014group,wolf1985born,chan2014non}, focusing on efficiently identifying a small number of defective items within a large population by grouping items together in tests. Originally motivated by the need to screen for syphilis during World War II, group testing has recently garnered renewed attention due to its applications in areas such as COVID-19 testing \cite{gollier2020group,narayanan2020accelerated}, privacy-preserving \cite{ibarrondo2023grote}, and DNA sequencing \cite{daniels2023overview}.}

Group testing involves a population (set of nodes) where a subset is infected, and any sub-population can be tested. A positive test result indicates at least one infection within the subset, and the primary goal is to identify all infected individuals. 
Historically, group testing has been studied under the assumption of independent node states (see \cite{li2014group,du2000combinatorial} and the references therein). However, recent advancements have highlighted the oversimplification of this assumption. For example, in disease spread scenarios, individuals within the same household exhibit  {\it correlation} in their risk; if one member is infected, the likelihood of infection among other members increases 
\cite{nikolopoulos2021group2, nikpey2022group, arasli2023group}. 
{To a lesser extent, a weaker correlation exists in the population of a neighbor where an outbreak might occur, and even weaker correlation exists in the population of a city, so different strengths of correlation can be assumed in a general population.} More generally,  the physical and biological interactions that govern disease propagation impose a correlation structure on the states of the nodes. 
Similarly, in network fault diagnosis, geographically localized connections/devices (to which we refer to as nodes) are more likely to experience simultaneous faults due to their reliance on shared physical infrastructure and exposure to common external factors, such as environmental conditions or power supply issues \cite{xu2023performance}. Furthermore, the cascading nature of many failure mechanisms introduces correlations across the states of these nodes, as a fault in one part of the system can increase the likelihood of faults in nearby components. {Examples include an outage of power in some part of the network\cite{bernstein2014power}, an overload in a specific district\cite{athari2017impacts}, or typically localized cascading line failures in the transmission system of the power grid \cite{soltan2014cascading}.}  Modeling these correlations and exploiting them in the design of group testing algorithms can enhance accuracy and efficiency, leading to more effective group testing.

To capture correlation, much of the existing literature propose models that are context-specific and oftentimes oversimplified. Furthermore, the proposed group testing methods are tied to the underlying correlation models that are studied. 
In the context of disease propagation, the correlation models are built either on the notion of proximity (defined through graphs) \cite{nikpey2022group,arasli2023group,lau2022model,arasli2021graph}  or in conjunction with a disease spread model \cite{zhang2023adaptive,brault2021group, nikolopoulos2023community,ahn2023adaptive}. 
In the context of network fault diagnosis, the correlation is based on physical proximity between the nodes in a network, or frequency of communication between two nodes \cite{xu2023performance}. 

{In this work, we propose a unifying framework for group testing in the presence of correlation, without being restricted by context-specific spread models.} Our objective is to study a general correlation model, i.e., a general joint distribution denoted by $\mathcal{D}$ on the state of the nodes, where $\mathcal{D}(S)$ represents the probability that exactly a subset $S$ of individuals are infected. Our aim is to devise group testing algorithms tailored to this expansive scope without substantial simplification. {Our proposed framework, along with the group testing strategies we propose, can directly apply to any application with an underlying statistical correlation model. 

{The paper is organized as follows. In the remainder of this section, we discuss related work and the contributions of this paper. In Section~\ref{sec:2}, we introduce our proposed model and the mathematical problem formulation. In Section~\ref{sec:3}, we study adaptive group testing and propose a greedy adaptive algorithm that imposes no limitation on the underlying distribution of the hyper-edges. The algorithm is analyzed in Section~\ref{sec:analysis}. In Section~\ref{sec:5}, we apply our algorithm to prior models to recover and improve prior results. Section~\ref{sec:loweradapt} provides examples where the algorithm is order-optimal and others where it is not. In Section~\ref{sec:nonadapt}, we present algorithms that are not adaptive, and in Section~\ref{sec:noisy}, we consider the setting where tests are noisy. Finally, we conclude and discuss open problems in Section~\ref{sec:last}.}



\subsection{Related Work}\label{sec:relwork}

Traditionally, group testing has been examined via two primary paradigms; namely,  combinatorial group testing and probabilistic group testing. Here, we first go through the combinatorial version and then introduce the probabilistic version. In each, we delve into recent research efforts that establish correlation models and harness this correlation to reduce the number of required tests.

 Combinatorial group testing considers scenarios in which, out of a population of size $n$, a maximum of $d$ individuals are infected \cite{du2000combinatorial}. Adaptive group testing, where each test is designed based on the result of the previous tests, has been proved to efficiently identify the infected set using at most $d \log n + O(d)$ tests, complemented by a corresponding lower bound. In the non-adaptive case, where the tests are designed ahead of time, the best known upper bound is $d^2 \log(n/d)$, complemented with an almost matching lower bound of $\Omega(\frac{d^2 \log n }{\log d})$ \cite{erdHos1985families}. {There has also been recent literature on group testing with only a few rounds of adaptation \cite{wang2024quickly,coja2020optimal,scarlett2019efficient}. For example, \cite{wang2024quickly} shows that using $O(\log d)$ rounds of testing,  almost all infected nodes can be recovered using almost $ \frac{1}{\text{capacity}(Z)} d \log n/d$ tests where the tests are noisy and $Z$ is a generic channel corrupting the test outcomes. More recent works have aimed not only to be order-wise optimal but also to determine the exact constant coefficients for the number of tests required \cite{chan2014non, scarlett2016phase, aldridge2016improved,wang2024quickly}.}

 {Another line of work is quantitative group testing which  explores scenarios where the test results are not binary but instead, a test returns the exact number of infections within a group. It is shown that in the adaptive setting, $O(\log n)$ tests are enough \cite{christen1980fibonaccian,aigner1986search} and later the coefficient of $\log n$ was improved \cite{gargano1992improved}. Recently,  \cite{soleymani2024quantitative}   designed low-complexity algorithms with efficient constructions and decoding.  For non-adaptive group testing, \cite{hahn2022near} obtains $O(d \log n)$ which matches the information-theoretic lower bounds up to a factor. \cite{soleymani2024non} has improved the decoding complexity.} 

In a recent study,  \cite{gonen2022group} studies a scenario in which only specific subsets of individuals can be infected, collectively referred to as the set of candidate subsets denoted as $E$. This model can partly capture  correlation between the nodes. \cite{gonen2022group} models candidate subsets as hyperedges of a hypergraph defined on the set of individuals (nodes) $V$ and assumes that each hyperedge has a maximum size of $d$.  It provides non-adaptive and adaptive group testing strategies that require $O(d\log E)$ and $O(\log |E| + d^2)$ tests, respectively. These bounds are complemented with the lower bound  $\log |E| + d$ for a specific ``worst-case" hypergraph.
This work is closest to our work, both in terms of the underlying model and the conceptual framework for algorithm design. However, their framework does not allow different likelihoods of the edges. This assumption is not practical for many applications. For instance, very large and very small edges might be unlikelier
than the medium-sized edges. Or different classes of individuals may have different risk factors. Notably, our approach encompasses not only the definition of hyperedges but also introduces a distribution over them to introduce likelihood, and as such our work stands in the so called probabilistic group testing paradigm.

{
There is a variant of group testing, also known as (hyper)graph learning, which is based on graphs and hypergraphs where items are semi-defective, meaning that a test result is positive if there is a hyperedge that is completely contained in the test \cite{abasi2018error, abasi2019learning, bui2024concomitant}. Several bounds have been established for general graphs and hypergraphs \cite{abasi2018error, abasi2019learning}, and these bounds have been improved by assuming specific structures on the hypergraphs \cite{bui2024concomitant}. Although these works utilize the concepts of graphs and hypergraphs, our problem differs in its fundamental nature. }

{There have been other works where specific classes of  graphs or hypergraphs are assumed in the design of group  testing. For instance,  \cite{bui2021improved} and \cite{bui2023non}  assume that the graph is a line and a block of consecutive nodes is positive. Another example is \cite{nikpey2022group} where theoretical bounds are provided  on the number of tests when the graph is a tree or grid. In this work, we do not impose any constraints on the structure of the hypergraph.}

Probabilistic group testing was studied in \cite{li2014group}, where each individual $i$ is infected with a probability $p_i$, independently of others. \cite{li2014group} establishes  near-optimal bounds for both adaptive and non-adaptive group testing, quantified as $O(H(X) + \mu)$, where $\mu = \sum_i p_i$ represents the average number of infections, and $H(X) = \sum_i -p_i \log p_i$ denotes the entropy. 
In recent years, there has been a surge of interest in modeling correlations in group testing scenarios.  \cite{nikpey2022group,arasli2023group}  model correlation using edge-faulty graphs where a simple graph $G$ is considered as the underlying contact graph and each edge is dropped with some probability, forming random components. The nodes within a component are then assumed to be in the same state. {In \cite{nikpey2022group}, it was shown for several families of graphs that exploiting correlation can lead to significant improvement in the required number of tests (compared to independent group testing). \cite{arasli2023group} shows improvement when the realized graphs have a certain structure, i.e. they are nested.} Another line of work is using communities to model  correlation \cite{nikolopoulos2021group2,ahn2021adaptive, nikolopoulos2023community,nikolopoulos2021group,jain2024sparsity}. For example, in \cite{nikolopoulos2021group2}, it is assumed that each community is infected with some probability, and if a community is infected, all of its members are infected with some non-zero probability independent of the others. If a community is not infected, none of its members are infected. {For this model, the authors designed efficient adaptive and non-adaptive algorithms that significantly reduce the number of tests by accounting for the structure and pooling the infected communities. They also developed more reliable tests using the structure when the tests are noisy.  \cite{ahn2021adaptive},  considers similar communities, but the communities become infected by random initial seeds. The authors design an adaptive algorithm that significantly reduces the number of tests, further supporting the benefits of exploiting correlation. 
}

\subsection{Contributions}
The contributions of this work are as follows:
\begin{itemize}
\item In Section~\ref{sec:2}, we model arbitrary statistical correlation by employing hypergraphs, drawing inspiration from \cite{gonen2022group}. A hypergraph $G=(V,E)$ has node set $V$, $|V|=n$, and edge set $E$, where each hyperedge $e\in E$ is a subset of nodes: $e \subseteq V$. Enhancing the combinatorial model in \cite{gonen2022group}, we consider a probability mass function over the hyperedges. This captures arbitrary statistical correlation among infection of different nodes. We design group testing algorithms that return the set of infected nodes  with a probability that converges to $1$ as $n$ increases to infinity for each of the scenarios described below.

\item In Section~\ref{sec:3}, we propose an adaptive group testing algorithm that is capable of exploiting correlation by the means of updating the posterior distribution on the hyperedges given the previous test results. The algorithm works in two stages. Stage 1 focuses on conducting  ``informative'' tests, which significantly narrows down the search space. When these tests are no longer feasible, the algorithm transitions to a second stage where it tests the remaining uncertain nodes individually. We demonstrate that this algorithm successfully identifies the infected set and requires an expected number of tests that can be upper bounded as a function of entropy $H(X)$ and the expected number of infections $\mu$ in Theorem~\ref{theorem:final}. We also show  that when a stochastic upper bound $u$ on the number of infected nodes is known (ie, the probability that the number of infected nodes is upper bounded by $u$ converges to $1$ as the number of nodes increases to infinity),  then  $O(H(X)+u)$ tests are sufficient in expectation in Corollary~\ref{cor:concentr}. Thus, the guarantee on the expected number of tests becomes better as the upper bound becomes tighter. If the number of infections is concentrated around $\mu$, the average number of infections, then $O(\mu)$ constitutes the stochastic upper bound and hence $O(H(X) + \mu)$ tests are sufficient.  
In Section~\ref{sec:5}, we apply our algorithm to several existing group testing frameworks and compare our results against these prior approaches.

\item  {In Section~\ref{sec:loweradapt}, we provide various examples to illustrate the extent of optimality of the algorithm. $H(X)$ is a known lower bound on the number of tests. We show that our proposed algorithm, and the upper bound we establish on the number of tests, ${O}(H(X)+\mu)$, is worst-case order-optimal.  We further show that (i) $H(X)$ can be a loose lower bound with a significant gap to the number of tests needed ; (ii) for some families of graphs, $\mu$ is a lower bound on the number of tests; and (iii) the lower bound cannot be described using only $\mu$ and $H(X)$. In particular, we provide examples where the tight lower bound is greater than the entropy but less than $\mu$.  We subsequently consider random $d$-regular hypergraphs, where each hyperedge of size $d$ is present with probability $r$. In Section~\ref{sec:randomgraph}, we show that the algorithm is order optimal for the expected number of tests in dense $d$-regular random  hypergraphs (ie, when $r$ is large), and a modification thereof attains order optimality in sparse $d$-regular random hypergraphs (ie, when $r$ is small);  some of the proofs  use mild additional conditions that arise naturally}. 

\item In Section~\ref{sec:2}, we introduce a semi-non-adaptive group testing model, a model in between adaptive and non-adaptive group testing, in which tests are performed as per a predetermined sequence, which is designed before result of any test is known, but the tests can be terminated any time depending on the results of the already executed tests. 
In Section~\ref{sec:nonadapt}, we provide an 
algorithm that requires ${O}(u H(X)+u \log n)$ tests in this model, where $u$ is a stochastic upper bound on the number of infected nodes as defined in the second bullet (Theorem~\ref{theorem:NATAS}). The guarantee is clearly worse than that we obtain in a fully adaptive setting. But, this algorithm is easier to execute in that the testing sequence is preplanned and therefore tests can be scheduled in advance. We also highlight challenges in designing non-adaptive algorithms. Specifically, we show that, unlike classic independent group testing, there can be a significant gap in the number of tests needed in adaptive vs non-adaptive testing.

\item {In Section~\ref{sec:noisy}, we extend our methods and results to noisy group testing (where test results are corrupted by a symmetric noise). The impact of such noise can clearly be mitigated by repeating each test a certain number of times and going with the majority verdict, such repetition  however increases the overall number of tests. In Theorem~\ref{theorem:noisyadapt}, we show that in the adaptive scenario the number of repetitions can be substantially reduced by considering the impact of the error in update of the posterior probabilities of edges after each test. The guarantees we obtain on the expected number of overall tests become better as the probability of error reduces. }

\end{itemize}


\section{Problem Formulation}\label{sec:2}
\subsection{Generalized Group Testing} \label{sec:model}
We model a generalized group testing problem via hypergraphs. Consider the hypergraph $G=(V,E)$ where $V$ is the set of nodes and $E$ the set of hyperedges. Each hyperedge  $e\in E$ is a subset of the nodes, $e\subseteq V$.  We assume $|V| = n$. In a population of size n, each individual is represented by a node in this hypergraph and the hyperedges capture possible dependencies between the nodes' states. 
A distribution $\mathcal{D}$ is assumed over the edges in $E$, and \textbf{one} edge $e \sim \mathcal{D}$ is sampled from this distribution to be infected. We use the terminologies ``sampled" edge, ``infected" edge, and ``target" edge interchangeably. The probability that hyperedge $e$ is infected is denoted by $p(e)$. When a hyperedge is infected,  all nodes $v \in e$ are infected and all other nodes $u \notin e$ are not infected. 
The goal is to perform group tests, with the minimum number of tests, until the infected edge is identified accurately.

Figure~\ref{fig:hypgraph} show an example of a hypergraph $G = (V,E)$ where $V = \{v_1, v_2, v_3, v_4, v_5 \}$ is the set of nodes, $E = \{ e_1 = \{v_1,v_2,v_3\}, e_2 = \{v_1,v_5\}, e_3 = \{v_4,v_5\} \}$ is the set of edges, and $p_{e_1} = 0.3$, $p_{e_2} = 0.2$, $p_{e_3} = 0.5$ is the distribution over the edges.
Note that $v_2$ and $v_3$ have a complete correlation and are always in the same state because each edge contains either both of them or none of them. Also, $v_1$ and $v_4$ are also completely correlated as they are always in opposite states because each edge contains exactly one of them.

To formulate the problem more precisely, let $X = (X_1, X_2, \ldots, X_n)$ be the vector of the nodes' states where $X_i = 1$ if the $i$'th node is infected  and $X_i = 0$ otherwise. We denote the probability that node $i$ is infected  by $p_i$. When an edge $e$ is sampled, we have $\forall v \in e : X_v = 1$ and $\forall v \notin e : X_v = 0$. Within this model, each node $v$ may belong to multiple edges and therefore, we have $$p_v=\mathbb{E}[X_v]=\sum_{e\in E:\ e\ni v} p(e).$$

\noindent The expected number of infections, $\mu$, is thus given by
\begin{equation}
\label{eq:mu}
\mu = \mathbb{E}[X_1 + X_2 + \cdots + X_n] = \sum_{v\in V} p_v.
\end{equation}
 It is worthwhile to mention that in prior work such as \cite{li2014group}, it is often assumed that $\mu \ll n$, based on the observation that otherwise the number of tests needed is $\Omega(n)$. Interestingly, this is not true in general when we have correlation, as captured in our model, and an example is given in Section~\ref{sec:loosemu}.

We next define testing and recovery before formalizing our objective.

\paragraph{Testing.} 
A group test $T$ is defined by the subset of nodes $T\subseteq V$ that take part in that test. We denote the $i$'th test by $T_i$,
 $T_i \subseteq V$ and its result by $r_i$. {When $r_i = 1$, we say that the test is positive, and when $r_i = 0$, we say that the test is negative.}

{\textbf{Adaptive/non-adaptive/semi-non-adaptive testing.} Testing can be adaptive or non-adaptive: In adaptive testing, each test $T_i$, $i=1,2,\ldots,L$,  is a function of the prior tests $T_1,\ldots,T_{i-1}$ as well as their respective results $r_1,\ldots,r_{i-1}$. {Note that the number of tests, $L$, is a random variable that can depend on the randomness in the algorithm, the randomly chosen target edge, as well as the randomness in the test results (if testing is adaptive and noisy).}}  In non-adaptive testing, all the tests $T_1, T_2, \ldots, T_L$ are designed a-priorily and can be done in parallel. In other words, the design of the $T_i$'s are independent of their results. 

{We further introduce a semi-non-adaptive group testing model ({SNAGT}) that serves as an intermediate approach between non-adaptive and adaptive group testing. This model is relevant because it combines the advantages of both strategies: the design of individual tests is non-adaptive, meaning each test is planned without relying on the outcomes of previous tests. However, it introduces adaptiveness in the decision to stop testing, allowing the total number of tests, L, to be determined based on prior results. From a practical perspective, the utility of any SNAGT testing strategy is that since who to test at any point in the testing sequence can be determined apriori (ie before any result is known), the individuals who must be tested at any point can be informed well in advance, actually at the start of the overall testing process. This renders the testing process as convenient for the subjects. Cancelation of tests can however be on the fly, as results become available.  
In Section~\ref{sec:algrep}, we present an algorithm for the adaptive testing. In Section~\ref{sec:nonadapt}, we examine non-adaptive and semi-non-adaptive algorithms, highlighting the challenges associated with each. 
}



\textbf{Noisy Testing.} Test results can be noiseless or noisy. In the noiseless setting, we have $r_i = 1$ (positive test result) iff there is at least one node in $T_i$ that is infected (i.e., belongs to the sampled edge) and we have $r_i = 0$ otherwise. In the noisy setting,  results are flipped with probability $\delta$, i.e. when $T_i$ contains a node from the target edge, $r_i = 1$ with probability $1 - \delta$ and $r_i = 0$ with probability $\delta$. Similarly, when $T_i$ does not contain any node from the target edge, $r_i = 0$ with probability $1 - \delta$ and $r_i = 1$ with probability $\delta$. In other words, we assume that the error is symmetric. {The symmetric noisy model is studied in the literature for independent group testing in works such as \cite{scarlett2018noisy,scarlett2018near,chan2011non,malyutov1998jaynes,cai2013grotesque}. In Section~\ref{sec:noisy}, we study this noisy testing for group testing with correlation.



\paragraph{Recovery.}  Given the test sequence  $T_1,\ldots,T_L$ and the respective results $r_1,\ldots, r_L$, the nodes' states are estimated as $\hat{X} = (\hat{X}_1, \hat{X}_2, \ldots, \hat{X}_n)$ where $\hat{X}_i$ is the estimated state of node $i$, $i=1,\ldots,n$.
The probability of error is then defined as $P_e = P(X \neq \hat{X})$, where the randomness is over $\mathcal{D}$ (probability distribution over the edges) and possible randomness of the testing design. 

\paragraph{Objective.} In this work, we aim to devise algorithms that, in expectation, use the least number of tests (minimize $\mathbb{E}[L]$) and recover the sampled edge with probability of error $P_e$. {The error probability $P_e$ can be zero, a value that goes to zero as $n$ goes to infinity, or a fixed goal error.}


{
Note that the above model is equivalent to arbitrary correlation among nodes. Consider an arbitrary distribution $\mathcal{P}(S)$, $S\subset V=\{1,2,\ldots,n\}$, defined over the power set of a population size $n$, encoding the probability that exactly subset $S$ of the population is infected. Then, there is a corresponding hypergraph $G$ where an edge $e = S \subseteq V$ being the target edge is equivalent to set $S$ being infected, and $\mathcal{D}(e) = \mathcal{P}(S)$. It is important to note that the set of hyperedges can also be a proper subset of the power set, depending on the distribution $\mathcal{D}$, so we have $|\text{support}(\mathcal{D})| = |E|$.}

\subsection{{Prior Models as Special Cases}} \label{sec:priormodel}
We next discuss how different statistical models for infection of nodes in group testing can be considered as a special case of the above model. In particular, we find the probability mass function $\mathcal{D}(S)$ for all possible hyperedges $\mathcal{S}\subseteq V$.

First consider independent group testing \cite{li2014group}. In this model, we have $\mathcal{D}(X) = \Pi_i [X_i p_{v_i} + (1-X_i)(1-p_{v_i})]$. 
Going beyond independent group testing, {we  consider models where correlation is modeled by edge faulty graphs \cite{arasli2023group,nikpey2022group}.}  
\hesam{ \cite{nikpey2022group}} considers a \emph{simple} graph $H = (V_H,E_H)$ where each (simple) edge is sampled in $H$ randomly {with the same probability} and the nodes that end up in the same component are assumed to have the same state, independently of the nodes in other components. Translating to the hypergraph model, for a binary vector $X=(X_1,\ldots,X_n)$, let $S = \{i \mid X_i = 1\}$. Then $\mathcal{D}(S)$ is the probability that the nodes in $S$ (and only those) are infected in $H$. 
 {One can compute $\mathcal{D}(S)$ by taking a sum over the probability of graphs in which $S$ is disconnected from the rest of the graph, all the components formed in $S$ are infected and all the components formed in $V \setminus S$ are not infected. This gives us the equivalent hypergraph model for \cite{nikpey2022group}. For more details, please refer to Appendix~\ref{app:hypmodel}.}
In a related work in \cite{arasli2023group}, a similar model is considered, but with the restriction that only one component can be infected, and the infection probabilities across edges may vary. {Again, we can build the equivalent hypergraph model by computing the probability mass function $\mathcal{D}(S)$ for every hyperedge $S$. This is done by taking the sum over probabilities of graphs where $S$ is a connected component after dropping the edges and is the only infected component. For more details, please refer to Appendix~\ref{app:hypmodel}.}

In \cite{ahn2023adaptive}, the authors  consider $m$ communities each of size $k$.  Initially, each node is a ``seed'' with probability $q$. Then the seeds infect nodes in the same community with probability $q_1$ and outside of the community with probability $q_2 < q_1$. The  probability that exactly one set $S$ is infected is the probability that no one is seed in $V \setminus S$ (where $V$ is the union of all communities) and set $S$ becomes infected with $i, 1\leq i \leq |S|$ seeds.  
{For more details, please refer to Appendix~\ref{app:hypmodel}.}

In \cite{nikolopoulos2021group2}, the authors consider $F$ families where each family is infected with probability $q$ independently. Then each node in an infected family $F_j$ is infected with probability $p_j$ independent of the rest, and the non-infected families have no infected nodes. The probability that a subset $S$ of the individuals are infected is that first, their family is infected (with probability $q$) and second, if $i \in S$ is in $j$'th group, $i$ is infected (with probability $p_j$). {Now if $S = \{v_1,\ldots, v_k\}$, and $I_j$ is the number of nodes of $F_j$ in set $S$, then the probability that only set $S$ becomes infected is \[
\mathcal{D}(S) =  \prod_{\substack{j : I_j > 0}} q(1-p_j)^{|F_j| - I_j} (p_j)^{I_j} \cdot \prod_{\substack{j : I_j = 0}} (1-q + q(1-p_j)^{F_j})
.\]}

\begin{figure}
    \centering
    \input{graph}
    \caption{A hypergraph with $V = \{v_1, v_2, v_3, v_4, v_5 \}$ and $E = \{ \{1,2,3\}, \{1,5\}, \{4,5\} \}$ where $p_{\{1,2,3\}} = 0.3$, $p_{\{1,5\}} = 0.2$, and $p_{\{4,5\}} = 0.5$.}
    \label{fig:hypgraph}
\end{figure}

\subsection{Preliminaries}
In this section, we provide the tools and basic ideas that we will use to describe and analyze our proposed group testing algorithm. 
We first establish a  lower bound on the number of tests needed. Our result is inspired by \cite{li2014group} where  almost matching lower  and upper bounds are proved for  adaptive and non-adaptive independent group testing. Following a similar argument, we prove a lower bound for the hypergraph problem.  Specifically, in \cite{li2014group} they proved the following result:
\begin{theorem}
\label{thm:LCH}
    \cite{li2014group} For the case where nodes are independent, i.e. $\mathcal{D}(X) = \Pi_i [X_i p_{v_i} + (1-X_i)(1-p_{v_i})]$, for any algorithm that recovers the infection set and performs $L$ tests with probability $1-\epsilon$ we have
    $$L \geq (1-\epsilon)H(X).$$
    On the other hand, there is an adaptive algorithm that finds the infected set with a probability that goes to 1 as $n\rightarrow\infty$ using $L \leq O(\mu + H(X))$ tests \textcolor{black}{where $H(X)=\sum_{v\in V}-p_v\log p_v$ and $\mu$ is given by \eqref{eq:mu}.}
\end{theorem}

The proof provided for the $(1- \epsilon) H(X)$ lower bound is quite general and can be extended to our work as follows:

\begin{theorem}\label{theorem:lowerbound}
    For any algorithm that recovers the sampled edge $e^*$ with probability at least $1-\epsilon$ using $L$ tests, we have
    \begin{equation}\label{eq:lowerboundentropy}
        L \geq (1-\epsilon)H(X)
    \end{equation}
    where $H(X) = \sum_{e \in E} - p(e) \log p(e)$.
\end{theorem}

\begin{remark}
In contrast to the result of Theorem \ref{thm:LCH} where $H(X)$ characterizes the minimum number of tests (up to an additive factor), $H(X)$ is not necessarily tight in our setting. We discuss this more later in Example~\ref{app:examplehighprob} where $H(X) = O(\log n)$ but $\Omega(n)$ tests are needed. 
\end{remark}

\section{{Adaptive Testing: An Algorithmic Approach}} \label{sec:3}
\subsection{Challenges}

Most classic adaptive algorithms for independent group testing are, in essence, built on generalized binary-splitting which greedily chooses the test that most evenly splits the candidate nodes or equivalently chooses the test with the maximal information gain. In probabilistic group testing (with independent nodes) \cite{li2014group}, this involves choosing subsets of nodes to test such that the probability of a subset containing an infected node is close to $1/2$ (meaning that the test provides close to one bit of information). A negative test result allows a subset to be ruled out. Candidate subsets that test positive are, however,  partitioned into two subsets (with roughly equal probability mass) for further testing. With such a design, the posterior probability of the sets testing positive, given the sequence of previous test results, is close to~$1/2$.
To follow this paradigm, we face three challenges.

First and foremost, it is important to note that treating each node individually and disregarding the underlying correlation is not efficient in terms of the number of tests needed and we aim to propose group testing strategies that exploit the correlation. 

 Under this model, even verifying whether a specific edge $e$ is the target edge or not requires many tests. As a matter of fact, it may need a number of tests that are of the order of its size. For example, consider the hypergraph in Figure~\ref{fig:graphn} with $k = 4$ nodes 
and consider $e_1$ as the suspected edge set with size $k-1$. To verify whether $e_1$ is the target edge, we cannot treat it as a whole and conclusively determine if $e_1$ is the actual edge or not based on a single test: If the test is positive, it can be $e_1$ that is realized or any other edge. In fact, any test of size greater than one is positive. Indeed, it takes $\Omega(k)$ tests to determine if $e_1$ is the target edge. This is because if $e_1$ is not the target edge, only a single node is not infected, and detecting the negative node will take $\Omega(k)$ individual tests. 

\begin{figure}
    \centering
    \input{graph_n-1}
    \caption{A graph with 4 nodes and 4 edges. Each edge contains 3 nodes, for example, $e_2 = \{v_1, v_3, v_4\}$}
    \label{fig:graphn}
\end{figure}

Now consider the classical idea of greedily testing subsets and ruling out those subsets that test negative. If one wishes to utilize the correlation that is inherent in our model, it is not straightforward how the group tests should be designed: (i) Suppose that you have chosen a subset $S$ of nodes to test. If this test is negative, not only all nodes in $S$ are negative but also all those edges that contain a node in $S$ are not the target infected edge and could be ruled out. (ii) If the test is positive, we cannot conclude that the nodes that belong to $V\backslash S$ are negative. As a matter of fact, all can still remain valid candidates if they share an edge with an element of $S$. In other words, we can not conclude that the target infected edge is a subset of the  nodes of set $S$.   The design of the tests is thus nontrivial.

In the following, we demonstrate in 
Lemma~\ref{lemma: update}  how a test can rule out a set of edges. Using this lemma, and updating the posterior probabilities of infections after each test, we design a greedy algorithm to sequentially rule out negative edge sets. 

\subsection{Some Useful Definitions}\label{sec:newdefs}

In order to design the tests sequentially, we utilize the posterior probability of nodes and edges of the hypergraph being infected given the previous tests and their results. In particular, suppose tests $T_1, T_2, \ldots, T_{k}$ are previously performed and the results are $r_1, r_2,\ldots, r_k$. We show the posterior probability of edge $e$ with $$q_{e\mid \{(T_1, r_1),(T_2, r_2), \ldots, (T_k, r_k) \}}.$$ Note that $q_{e|\{\}} = p(e)$. Posterior probability of $v \in V$ being infected is similarly defined by
\begin{align}
\label{eq:qvnotation}
q_{v\mid\{(T_1, r_1),(T_2, r_2), \ldots, (T_k, r_k) \}} = \sum_{e\in E:\ e \ni v} q_{e\mid\{(T_1, r_1),(T_2, r_2), \ldots, (T_k, r_k) \}}.
\end{align}
When it is clear from the context, we drop $\{(T_1, r_1),(T_2, r_2), \ldots, (T_k, r_k) \}$ and use $q_e$ and $q_v$.

To connect edge probabilities to the probability that a test becomes positive, we need the following definition.

\begin{definition}
    For a subset $S \subseteq V$, the edge set of $S$ is defined as $E(S) = \{e \mid e \in E \wedge \forall v \in e: v \in S \}$. The weight of the set $S$ after tests $\{(T_1, r_1),(T_2, r_2), \ldots, (T_k, r_k) \}$ is defined as
    \begin{equation}
        w(S| \{(T_1, r_1),(T_2, r_2), \ldots, (T_k, r_k) \}) = \sum_{e \in E(S)} q_{e| \{(T_1, r_1),(T_2, r_2), \ldots, (T_k, r_k) \}}.\label{eq:w}
    \end{equation}
    When clear from the context, we drop $\{(T_1, r_1),(T_2, r_2), \ldots, (T_k, r_k) \}$ and simply write $w(S)$.
\end{definition}

After performing $k$ tests and seeing the results, the expected number of infections would be updated according to the posterior probabilities discussed above. We denote the expected number of infected nodes in the posterior probability space by:
\begin{equation}
\mu_{\{(T_1, r_1),(T_2, r_2), \ldots, (T_k, r_K) \}} = \sum_v q_{v \mid {\{(T_1, r_1),(T_2, r_2), \ldots, (T_k, r_k) \}} }
\end{equation}

\noindent We drop ${\{(T_1, r_1),(T_2, r_2), \ldots, (T_k, r_k) \}}$ when it is clear from the context, and show it with $\tilde{\mu}$ to avoid confusion with the original $\mu$ before performing any test.

 Throughout this work, when a new test $(T_i, r_i)$ is done, we compute $q_{e | (T_1,r_1),\ldots, (T_i,r_i)}$ from  the last posterior $q_{e | (T_1,r_1),\ldots, (T_{i-1},r_{i-1})}$. In other words, we initially start with $\mathcal{D}_0 = \mathcal{D}$ and after observing the results of the $i$'th test, compute $\mathcal{D}_i$  from $\mathcal{D}_{i-1}$ based on  $(T_i, r_i)$ where $\mathcal{D}_i=\{q_{e | (T_1,r_1),\ldots, (T_i,r_i)}\}_{e\in E }$. In Lemma~\ref{lemma:poster}, we prove that computing the posterior probability entails ``removing some edges'' from the hypergraph and ``scaling the distribution'' according to the following definition:

\begin{definition}\label{def:update}
    Let $\mathcal{D}$ be a distribution over a set $E$. Removing $e \in E$ from $\mathcal{D}$ entails setting $\mathcal{D}(e) = 0$ and re-normalizing the distribution  by a factor $c > 1$, $\mathcal{D}(e) \leftarrow c\mathcal{D}(e)$, such that $\sum_{e \in E} \mathcal{D}(e) = 1$.
\end{definition}

In the graph in Figure~\ref{fig:hypgraph}, we can see that $p_{v_1} = 0.5$, $p_{v_2} = p_{v_3} = 0.3$, $p_{v_4} = 0.5$, and $p_{v_5} = 0.7$, which are proabilities of nodes being infected and $q_{v_i | \{ \}} = p_{v_i}$. If $S = \{v_1, v_2, v_3, v_5\}$, then $E(S) = \{e_1, e_2 \}$ and $w(S) = p_{e_1} + p_{e_2} = 0.5$.
The average number of infections is $\mu = 0.5+0.3+0.3+0.5+0.7 = 2.3$. If a test $\{v_2, v_4\}$ returns positive, $q_{e_2} = 0$ as it does not share a node with $\{v_2, v_4\}$, and the posterior is re-normalized by the scalar $0.8$ according to Definition~\ref{def:update} so that $q_{e_1} = p_{e_1}/.8 = .375, q_{e_3} = p_{e_3}/.8 = .625$  sum up to 1. Now $q_{v_1  | \{(S,1)\}} = 0.375$, $q_{v_2  | \{(S,1)\}} = .375$, $q_{v_3 | \{(S,1)\}} = 0.375$, $q_{v_4 | \{(S,1) \}} = 0.625$, and $p_{v_5 | \{(S,1)\}} = 0.625$

We further elucidate edges and relation between weights of a set of nodes, probabilities that the nodes in it are infected, probabilities that edges in it are target edges using an additional example. We illustrate the relations that are contrary to common intuition:
\begin{enumerate}
\item An edge $e_1$ can be a proper subset of another edge, $e_2$, but the probability that $e_1$ is target edge can be considerably higher than the probability that $e_2$ is a target edge.
\item Weight  of a set $S$ of nodes is not equal to the sum of the probabilities that the nodes in it are infected, the former can be considerably lower than the latter, the former can be $0$ even when it contains several nodes which are infected with probability close to $1$.
\end{enumerate}
\begin{Example}\label{ex:islands}
Consider a network with $k$ islands where each island has $m$ nodes, so $n = mk$. Nodes on the same island are perfectly correlated, ie, they always have the same state, as such every hyper edge always contains all or none of any given island (Figure~\ref{fig:island}, $k=4, m = 3$). Thus, there are $2^k$ hyperedges at most; note that in this case the number of hyper edges can be significantly smaller than the size of the power set of nodes, $2^{mk}$. Let nodes in different islands be independent. Consider a setting with three types of islands: low, moderate and high, characterized by their  probability of infection. For all nodes in high islands, $p_i =  1-\gamma$ (e.g. island 4 in Figure \ref{fig:island}), where $\gamma$ is a small positive number.  In moderate islands (e.g. islands 2,3), $p_i \approx \frac{1}{2}$ and in low islands, e.g. island 1, $p_i = \epsilon$, where $\epsilon$ is a small, positive number.

Nodes in an edge can be a proper subset of nodes in another edge, eg, one edge may comprise of $2$ islands, another may comprise of only one of those $2$ islands. Probability that an edge is a target edge can be lower than that for another edge which is its proper subset. For example in Figure \ref{fig:island}, consider an edge that is island $4$, and another that is islands $1, 4$. The first is a proper subset of the second. Probability that the 1) first is a target edge is $p_4 (1-p_1)(1-p_2)(1-p_3)= (1-\gamma) (1- \epsilon) 1/4 \approx 1/4$ 2) second is a target edge is $p_4p_1 (1-p_2)(1-p_3) = (1-\gamma)\epsilon 1/4 \approx 0.$ The difference arises because $1-p_1 \approx 1,$ while $p_1 \approx 0.$

Next, weight of a set $S$ of nodes is not equal to the sum of the probabilities that the nodes in it are infected. Simplest way to see it is to consider $w(S)$ for this example when $S = V$. Since $V$ consists of all the edges, $w(S)$ equals the sum of the probabilities that each edge is a target edge which is $1$ since exactly $1$ edge is a target edge. But note that each node in island $4$ is infected with probability close to $1$. Thus, sum of the probabilities that a node is infected over all nodes in $V$ is at least a number close to $3$ (it is close to $6$ in reality given that there are $2$ moderate islands). Weight of a set $S$ can be $0$ even when it contains nodes whose probability of infection is close to $1$. For example consider a $S$ that consists of $2$ nodes in the high island $4$, it consists of $2$ nodes which are infected with probability close to $1$, but it does not consider a single edge because an edge has either all or no node in any given island. Thus $w(S) = 0.$
\end{Example}

\begin{figure}
 \centering
 \includegraphics[scale = .14]{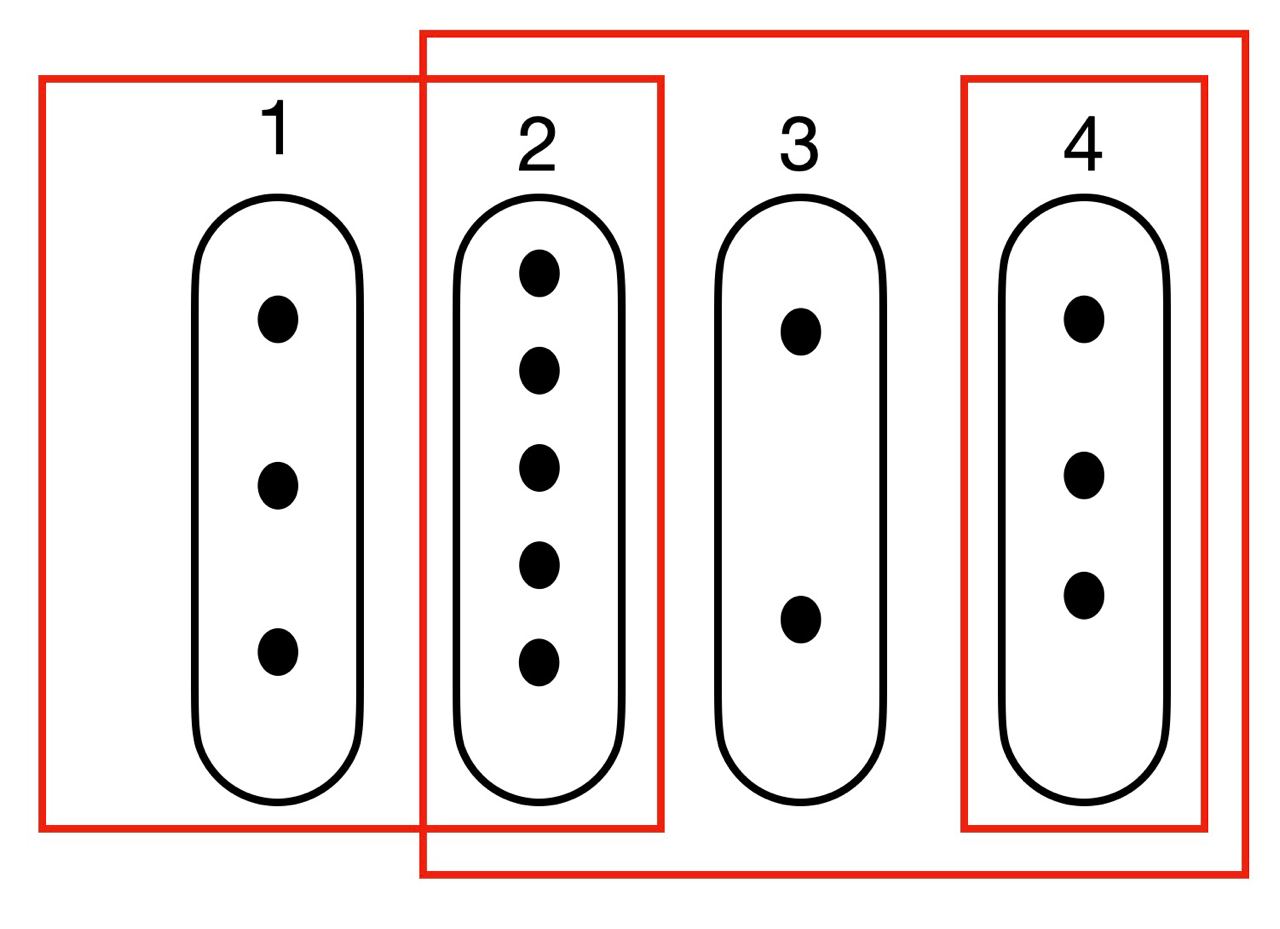}
 \caption{An example with $k = 4$ islands. Each edge contains either all nodes in an island or none of it. There are therefore $2^4 - 1 = 15$ edges. Three edges are shown  in red. Nodes in an edge can be a proper subset of nodes in another edge, eg, one of the edges marked in red consists of islands $2, 3, 4$, another consists only island $4.$ Nodes in island 1 are infected with $p_1 = \epsilon \approx 0$, nodes in islands 2 and 3 are infected with probability $p_2 = p_3 \approx 1/2$, and nodes in island 4 are infected with high probability $p_4 = 1-\gamma \approx 1$.}  \label{fig:island} \end{figure}

Despite the counterintuitive properties above, the following intuitive property holds: $w(V \setminus \{v\}) = 1 - p_v.$ This can be argued as follows. Note that when  a node $v$ is removed from $S$, for any $S \subseteq V$, $w(S)$ reduces by the sum of probabilities of the edges  containing $v$ which were contained in $S$. This sum equals the sum of probabilities of all edges containing $v$ when $S=V$ as then all edges are contained in $S.$ And, the sum of probabilities of all edges containing $v$ equals the probability that the target edge contains $v$ (since only one edge can be the target edge), which in turn equals the probability that $v$ is infected as $v$ is infected if and only if the target edge contains it.

\subsection{Overview and Ideas} \label{sec:overview}

\begin{figure}
    \centering
    \includegraphics[scale = .15]{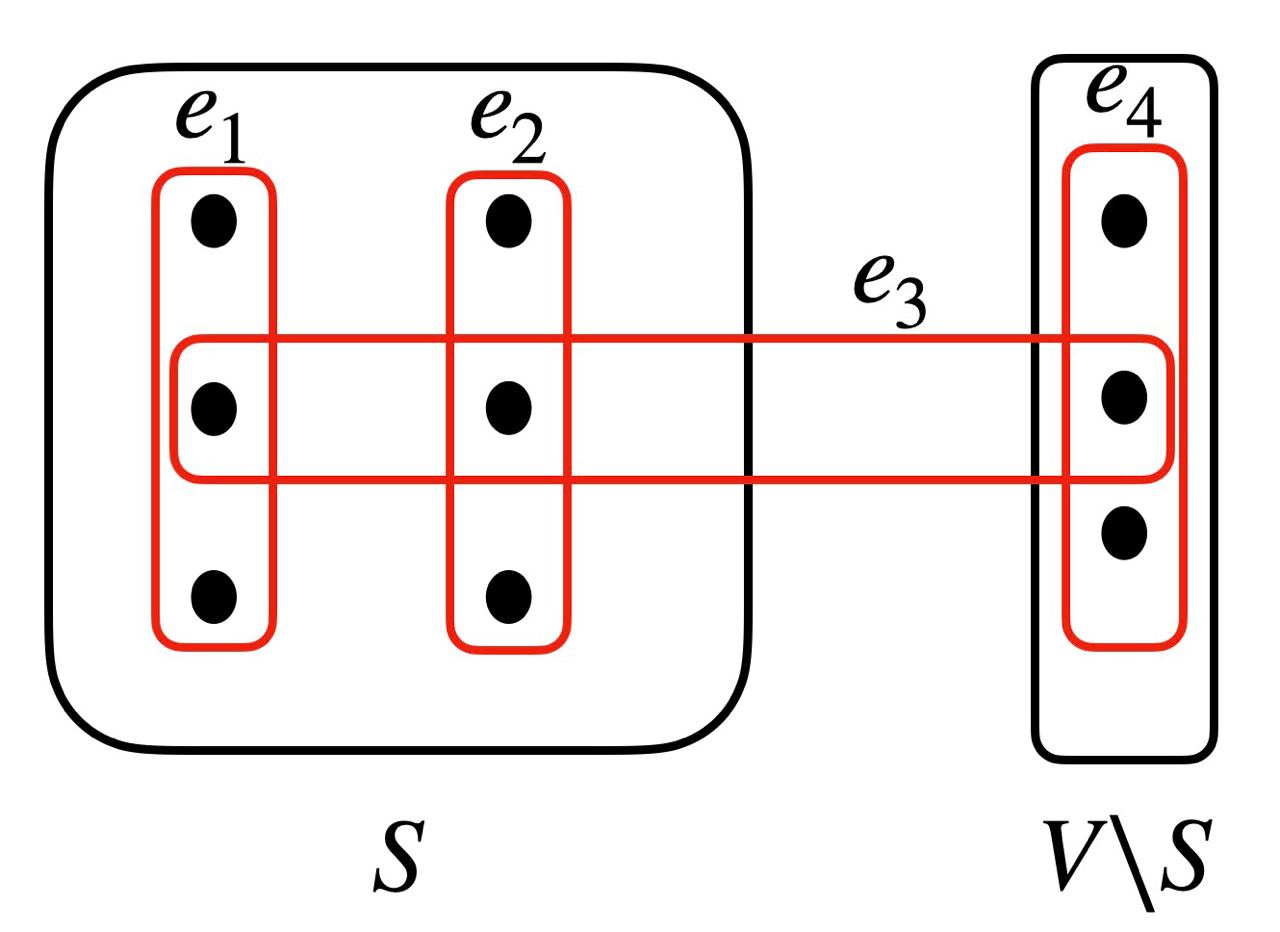}
    \caption{An illustration of $E(S)$. Here, $E(S) = \{ e_1,  e_2\}$ but $e_3 \notin E(S)$ as one of its endpoints is outside of $S$. Hence, $w(S) = p_{e_1} + p_{e_2}. $ Test $V \setminus S$ is positive iff $e^* = e_3$ or $e^* = e_4$.}
    \label{fig:E(S)}
\end{figure}

The ideas behind our algorithm can be summarized in three points: 
\begin{itemize}
    \item Devise a testing mechanism in which by testing a subset of the nodes, we can rule out edge sets that are not compatible with the result.
    \item Design the test so that, independent of the test result, a constant fraction of the mass is ruled out.
    Here, by a constant, we mean a mass that is independent of all the problem's parameters such $G$ and $\mathcal{D}$. We argue that a most informative test can provide at most one bit of information by ruling out edge sets that constitute half of the probability mass. By removing a constant fraction of the mass, we make sure that our tests are near-optimal.
    \item Updating posterior probabilities of infections for the nodes and edges based on the prior test results. It turns out that this is equivalent to removing edges that are ruled out from the hypergraph, and scaling the remaining edge probabilities uniformly up so they sum up to 1. We perform such posterior updates sequentially after each group test.  If at every step we can do this until we reach a single edge (with probability 1), then we would have done $O(H(X))$ tests in expectation (Theorem~\ref{theorem:edgesize}) which is optimal by the matching lower bound in Theorem~\ref{theorem:lowerbound}. These types of tests are ``informative'', as they shrink the search space significantly.
\end{itemize}
If we can not find such a test, we observe that some nodes have a high probability of infection and other nodes are likely to be negative. (Proposition~\ref{prop:nodeweights}). By testing the ``unlikely'' nodes, we either rule them all out if the test is negative and are left with nodes that have a high probability of infection, or will have an informative test if the test is positive as we can rule out a high mass of the edges.

Suppose a set $S$ with $w(S) = 1/2$ is given and we test $V\setminus S$. If the test is positive, it means that the sampled edge $e^*$ is not in $E(S)$ which contains half the probability mass. If the test is negative, then $e^*$ is in $E(S)$ which rules out $E \setminus E(S)$. So in either case, we gain one bit of information. \textcolor{black}{Figure~\ref{fig:E(S)} illustrates $S$, $E(S)$ and $V\backslash S$.} Since finding a set that has $w(S) = 1/2$ might not be possible,  we relax this condition. We aim to find a set $S$ so that $w(S)$ is between two constants $c = 1/2 - \delta$ and $1 - c = 1/2 + \delta$, let's say between $0.05$ and $0.95$.
Then by testing $V \setminus S$, we can understand if the sampled edge is in $E(S)$ or not. Similar to the case of $c=1/2$, either $E(S)$ or $ W \setminus E(S)$ will be ruled out, which has a mass of at least $0.05$.
There is a nuance in the definition of $E(S)$, which contains all the edges that are \textbf{completely} inside of $S$. If we had included the edges that have at least one node in $S$ as $E(S)$, then we couldn't have ruled out the edges as we did.

We next describe how to find  $S \in V$ with $.05\leq w(S)\leq .95$. We design a greedy-like algorithm that finds $S$ iteratively by removing the nodes from the working set. The algorithm starts with $S = V$. Initially, $w(S) = 1$. Then the algorithm drops the nodes in $S$ one by one
until it finds $0.05 \leq w(S) \leq 0.95$ (at which point the algorithm performs a test on $V \setminus S$ and based on the result updates the posterior distribution).  

What happens if such an $S$ does not exist? In this case weight of every subset of nodes is either more than $0.95$ or less than $0.05$, and removing any node from any subset whose weight exceeds $0.95$ reduces the weight of the subset to below  $0.05$. This can for example happen when each node whose status is yet to be known definitively either is infected with a very high probability (exceeding $0.9$) or with a very low probability (a set of all such nodes will test positive with probability at most $0.05$). What we propose in this stage is the following: in a single test, the algorithm can test all the low-probability nodes, and if it is positive, we rule out $0.95$ of the mass, which is an unexpected gain, and continue with the algorithm. If it is negative, we end up only with nodes that all have high probabilities. Then the algorithm tests nodes individually. 

We refer to the first stage of the algorithm when an informative test can be found based on $S$, $c<w(S)<1-c$, as Stage 1 of the algorithm and the second stage of the algorithm for which no $S$, $c<w(S)<1-c$, can be found as Stage 2 of the algorithm.

\textcolor{black}{In Corollary~\ref{corol:l1}, we show that in Stage 1, where each test removes at least $c$ fraction of the mass, the algorithm performs at most $O(H(X))$ tests in expectation which is near-optimal. In Stage 2, where high probability nodes are tested individually, we show in Observation~\ref{obs:L2} that the number of tests depends on the expected number of infections at that point, $\tilde{\mu}$. This expectation is obtained by the posterior probability after performing all the previous tests. We prove bounds on $\tilde{\mu}$ under assumptions on the size of the edges. Specifically, when the size of the edges is concentrated, we get a bound in the form of $O(\mu + H(X))$, which is stated in Theorem~\ref{theorem:edgesize}. Another way to bound the number of tests in Stage 2 is to first drop the edges that are much larger than $\mu$, which have a low mass, and then run the algorithm. In this case, we can guarantee that $\tilde{\mu}$ is not much larger than $\mu$, and the guarantee for this approach is shown in Theorem~\ref{theorem:final}.}

\subsection{The Proposed Algorithm}\label{sec:algrep}
\begin{algorithm}
\caption{Adaptive Algorithm for General Group Testing}\label{mainalg}
    \textbf{Input:} Graph $G = (V,E)$, distribution $\mathcal{D} = \mathcal{D}$ over $E$, $c < 1/2$.\\
    \textbf{Output:} Target edge $e^*$.\\
    Initialize $W = \phi$ which is the set containing negative nodes.\\
Initialize set $S = V \setminus W$. \label{alg:beginloop}\\
    Compute $w(S)$ by \eqref{eq:w} using the (updated) posterior distribution $\mathcal{D}=\{q_e\}_{e\in E}$. \label{alg:4}\\
    \textbf{If} $c \leq w(S) \leq 1-c$, test $V\setminus S$: \label{alg:linetest}\\
         \hspace{10pt} If the test is positive, remove $E(S)$ from $\mathcal{D}$ and update the posterior distribution following Definition~\ref{def:update}. Set $W = W \bigcup\limits_{v: q_v = 0} v$.\label{alg:update1}\\
         \hspace{10pt} If the test is negative, remove $E \setminus E(S)$ from $\mathcal{D}$ and update the posterior distribution following Definition~\ref{def:update}. Set $W = W \bigcup\limits_{v: q_v = 0} v$.\label{alg:update2}\\
         \hspace{10pt} Return to line~\ref{alg:beginloop}.\\
    \textbf{Else}: \\
    \hspace{10pt} If there is an edge with $q_e = 1$, return $e$. \label{alg:finish1}\\
    \hspace{10pt} If $\exists v: c \leq w(S\setminus v) \leq 1-c$, remove $v$ from $S$ and go to line~\ref{alg:linetest}.  \label{alg:removebig}\\
    \hspace{10pt} If $\exists v: 1-c < w(S\setminus v)$, remove $v$ from $S$ and go to line~\ref{alg:removebig}. \label{alg:removesmall}\\
     \hspace{10pt} Otherwise ($w(S) > 1-c$ and $\forall v: w(S\setminus v) < c$), test $V\setminus S$. \label{alg:setnotfound}\\
     \hspace{20pt} If the test is positive, remove $E(S)$ from $\mathcal{D}$ and update the posterior distribution following Definition~\ref{def:update}. Set $W = W \bigcup\limits_{v: q_v = 0} v$. Go to line~\ref{alg:beginloop}. \label{alg:update3}\\
     \hspace{20pt} If the test is negative, remove $E \setminus E(S)$ from $\mathcal{D}$ and update the posterior distribution following Definition~\ref{def:update}. Test every node $v \in S$ individually with $0 < q_v < 1$, update the posterior after every test and return the nodes with $q_v = 1$. \label{alg:finaltest}\\
\end{algorithm}


The proposed algorithm is presented in Algorithm~\ref{mainalg}. The algorithm undertakes two broad tasks repeatedly:
\begin{enumerate}
\item Determine the set of nodes it tests, denoted as $V\setminus S$, it identifies a set $S$ whose complement it tests;
\item \label{step2} If the set tests positive, only the edges which intersect with the tested set is retained, the rest are removed, if the set tests negative, all edges containing at least $1$ node in the set are removed, the tested set is removed from $V$ and added to set $W$, a set of nodes which are known to be negative; the posterior probability of each surviving edge is updated and $w(S)$ is calculated for each $S$ using  \eqref{eq:w} and the updated posteriors. 
\end{enumerate}

We now describe how these broad tasks are executed.  In Stage 1 of the algorithm which corresponds to lines~\ref{alg:beginloop}-\ref{alg:update3}, the algorithm greedily finds a set $S$ so that $c \leq w(S) \leq 1-c$. This is captured in lines~\ref{alg:removebig} and \ref{alg:removesmall}. When the set $S$ is found, $V \setminus S$ is tested in line~\ref{alg:linetest} and the posterior is updated accordingly. Such a test helps rule out edges with total weight of at least $c$ regardless of the outcome of the test. \textcolor{black}{If the algorithm can not find a set $S \setminus v$ so that $c \leq w(S \setminus v) \leq 1-c$ at this stage, it tests $V \setminus S$ in line~\ref{alg:setnotfound} which has a high probability of being negative. If at line~\ref{alg:setnotfound}, the test result is positive (unexpectedly), the edges that the algorithm rules out constitute of $1-c$ fraction of the mass and the algorithm gets back to line~\ref{alg:beginloop}.} Stage 2 of the algorithm is outlined in line~\ref{alg:finaltest}, where it tests uncertain high probability nodes individually. Note that at this point, all the nodes outside of $S$ are negative.

We now provide details of the sequence followed in the algorithm. At the start of the algorithm, typically the set of nodes that are known to be negative, $W$, is empty. Thus in line~\ref{alg:beginloop}, the set $S$ is initialized as $V \setminus W = V.$ Clearly in line~\ref{alg:4} we have $w(S) = 1$, and $S $ does not satisfy the condition in line~\ref{alg:linetest}. The algorithm would now move to line~\ref{alg:removebig} (typically the condition in line~\ref{alg:finish1} would also fail as the target edge will not be clear at this point). In line~\ref{alg:removebig}, the algorithm identifies individual nodes which are tested in the loop starting from line~\ref{alg:linetest}, and testing these individual nodes rules out edges with total weight at least $c$ regardless of the outcome of the test. After each test, posteriori probabilities are modified and the graph is modified as in Definition~\ref{def:update} (this modification happens in lines~\ref{alg:update1},\ref{alg:update2} and in line~\ref{alg:update3}). The above sequence is repeated until there does not exist individual nodes such that by removing them $S$ satisfies the condition in line~\ref{alg:linetest}. Note that $w(S) = 1$ the first time that the algorithm reaches  line~\ref{alg:removebig}. The tests are all individual while this duration lasts. 

Once the duration described in the previous paragraph ends the algorithm moves to line~\ref{alg:removesmall}. Note that $w(S) = 1$ the first time that the algorithm reaches  line~\ref{alg:removesmall} under the then posterior probabilities. When the algorithm reaches line~\ref{alg:removesmall}, $S$ consists only of nodes that are positive either with very high probability (such that if any one of those is removed, $w(S)$ falls below $c$ from $1$), or with very low probability (such that $w(S)$ exceeds $1-c$ even when any of those is removed from $S$). In  line~\ref{alg:removesmall}, the low probability nodes are identified one at a time, and once such a low probability node is identified, line~\ref{alg:removebig} determines if now there exists a node whose removal ensures that $w(S)$ satisfies the condition in line~\ref{alg:linetest}. If there exists such a node, in line~\ref{alg:linetest} it is tested together with the low probability node identified in line~\ref{alg:removesmall}. If there does not exist such a node, then the algorithm gets back to line~\ref{alg:removesmall} again and removes another low probability node and the same process is repeated. If no node is found in line~\ref{alg:removesmall}, then all low probability nodes are tested together in line~\ref{alg:setnotfound}. Thus, tests are no longer individual once the algorithm reaches line~\ref{alg:removesmall}. When a group of nodes is tested, either in line~\ref{alg:linetest} or line~\ref{alg:setnotfound}, the weights are modified as in Definition~\ref{def:update}. If the test in  line~\ref{alg:setnotfound} is negative, then all tested nodes are removed from $V$, and the remaining nodes are all positive with high probability and are tested individually in line~\ref{alg:finaltest}. If the test is positive, then only edges consisting of at least $1$ node in the set are retained, and the algorithm returns to line~\ref{alg:beginloop} and the sequence is repeated.


\begin{Example}
We now run Algorithm~\ref{mainalg} on the graph in Figure~\ref{fig:hypgraph} with $c = 0.1$. Refer to the values of $ \{q_e\}$ in the caption of Figure~\ref{fig:hypgraph}, $\{p_v\}$ calculated in Section~\ref{sec:newdefs} and that $w(V \setminus \{v\})= 1 - p_{v}$ also from Section~\ref{sec:newdefs}. In line~\ref{alg:beginloop} we set $S = V =  \{v_1, v_2, v_3, v_4, v_5 \}$, where  $W = \phi$. Thus, $w(S) = 1$. Thus, $S$ does not satisfy the condition in line~\ref{alg:linetest}, and the Algorithm moves to line~\ref{alg:finish1} and then line~\ref{alg:removebig}. Here,  $w(V \setminus \{v_1\})= 1 - p_{v_1} = 0.5$. Thus, $c < w(S\setminus v_1) = 0.5 < 1-c$, the algorithm removes $v_1$ from $S$, and tests $v_1$ in line~\ref{alg:linetest}. If it is negative, then $v_1$ is added to $W$,  $e_1$ and $e_2$ are ruled out as target edges as both contain $v_1$, and after scaling the edges in line~\ref{alg:update2} following Definition~\ref{def:update}, posterior  $q_{e_3} = 1$ since $e_3$ is the only remaining edge. And the algorithm returns it as the target edge in line~\ref{alg:finish1}. If $v_1$ is positive,  in line~\ref{alg:update1}, $e_3$ is ruled out as the target edge as it does not contain $v_1.$ Since $q_{e_3} = 0.5$, in line~\ref{alg:update1}, the algorithm  scales $q_{e_1}, q_{e_2}$ by 2 so posterior $q_{e_1} = 0.6$ and $q_{e_2} = 0.4$. Again in line~\ref{alg:beginloop}, $S$ would be initialized to $V.$ Thus $w(S) = 1$ and $S$ does not satisfy the criterion in line~\ref{alg:linetest}. Thus, the algorithm moves to line~\ref{alg:finish1}, finds that there is no $q_e$ such that $q_e = 1.$ It then moves to line \ref{alg:removebig}. Since  $S = V$, $w(S \setminus \{v_2\}) = 1 - p_{v_2} = 0.4$, so the algorithm identifies $v_2$ for individual test. The test result will eliminate $e_1$ if $v_2$ is positive, and $e_2$ otherwise in lines~\ref{alg:update1} and \ref{alg:update2}. Either way, only one edge remains, thus its posterior probability of being the target edge is scaled to $1$ in line~\ref{alg:update1} or \ref{alg:update2} (Definition~\ref{def:update}), and it is identified as the target edge in line~\ref{alg:finish1}. The expected number of tests here is $1.5.$ 
\end{Example}

In the following, we give an important remark and proceed with a more abstract and larger example.

\begin{remark} \label{nofurthertest} 
We remark that any node for which the posterior probability becomes 0 or 1 will no longer be tested. In particular, if an individual node (set) is tested negative, since it is removed from $V$, it is never tested again. If an individual node (set) is tested positive, only edges intersecting that node (set) remain in the graph, thus the weight of the complement of that node (set) becomes $0$. Thus, the complement set will never have weight in the interval $[c, 1-c]$ and is therefore not tested again. In the following, we elaborate in details on why a single node is not tested again once its infection probability reaches zero or one.

Suppose at any stage of Algorithm \ref{mainalg}  all remaining edges contain a particular node $v.$ This may happen for example right after $v$ tests positive in an individual test. Then  following Definition~\ref{def:update}, the posterior probability of the node will be $1$, and removing the node from any set $S$ will result in $S$ not containing any edge, and the weight of the set will become $0$. Thus, such a node can not be included in any test.

In addition, if at any stage of the algorithm no remaining edge contains a node $v$, which can for example happen if any group including $v$ tests negative,  then its posterior probability becomes $0$, and it is moved to $W$ from $V.$  Right after a node is moved to $W$, the set $S$ is initialized as $V \setminus W$ in line~\ref{alg:beginloop}, thus such a node does not belong to $S$ henceforth and is therefore not removed from it in any line and is therefore not tested further.

\end{remark}
\begin{Example}
We now elucidate how Algorithm \ref{mainalg} functions for Example \ref{ex:islands} in Section~\ref{sec:newdefs}. First note that whenever a node $v$  is tested individually, the test reveals its state and also the states of all nodes in its island as all nodes in an island have the same state. Accordingly no node in $v$’s island is tested further as we argue next. If $v$ is negative, all edges that contain it are removed. But any edge containing it also contains all nodes in its island. Thus all edges containing any node in $v$’s island are removed. If $v$ tests positive, only the edges containing it are retained (line~\ref{alg:update2}), and all such edges contain all the nodes in its island since every edge contains all or none of the nodes in each island. Thus, every remaining edge has all nodes in $v$’s island. Thus, by Remark \ref{nofurthertest}, no node in $v$’s island is tested henceforth regardless of the outcome of $v$’s test. We now describe the sequence of testing of nodes. 

The  first decision on which node to test is made when Algorithm \ref{mainalg} reaches line~\ref{alg:removebig} for the first time. Then $S = V$ and $w(S) = 1.$ Recall that we had argued in Section~\ref{sec:newdefs} that
$w(V \setminus \{v\}) =  1- p_v.$ Thus,
$w(S\setminus \{v\})$ is $ \gamma $, $\approx 1/2,$  $1-\epsilon$ if $v$ is respectively in a high, moderate, low island, where $\gamma, \epsilon$ are small positive numbers.

Assuming that $\gamma, \epsilon < c$, it follows that when the algorithm reaches line~\ref{alg:removebig} for the first time, $c < w(S \setminus v) < 1-c$ if and only if $v$ is in a moderate island. Thus,  Algorithm \ref{mainalg} first picks a candidate from one of the moderate islands in line~\ref{alg:removebig} and tests it individually. As discussed above, testing  $v$ reveals the state of all nodes in $v$’s island and subsequently  no other node from the same island is tested. Since nodes in different islands are independent, the knowledge of the state of the tested node does not affect the probabilities that nodes in other islands are positive. Thus following the same arguments as before the algorithm continues to visit line~\ref{alg:removebig}, and in each such visit identifies a node from a new moderate island for individual testing. The outcome of the test reveals the state of all nodes in its island and subsequently  no other node from the same island is tested. This continues until all moderate islands are exhausted. The remaining islands are high or low.

Then, the algorithm reaches  line~\ref{alg:removesmall} and identifies a low probability node there, that is, a node from a low island. Assuming that the $\epsilon < c/k_3$ where $k_3$ is the number of low islands, removing an additional node in line~\ref{alg:removebig}, either reduces the weight of $S$ to below $c$ or maintains it at greater than $1-c$, thus no further node is identified for testing in line~\ref{alg:removebig}.

Repeated visits to line~\ref{alg:removesmall}, gathers all the nodes from all low islands, and tests them all together in line~\ref{alg:setnotfound}. With high probability, the test is negative, and if so, all nodes in the low islands are moved to $W$, all edges containing low islands are removed, and only the high islands remain. Then the algorithm tests nodes in high islands individually in line~\ref{alg:finaltest}. Each individual test reveals the state of all nodes in the tested node’s island and no other node in the tested node’s island need to be tested further as argued in the first paragraph. The algorithm outputs the edge containing all the positive nodes as the target edge. 

If however the group test of all the nodes in the low islands turn out to be positive, which happens with low probability, then these nodes would subsequently be either tested individually, identified one at a time in line~\ref{alg:removebig}, or tested individually in line~\ref{alg:finaltest}, or tested in smaller subgroups which are identified in line~\ref{alg:removesmall}  or in both line~\ref{alg:removesmall} and line~\ref{alg:removebig}. \footnote{Here, we describe in details how the algorithm performs in the unlikely event that low islands test positive. If the group test of all the nodes in the low islands turns out to be positive, the edges consisting only of the high islands will be removed. This would substantially increase the posterior probabilities of edges which include low islands which would in turn substantially increase the posterior probabilities of the nodes in low islands. The posterior probabilities of the nodes in high islands will not change because of independence of states of nodes across islands. The algorithm will move to line~\ref{alg:beginloop} and subsequently to line~\ref{alg:removebig}. Since the posterior probabilities of nodes in high islands remain same as before, they will not qualify for individual tests as before.

But, since the posterior probabilities of nodes in low islands have increased, they may qualify for individual tests in line~\ref{alg:removebig}. If so, as for moderate islands, one node in each low island would be tested and the outcome would be used to determine the states of each node in the same island and those nodes will not be tested further. At the end of the process, only the high islands will remain. Since they are all high probability nodes, no node would be  identified for individual testing in  line~\ref{alg:removesmall}. But these nodes would be individually tested in line~\ref{alg:finaltest} and their states determined.

If however the posterior probabilities of nodes in low islands increased significantly such that the removal of any of them reduces the weight of $S$ to below $c$, then they would not qualify for individual tests in line~\ref{alg:removebig}. Then all remaining nodes including those of low islands are high probability nodes. Thus they would not be removed in line~\ref{alg:removesmall} either, and all remaining nodes would be individually tested in line~\ref{alg:finaltest}. The last possibility is that the posterior probabilities of nodes in low islands increased, but still removing any one of them retains the weight of $S$ to above $1-c.$ Then, again they would not qualify for individual tests in line \ref{alg:removebig}. But they can be removed as low probability nodes in line~\ref{alg:removesmall}. But note that all nodes of low islands can not be removed together as low probability nodes in line~\ref{alg:removesmall}, because only edges that contain at least one of them have been retained, thus together all remaining edges contain all of them, and hence removing all of them removes all the edges too and reduces the weight of $S$ to $0 < c.$ Then subsets of nodes in low islands would be identified in lines~\ref{alg:removebig} and \ref{alg:removesmall} and tested in line~\ref{alg:linetest}, and the sequence is repeated.}

Summarily, the nodes in high and moderate islands are individually tested, but those in low islands are tested all together at least once, and  in the unlikely scenario that this group test is positive they are tested in various subgroups to determine their states. Only one node in each high and moderate island is tested. 
\end{Example}
\begin{remark}
In line~\ref{alg:removebig}, the algorithm identifies the first node that satisfies the condition given in this line, and removes it from the set $S$. However, the algorithm could instead consider all nodes individually or multiple nodes simultaneously. Doing so would ensure not only that the condition in this line is satisfied but also that the weight of the set $S$ after removal is as close to $1/2$ as possible. A set weight closer to $1/2$ results in a more informative test in practice. However, such a selection strategy, while potentially improving practical performance, would be computationally more expensive and would not affect the theoretical guarantees or our analysis.
\end{remark}

\section{Analysis}\label{sec:analysis}


In this section, we bound the expected number of tests performed by Algorithm~\ref{mainalg}. The main result is built on the following key lemma. Here,  $H(X) = \sum_e p_e \log 1/p_e$ is the entropy of the nodes and $\tilde{\mu}$ is the expected number of infections in line~\ref{alg:finaltest} (computed with the posterior probabilities).

\begin{lemma}\label{lemma:main}
    Algorithm~\ref{mainalg} finds $e^*$ with 
    $\mathbb{E}[L_1 + L_2] \leq \frac{1}{\log \frac{1}{(1-c)}}H(X) + 1 + \frac{\mathbb{E}[{\tilde{\mu} }]}{1-2c}$ tests in expectation. 
\end{lemma}
\textcolor{black}{The roadmap of our proof for Lemma \ref{lemma:main} is as follows. We first prove that removing $E(S)$ (resp. $E\backslash E(S)$) from $\mathcal{D}$ (see Definition \ref{def:update}) upon obtaining a positive (resp. negative) test result for $V\backslash S$ provides the true posterior distribution. This is established in  Lemma \ref{lemma:poster}. 
We next bound the number of tests $L_1$ done in Stage 1 of the algorithm by $\frac{1}{\log \frac{1}{(1-c)}}H(X)$, in expectation, and then show that the number of individual tests $L_2$ done in Stage 2 is bounded by $\frac{{\tilde{\mu} }}{1-2c}$. }

\paragraph{Posterior Update.} We first need the following lemma to justify the correct operation of Algorithm~\ref{mainalg} in removing edges after each test. In particular, the target edge $e^*$ is never removed and each test indicates whether the target edge is in $E(S)$ or $E \setminus E(S)$. 
\begin{lemma}\label{lemma: update}
The result of a test $V \setminus S$ is positive iff $e^* \in {E \setminus E(S)}$.
    
\end{lemma}
\begin{proof}
    A test $V\setminus S$ is positive if at least one of the nodes in $V \setminus S$ is positive, which means $e^*$ contains a node in $V \setminus S$. But by definition, $E(S)$ is the set of edges with all nodes in $S$, hence $e^* \notin E(S)$ and $e^* \in E \setminus E(S)$. 
    Similarly, if $e^* \in E \setminus E(S)$, it contains at least one node in $V \setminus S$ and hence the result is positive.
\end{proof}
\noindent Building on Lemma \ref{lemma: update}, we further prove that the update rule utilized in Algorithm \ref{mainalg} computes the posterior probabilities given all the previous tests.



\begin{lemma}\label{lemma:poster}
    Consider a distribution $\mathcal{D}$ over the edge set $E$ and a test $T=V \setminus S$ with result $r$. If $r=1$, the posterior probability is obtained by removing $E(S)$ according to Definition~\ref{def:update}. If $r=0$, the posterior probability is obtained by removing $E \setminus E(S)$ according to Definition~\ref{def:update}.
\end{lemma}
\begin{proof}
For simplicity, consider the update after the first test and $r = 1$. Then the posterior is
\begin{align}
q_{e_1 \mid (T_1 = S \setminus V, r_1 = 1)} &= Pr(e_1 = e^* |T_1 = S \setminus V, r_1 = 1)  \label{eq:1update}\\
&= Pr(e_1 = e^* | e^* \in E \setminus E(S)) \text{ by Lemma~\ref{lemma: update}}\\
&= Pr(e_1 = e^*, e^* \in E \setminus E(S))/ P(e^* \in E \setminus E(S)) \\
&= \begin{cases}0, & \text { if } e_1 \notin E\setminus E(S) \\ Pr(e_1 = e^*) / (\sum_{e \in E \setminus E(S)} p(e)) & \text {otherwise} \end{cases}
\label{eq: lupdate}
\end{align}

\noindent The case for $r_1 = 0$, a negative test, can be argued similarly. The main difference in this case is that we will prove the reverse of lemma~\ref{lemma: update}, ie,  $r_1 = 0$  iff $e^* \in E(S)$. Now, $r_1 = 0$ needs to be substituted with $r_1 = 1$, and $E(S)$ with  $E \setminus E(S)$ in the above equations, and otherwise the same result can be obtained.  

For tests $i$, $i > 1$, Lemma~\ref{lemma: update} holds for $r_i = 1$, and considering $E$ updated after $i-1$ tests. Equations \eqref{eq:1update} to \eqref{eq: lupdate} follow with the conditioning updated to include results of all tests and using posterior probabilities $q_e$ instead of $p(e)$. The same equations can be derived for $r_i = 0$ as in the previous paragraph. 
\end{proof}

\paragraph{Informative Tests: } Formally, a test is informative if the edge updates right after it remove at least $c$ fraction of the mass. The following lemma upper bounds the expected number of informative tests for any algorithm that returns an edge with weight $1$ as the target edge. 

\begin{lemma}\label{lemma:cfrac}
    The expected number of informative tests is at most  $ \frac{1}{\log\ \frac {1} {(1-c)}}H(X)$ for any algorithm that decides that an edge $e$ is the target edge whenever $q_e=1.$ 
\end{lemma}

\begin{proof}
We know that every time the posterior is updated after every informative test, at least a mass $c$ is removed and the remaining edge probabilities are divided  by at most a factor of $1-c$.  It therefore takes at most $\log_{\frac{1}{1-c}}(1/p(e))= \frac{\log 1/p(e)}{\log(\frac{1}{1-c})} $ informative tests for the posterior on a given edge $e$ to become 1 (if it survives those tests which it would if it is the target edge). This leads to the following upper bound on the expected number of informative tests:
    \begin{equation*}
        \sum_{e \in E} p(e) \log_{\frac{1}{1-c}}(p(e)) \\
        = \sum_{e \in E} p(e) \frac{\log 1/ p(e)}{\log \frac{1}{1-c}}\\
        = \frac{1}{\log\frac{1}{1-c}} \sum_{e \in E} p(e) \log 1/ p(e)\\
        = \frac{1}{\log\ \frac {1} {(1-c)}}H(X).
    \end{equation*}
   
\end{proof}



The following observation identifies the informative tests in the algorithm. 

\begin{observation}\label{observation:fractionremoved}
    For Algorithm~\ref{mainalg}, tests in line~\ref{alg:linetest} and positive tests in line \ref{alg:setnotfound} are informative, provided $0 < c < 1/2.$
\end{observation}
\begin{proof}
It would be sufficient to prove that edge updates in lines~\ref{alg:update1}, \ref{alg:update2},\ref{alg:update3} remove at least $c$ fraction of the mass. 
If the algorithm scale $\mathcal{D}$ at lines~\ref{alg:update1} or $\ref{alg:update2}$, by lemma~\ref{lemma: update}, we either remove $E(S)$ or $E \setminus E(S)$. But $c \leq w(S) \leq 1-c$, so in either case we remove  $c$ fraction of the mass.

    If the algorithm scales $\mathcal{D}$ at line~\ref{alg:update3}, it means $w(S) > 1-c$ and $V \setminus S$ has tested positive. The latter implies that $E(S)$ is removed from $E.$ Weight of $E(S)$ is $ w(S).$ Thus, at least a mass of $1-c$ is removed. Now $1-c > c$ because $1-c > 1/2 > c$. Thus  at least a mass of  $c$  have been removed. 
\end{proof}

Note that there can be at most one negative test in line \ref{alg:setnotfound}, since right after such a negative test, the algorithm proceeds to line~\ref{alg:finaltest}, and never returns to lines $1$ to \ref{alg:update3}. Thus, all but $1$ tests performed before the algorithm proceeds to line~\ref{alg:finaltest}, ie, in Stage $1$ of the algorithm, are informative. Using Lemma~\ref{lemma:cfrac}, now we can upper bound the expected number of tests performed in the first stage of the algorithm.

\begin{corollary}\label{corol:l1}
    Before line~\ref{alg:finaltest}, the number of tests $L_1$ that Algorithm~\ref{mainalg} performs is bounded in expectation as follows: 
    $$\mathbb{E}[L_1] \leq \frac{1}{\log\ \frac {1} {(1-c)}}H(X)+1.$$  
\end{corollary}

Line \ref{alg:finaltest} constitutes Stage $2$ of the algorithm. 

\paragraph{Individual Testing in Stage 2.}  Stage $1$ of Algorithm \ref{mainalg} concludes when the greedy process of finding $S$, $c\leq w(S)
\leq 1-c$, fails in line~\ref{alg:removesmall}, and the nodes outside $S$ test negative in line~\ref{alg:setnotfound}. At this point, in Stage 2, we resort to individual testing in line \ref{alg:finaltest}. We will upper bound the expected number of such individual tests using the following Proposition. 

\begin{proposition}\label{prop:nodeweights}
    When Algorithm~\ref{mainalg} reaches~\ref{alg:setnotfound}, we have (i) $w(S) > 1-c$, (ii) $Pr(\exists v \in V \setminus S : v \in e^*) < c$ where $e^*$ is the target edge, (iii) $\sum_{e \in E(S), v \in e} p_e > 1-2c$, and (iv) $\forall v \in S: q_v > 1-2c$. 
\end{proposition}
\begin{proof}
\textcolor{black}{To see (i) note that at line~\ref{alg:beginloop} $w(S) = 1 > 1-c$ and, as we proceed by removing nodes from $S$ to meet the criteria of line \ref{alg:removebig}, it never gets met and therefore when the algorithm reaches line~\ref{alg:setnotfound}, we have $w(S) > 1-c$.   To prove (ii), we compute the probability that at least one node $v \in V \setminus S$ is positive. By Lemma~\ref{lemma: update}, this is equal to the probability that $e^* \in E \setminus E(S)$, i.e., $\sum_{e \in E \setminus E(S)} q_e=1-w(S)$ which is at most $c$ by (i).  
 Finally, to prove (iii) and (iv),
     we note that $\forall v \in S$, $q_v \geq  \sum_{e \in E(S), v \in e} q_e = w(S) - w(S \setminus v)$. Since the algorithm has passed line~\ref{alg:removebig}, so $w(S) > 1-c$ and  $w(S \setminus v) < c$ for any node $v$. Therefore, $q_v \geq  \sum_{e \in E(S), v \in e} q_e > 1-2c$, proving (iii) and (iv).}
\end{proof}

Using Proposition~\ref{prop:nodeweights}, we bound the expected number of tests in line~\ref{alg:finaltest}:

\begin{lemma}\label{obs:L2}
    Let $\tilde{\mu} = \mu_{(\mathbf{T_1,r_1),\ldots, (T_{L_1}, r_{L_1})}}$ be the expected number of infected nodes at the start of line~\ref{alg:finaltest} where $\mathbf{(T_i, r_i)}, 1 \leq i \leq L_1$ are random variables which indicate the tests and results performed.
    Then the algorithm performs at most $L_2 \leq \frac{\tilde{\mu}}{1-2c}$ tests in line~\ref{alg:finaltest}. 
\end{lemma}
\begin{proof}
    Suppose test $\mathbf{(T_i,r_i)} = (T_i,r_i)$ is done. In other words, we condition $\tilde{\mu}$ to the case where tests are $(T_i, r_i)$'s. First recall that at any step of the Algorithm, $q_v$ stands for $ q_{v \mid \{\{(T_1,r_1),\ldots, (T_{L_1}, r_{L_1})\}}$. Let $X_v$ be the indicator variable representing if node $v$ is infected or otherwise. Thus, $q_v = E[X_v \mid \{(T_1,r_1),\ldots, (T_{L_1}, r_{L_1})\}].$  And
\begin{align*}    
&\tilde{\mu} \mid \{\mathbf{(T_1,r_1)} = (T_1, r_1),\ldots, (\mathbf{T_{L_1}, r_{L_1})} = (T_{L_1}, r_{L_1})\}  \\
&\triangleq \mathbb{E}[\sum_v X_v \mid \{(T_1,r_1),\ldots, (T_{L_1}, r_{L_1})\}]  
\\&=  \mathbb{E}[ \sum_{v \in S} X_v \mid \{(T_1,r_1),\ldots, (T_{L_1}, r_{L_1})\}] 
\\&= \sum_{v \in S} \mathbb{E}[X_v \mid \{(T_1,r_1),\ldots, (T_{L_1}, r_{L_1})\} ]\\
&= \sum_{v \in S}  q_v > |S|(1-2c)  
\end{align*}

{Recall that $S$ here is defined in Algorithm~\ref{mainalg} when it reached line~\ref{alg:finaltest}.} The first equality is true because in the previous line, $V \setminus S$ tested negative, so $X_v = 0$ for $ v \in V \setminus S.$ The second equality is true because given the test results and which nodes have been tested thus far, the set $S$ is specified, thus the expectation can be moved inside the summation. The third equality follows from what $q_v$ stands for. The fourth inequality follows from Proposition~\ref{prop:nodeweights} (iii), from which we know that for each  $v \in S$, $\sum_{e: e \in E(S), v \in e} p_e  > 1-2c$ in line~\ref{alg:setnotfound}. As the test result in line~\ref{alg:setnotfound} is negative, following Definition~\ref{def:update}, the weight of the edges in $E(S)$ only scaled up at the start of  line~\ref{alg:finaltest}. Thus, $\sum_{e: e \in E(S), v \in e} p_e > 1-2c$, right after the negative test. Since $\forall v: q_v \geq \sum_{e: e \in E(S), v \in e} p_e$, $q_v > 1-2c$. But the algorithm performs $|S|$ tests at line~\ref{alg:finaltest}, so $L_2 = |S|$. Thus regardless of the test results,  $L_2 < \frac{\tilde{\mu}}{1-2c}.$ 
\end{proof}

The above lemma in conjunction with Corollary~\ref{corol:l1} proves Lemma~\ref{lemma:main}.

Theorem \ref{theorem:edgesize} below is a  direct implication of Lemma~\ref{lemma:main}  when the size of the edges do not differ by much. 

\begin{theorem}\label{theorem:edgesize}
    If $\forall e \in E: f_1(n) \leq |e| \leq f_2(n)$, then Algorithm~\ref{mainalg} finds $e^*$ with $\frac{1}{\log \frac{1}{(1-c)}}H(X) +  \frac{f_2(n)}{f_1(n)(1-2c)}\mu + 1$ tests in expectation. 
\end{theorem}
\begin{proof}
    Recall that $\mu$ is the average number of infections before performing any tests. As $f_1(n) \leq |e|$ for each edge $e$, we have $\mu \geq f_1(n)$. As $|e| \leq f_2(n)$, we have $\tilde{\mu} \leq f_2(n)$. Then the ratio $\frac{\tilde{\mu}}{\mu} \leq \frac{f_2(n)}{f_1(n)}$, hence $\mathbb{E}[\tilde{\mu}] \leq \frac{f_2(n)}{f_1(n)} \mu$. Replacing this in $\mathbb{E} [L_1 + L_2] \leq \frac{1}{\log \frac{1}{(1-c)}}H(X) + 1+  \frac{\mathbb{E}[\tilde{\mu}]}{1-2c}$ from Lemma~\ref{lemma:main} completes the proof.
\end{proof}

If most of the edges, not all edges, have an upper bound of the form $f_2(n)$, Algorithm~\ref{mainalg} can be modified to only return the correct edge when the target edge is smaller or equal than $f_2(n)$. This idea is summarized in the following corollary.
\begin{corollary}\label{cor:conctedge}
If there is an $E' \subseteq E$ such that $\sum_{e\in E'} p(e) = 1-\epsilon$ and $\forall e \in E': |e| \leq f_2(n)$, then Algorithm~\ref{mainalg} can be modified so that it performs $\frac{2}{\log \frac{1}{(1-c)}}H(X) + 2f_2(n)$ tests in expectation, and returns the target edge $e^*$ with probability $1-\epsilon$, where $c < 1/3$. 
\end{corollary}

\begin{proof}    Let $PN$ be the number of positive tests in line~\ref{alg:finaltest}, which  is modified as follows:

    \begin{enumerate}
    \setcounter{enumi}{15} 

        \item If the test is negative, remove $E \setminus E(S)$ from $\mathcal{D}$ and update the posterior distribution following Definition~\ref{def:update}. Pick a node $v \in S$ at random with $0 < q_v < 1$. Test $v$ and update the posterior and if $v$ is positive increment $PN$ by $1$. If $PN \geq f_2(n)$, return all the positive nodes as $e^*$; else, go to line~\ref{alg:beginloop}. \label{alg:new16}
    \end{enumerate}

  First consider the informative tests. As per Lemma~\ref{lemma:cfrac}, the expected number of informative tests is  at most $\frac{1}{\log\ \frac {1} {(1-c)}}H(X)$. As in Observation \ref{observation:fractionremoved}, 
    tests in line~\ref{alg:linetest} and positive tests in line \ref{alg:setnotfound} are informative.   We now prove that a negative test in line~\ref{alg:new16} is informative. When the algorithm reaches line~\ref{alg:linetest} , for each $v$ in $S$, $\sum_{e: e \in E(S),  v \in e} p_e > 1-2c$ (this can be argued similar to Proposition~\ref{prop:nodeweights}). If a node $v$ tests negative in line~\ref{alg:new16}, then all edges that contain $v$ are removed, thus a weight of at least $\sum_{e: e \in E(S),  v \in e} p_e$ is removed (note that edges that are not contained in $E(S)$ but have $v$ will also be removed, so we use the qualifier “at least”). Thus, a weight of at least $1-2c$ is removed. Note that   $1-2c > c$ since $c < 1/3$. The claim follows. Thus, the expected total number of tests in  line~\ref{alg:linetest} and  positive tests in line \ref{alg:setnotfound} and negative tests in line~\ref{alg:new16} is  $\frac{1}{\log\ \frac {1} {(1-c)}}H(X)$.

We now upper bound the remaining tests in expectation. The remaining tests are positive tests in line~\ref{alg:new16} and negative tests in line~\ref{alg:setnotfound}. We first consider the positive tests in line~\ref{alg:new16}. Here, every positive test increments  $PN$ by at least $1.$ It therefore follows from line~\ref{alg:new16} that there are at most $f_2(n)$ such positive tests. Next, note that  between any two negative tests in line~\ref{alg:setnotfound}, at least one test is conducted in line~\ref{alg:new16}, which is informative if negative, or positive otherwise. Thus the number of negative tests in line~\ref{alg:setnotfound} is at most the number of informative tests and the number of positive tests in line~\ref{alg:new16}. Thus, the expected number of negative tests in line~\ref{alg:setnotfound} is at most  $\frac{1}{\log\ \frac {1} {(1-c)}}H(X) + f_2(n).$ Putting all together, the algorithm performs at most $\frac{2}{\log \frac{1}{(1-c)}}H(X) + 2f_2(n)$ tests in expectation.

    If $|e^*| \leq f_2(n)$, we always have $PN \leq f_2(n)$ and hence the algorithm terminates with $PN < f_2(n)$ or $PN = f_2(n).$ In the former case the  algorithm terminates only when it discovers an edge $e^*$ such that $p(e^*) = 1.$ In the latter case it terminates returning $f_2(n)$ nodes that tested positive. Since $|e^*| \leq f_2(n)$ and all positive nodes are in $e^*$, the algorithm returns all the positive nodes, ie, it returns $e^*$. But  the target edge is less than or equal to $f_2(n)$ with probability  $1-\epsilon$; hence the probability of returning the target edge is at least $1-\epsilon$ (the target edge may be returned correctly even when it exceeds $f_2(n)$). 
   
\end{proof}

Note that under the conditions of Corollary \ref{cor:conctedge}, and when $\forall e \in E': f_1(n) \leq |e| \leq f_2(n)$, $\epsilon < 2c$, from Corollary \ref{cor:conctedge}, the expected number of tests in the modified algorithm is also at most $\frac{2}{\log \frac{1}{(1-c)}}H(X) + \frac{2f_2(n)}{f_1(n)(1-2c)}\mu$ in expectation. This is because  $\mu \geq f_1(n)(1-2c)$ which we argue next. Note that $\mu = \mathbb{E}[|e^*|]$, and $|e^*| \geq f_1(n)$ with probability $1-\epsilon$. Thus, $\mu = \mathbb{E}[|e^*|] \geq f_1(n) (1-\epsilon) \geq f_1(n) (1-2c)$ since $\epsilon < 2c$. The claim follows.

While Theorem~\ref{theorem:edgesize} imposes a constraint on  the edge sizes and finds the infected set with probability 1, the constraint can be relaxed if a small probability of error ($\epsilon$, fixed or approaching zero) can be tolerated. We use Markov Inequality to bound the probability that $|e^*|$ exceeds $\mu/\epsilon$. By removing those edges for which $|e| > \mu/\epsilon$, we ensure $\tilde{\mu} < \mu /\epsilon$ and establish the following result.

\begin{theorem}\label{theorem:final}
There is an algorithm that performs $\mathbb{E}[L_1 + L_2] \leq \frac{2}{\log \frac{1}{(1-c)}}H(X) + \frac{2\mu}{\epsilon}$ tests in expectation and returns the correct edge $e^*$ with probability $1-\epsilon$, where $ c \leq 1/3$. Optimizing over $c$, with $c = 1/3$ the expected number of tests is $\mathbb{E}[L_1 + L_2] \leq \frac{2}{\log \frac{3}{2}}H(X) + \frac{2\mu}{\epsilon}$.
\end{theorem}
\begin{proof}
    Let $z = |e^*|$. Then we know $\mathbb{E}[{z}] = \mu$ and by Markov inequality, $Pr[z > \mu/\epsilon] < \epsilon$, meaning that with probability $1 - \epsilon$, the size of the target edge is at most $\mu / \epsilon$.  To exploit this fact, we can modify Algorithm~\ref{mainalg} similar to that in the proof of Corollary~\ref{cor:conctedge}. The only difference would be that we would replace $f_2(n)$ with $\mu/\epsilon$ in the modification. Thus, $f_2(n)$ should be replaced with $\mu/\epsilon$ in the upper bound on the expected number of tests given in the statement of the Corollary.

    Replacing $f_2(n)$ with $\mu/\epsilon$ it follows from Corollary~\ref{cor:conctedge} that the expected number of tests is at most $\frac{2}{\log \frac{1}{(1-c)}}H(X) + \frac{2\mu}{\epsilon}$ and the target edge is returned with probability $1-\epsilon.$ Because $c\leq 1/3$, the upper bound on the expected number of tests is minimized when $c = 1/3$ and in this case, the expected number of tests is at most $\frac{2}{\log \frac{3}{2}}H(X) + \frac{2\mu}{\epsilon}$.
    \end{proof}

\textcolor{black}{In the above Theorem, there is a trade-off between the probability of error and the expected number of tests. With $\epsilon$ error, the second term of expectation increases with $1/\epsilon$. So if we set $\epsilon \rightarrow 0$ when $n \rightarrow \infty$, for example $\epsilon = 1/\log n$, we incur a multiplicative factor of $\log n$ on the second term
}

\textcolor{black}{Markov bound works for general distributions, but it might happen that the mass in $\mathcal{D}$ is concentrated around the mean. In this case, the dependency on $1/\epsilon$ might be slower. For example, if the concentration is exponential, i.e. $Pr(|e^*| > 2\mu) < 2^{-\Omega(n)}$, then with probability $1 - 2^{-\Omega(n)}$ the algorithm recovers $e^*$ and the expected number of tests is at most $\mathbb{E}[L_1 + L_2] \leq \frac{2}{\log \frac{3}{2}}H(X) + 2\mu$ (by replacing $f_2(n)$ with $2\mu$ in the statement and proof of Corollary~\ref{cor:conctedge}). More precisely, we have the following Corollary.
}


\begin{corollary}\label{cor:concentr}
    Suppose $e^*$ is sampled from the probability distribution $\mathcal{D}$ and we have $Pr(|e^*| > u) = o(1)$. Then, with probability $1 - o(1)$, the Algorithm~\ref{mainalg} as modified in the proof of Corollary~\ref{cor:conctedge} (with $f_2(n)$ replaced by $u$) returns the correct edge $e^*$ using $O(H(X) + u)$ tests with probability $1 - o(1)$.
\end{corollary}

Note that in the above corollary, $u$ can be any stochastic upper bound on the size of the target edge. Thus, 
the modification of Algorithm~\ref{mainalg} does not need to know $\mu$, it just needs to know any stochastic upper bound on the size of the target edge;  the guarantee on the expected number of tests however improves as $u$ becomes smaller. Since the smallest value of a stochastic upper bound 
can be $O(\mu)$, the best guarantee via the above corollary is obtained if $\mu$ is known.

\section{Recovering/{Improving} Prior Results}\label{sec:5}

In this section, we consider the statistical models provided by prior related works and compare the results by the performance of Algorithm~\ref{mainalg}.

The work \cite{li2014group} considers independent group testing where node $i$ is infected with probability $p_i$ independent of the other nodes. The authors provide adaptive algorithms that use $O(\mu + H(X))$ tests on average and output the infected set with high probability. They also prove almost surely bounds on the number of tests. By applying Theorem~\ref{theorem:final} to this setting for a constant $c$, {we obtain $O(H(X) + \frac{\mu}{\epsilon})$ as an upper bound on the expected number of tests where $\epsilon$ is the error probability.} If $\mu = \sum_i p_i \gg 1$, which is a common assumption in the literature, then the number of infections is concentrated around $\mu$ and by Corollary~\ref{cor:concentr}, with probability $1-o(1)$, the algorithm in Theorem~\ref{theorem:final} recovers the infected set with $O(H(X) + \mu)$ tests, which recovers {the adaptive upper bound result (Theorem~2 and Corollary~1) up to a constant} provided in \cite{li2014group}.

In \cite{nikolopoulos2021group2}, as we discussed briefly in Section~\ref{sec:priormodel}, the authors consider $F$ families where each family $j$ has size $M_j$, so $n = \sum_j M_j$. Each family $j$ is infected with probability $q$ and when a family is infected, each of its nodes is infected with probability $p_j$ independent of the rest. If a family is not infected, then none of its nodes is infected. They show that to have a constant probability of recovery, one needs 
$L \geq  H(X) = F h_2(q)+\sum_{j=1}^F q M_j h_2\left(p_j\right)-w_j h_2\left(\frac{1-q}{w_j}\right)$ where $ w_j=1-q+q\left(1-p_j\right)^{M_j}$ and $h_2$ is the binary entropy. For the symmetric case where $\forall j p_j = p, M_j = M$ and hence $w_j = w$, the entropy simplifies to $H(X) = F h_2(q)+ q n h_2(p) - Fw h_2(\frac{1-q}{w})$. They provide an algorithm that needs at most $O(Fq( \log(F) + M) + nqp \cdot \log n)$ tests in expectation. 

Note that $\mu = nqp$, and by Theorem~\ref{theorem:final}, with probability $1-\epsilon$ the target edge can be found using $O(F h_2(q)+ q n h_2(p) - Fw h_2(\frac{1-q}{w}) + \frac{nqp}{\epsilon})$ tests for a constant $c$. If $F \gg 1$ and $M \gg 1$, then the expected number of infections is concentrated around its mean, and with probability $1 - o(1)$ the algorithm succeeds using $O(F h_2(q)+ q n h_2(p) - Fw h_2(\frac{1-q}{w}) + {nqp}) = O(F (h_2(q) - w h_2(\frac{1-q}{w})) + nq(h_2(p) + p))$ tests by Corollary~\ref{cor:concentr}. This significantly improves the upper bound provided in \cite{nikolopoulos2021group2} as $Fq( \log(F) + M) \gg F (h_2(q) - w h_2(\frac{1-q}{w}))$, when fixing $q$ and $F \rightarrow \infty$, and $nqp \cdot \log n \gg nq(h_2(p) + p)$ when $n \rightarrow \infty$.

In \cite{gonen2022group}, the authors consider a combinatorial version of group testing where each edge has size $d$ and show that there is an algorithm that recovers $e^*$ {with high probability} using at most $O(\log |E| + d\log^2 d)$ tests.
If we put probability $1/|E|$ on each edge, then Theorem~\ref{theorem:edgesize} tells us that we can recover the infected set with probability 1 using $O(H(X) + \mu) = O(\log |E| + d)$ tests in expectation, slightly improving the bound in \cite{gonen2022group}.

In \cite{ahn2023adaptive}, the authors consider a stochastic block infection model (SBIM) where there are $m$ communities and each one has $k$ nodes, so $n = mk$. Each node is a seed with probability $p$ independently. Then every seed node infects nodes in the same community with probability $q_1$ and $q_2$ for the nodes outside of the community, $q_1 > q_2$. The authors show lower and upper bounds on the number of tests. For the special case where $q_2 = 0$, they further show optimality of the proposed algorithm. {The authors show that their algorithm is order-optimal when,  $q_2 = 0$, $q = q_1$, 
$k p \ll 1$,   $q \ll \frac{1}{\sqrt{k \log (\frac{1}{k p})}}$, 
$k p \ll m^{-\beta} \text { for some fixed $\beta \in (0,1)$, and } 1 \ll k q$.}

 {In Appendix~\ref{app:provecoms}, we show that with only assuming $kp \ll m^{-\beta}$, the conditions of Corollary~\ref{cor:concentr} hold and, the infected set is recovered with $O(H(X) + \mu) $ tests in expectation.  In \cite{ahn2023adaptive}, the authors prove that when $q_2 = 0$ and  $q = q_1$,  $\mu \leq m \cdot k^2 \cdot p \cdot (1/k+q)$, and when $k p \ll 1 \text {, and } q \ll \frac{1}{\sqrt{k \log \left(\frac{1}{k p}\right)}},$ the entropy is given by $H(X) = m \cdot k \cdot p \cdot \log \left(\frac{1}{k p}\right)+m \cdot k^2 \cdot p \cdot q \cdot\left(\log k+\log \log \left(\frac{1}{k p}\right)\right)+1$. When $1/k \leq q$, which is a relaxed condition compared to $1 \ll kq$, we have $\mu \leq 2m \cdot k^2 \cdot p \cdot q$. Since $1 \ll kq$, it follows that 
$\mu \leq H(X)$ and hence $O(H(X) + \mu)$ is $ O(H(X))$. Now, Corollary \ref{cor:concentr}  provides an order-optimal performance guarantee by Theorem~\ref{theorem:lowerbound} and the fact that $\epsilon = o(1)$.}
{More generally, \cite{ahn2023adaptive} attains no analytical results, without having the following two conditions:  $k p \ll m^{-\beta} \text { for some fixed $\beta \in (0,1)$, and } 1 \ll k q$. But, even without these two conditions, we provide the following result due to  Theorem~\ref{theorem:final}:  Algorithm~\ref{mainalg} can recover the infected set with probability $1 - \epsilon$ using at most $O(H(X) + \mu/\epsilon)$ tests in expectation.}


\section{Lower Bound and Optimality of Algorithm~\ref{mainalg}}\label{sec:loweradapt}

{In the probabilistic group testing literature,  entropy $H(X)$ constitutes an order-wise tight lower bound \cite{li2014group} (see also Theorem~\ref{theorem:lowerbound}.) 
The upper bounds we established in Theorems \ref{theorem:edgesize}-\ref{theorem:final} and Corollary~\ref{cor:concentr} involve not only $H(X)$ but also $\mu$, which is the expected number of infections. To shed light on the extent to which our proposed algorithm is order-optimal and our analytical upper bounds are tight, we seek to answer the following questions. (i) Is entropy always a tight lower bound in the generalized group testing problem with general correlation?  (ii)  Does $\mu$ play a fundamental role (lower-bounding the number of tests) or should one hope to tighten it in our results? (iii) Is the upper bound established in Theorem~\ref{theorem:final} and/or its corresponding algorithm order-wise optimal?}

{In this section, we present a family of graphs in which entropy is not a tight lower bound in Section~\ref{sec:looseent}, provide examples for when $\mu$ serves as a lower bound and when it does not in Sections~\ref{sec:loosemu} and \ref{sec:tightmu}, and derive conditions under which Algorithm \ref{theorem:final}  is order-optimal for random hyper-graphs in Section~\ref{sec:randomgraph}}.

\subsection{When Entropy is a Loose Lower Bound}\label{sec:looseent}

By Theorem~\ref{theorem:lowerbound}, we know that $H(X)$ is a lower bound on the number of tests. {Consider the following example  from \cite[Theorem 4]{gonen2022group} with some slight modifications to make the setting probabilistic.}

\begin{Example}\label{app:examplehighprob}
Graph $G=(V,E)$ has $n$ nodes where each edge of size $n-1$ is included in $E$. To each edge $e$, assign a uniform probability mass $p(e) = \frac{1}{n}$.
Figure~\ref{fig:graphn} demonstrates such a graph with $n = 4$.
\end{Example}

In the above example, the entropy is $\log n$. What makes this example important for us is that any algorithm with a constant probability of success requires at least $\Omega(n)$ tests. This can be proved using a very similar argument as in \cite[Theorem 4]{gonen2022group}, and for the sake of completeness, we provide a short proof next.  What is significant here is that the gap between the entropy and the required number of tests is $O(n/\log n)$, which is much larger than expected.

\begin{claim}
\label{AppendixB}
    Any algorithm with a constant probability of success with the hypergraph in Example~\ref{app:examplehighprob} requires at least $\Omega(n)$ tests. 
\end{claim}

\begin{proof}
   Since one edge is a target edge and each edge consists of $n-1$ nodes, there is only one negative node and the rest of $n-1$ nodes are all positive. Thus, identifying the target edge is equivalent to identifying the negative node. Also, since exactly one node is negative, every test of size two or greater will yield positive results, rendering them redundant.  Consequently, the negative node can be identified only by individual testing until the negative node is located. Every node is equally likely to be the negative node. 
   Suppose $k$ tests are performed. Then, the probability that the target negative node is identified is $1-\frac{n-1}{n} \cdot \frac{n-2}{n-1} \cdots \frac{n-k}{n-k+1} = 1- \frac{n-k}{n} = k/n$. Thus,  $k = \Omega(n)$, for this probability to be constant.  

\end{proof}

{We can generalize the above example as follows:}
\begin{Example}\label{ex:62}
    Graph $G=(V,E)$ has $n$ nodes where at least $r$ fraction of all edges of size $d$ is included in $E$, for a constant $r \in (0,1]$. 
\end{Example}

The proof of Claim \ref{AppendixB} can be slightly generalized to prove the following: 

\begin{claim}
\label{app:ex62}
There is a distribution over the edges of Example~\ref{ex:62} such that $H(X)$ is $o(d)$, specifically  $O(\log d),$ and any algorithm performing on this distribution, achieving a constant probability of success, needs to perform $\Omega(d)$ tests.
\end{claim}

This again shows that the entropy is not a tight lower bound. We will extensively use Example \ref{ex:62}, Claim \ref{app:ex62} and its proof presented below to obtain several other lower bounds on the number of tests in this section. 

\begin{proof}
Consider a subset $S$ with $d+1$ nodes. It has $d+1$ subsets of nodes of size $d.$ On an average $r$ fraction of these belongs to $E.$ Thus the average number $E(S)$ is $r(d+1).$ 

Now let the distribution $\mathcal{D}$ be the uniform distribution over the edges in $S$ and 0 elsewhere. Thus, no algorithm needs to test nodes that are not in edges fully in $S.$ Next, for every edge in $S$, there are $d$ nodes in it, and there is exactly one node that is in $S$ but not in the edge; the edge is the target edge if that node is negative and all $d$ nodes in the edge are positive, the $d$ nodes in the edge cover rest of the nodes in $S$. Thus since there is exactly one target edge in $S$, there is exactly one node in $S$ that is negative and the rest of the nodes in $S$ are positive. The number of such potentially negative nodes equals the number of the edges in $S$, which equals $r(d+1)$ in expectation. Determining which edge is the target edge is equivalent to determining which of the candidate nodes is negative. Thus, an algorithm needs to only test the potentially negative nodes. Since only one of the candidate nodes is negative, testing two or more candidate nodes will give a positive outcome. Thus, the candidate nodes must be tested individually to determine which one is negative. Each of the candidate nodes is equally likely to be negative. Thus, arguing as  in the proof of Claim \ref{AppendixB}, 
any algorithm needs to perform $\Omega(m)$ tests for a constant probability of success, where $m$ is the number of potentially negative nodes. Note that $m$ is $r(d+1)$ in expectation. Considering $c$ as a constant, any algorithm needs to perform $\Omega(d)$ tests in expectation for a constant probability of success.  Also, $H(X) = O(\log d)$. 
\end{proof}

{Before proceeding further, suppose we run Algorithm~\ref{mainalg} on the above examples. It is worth noting that the algorithm bypasses the initial stage and proceeds to test nodes individually in line~\ref{alg:finaltest}, requiring $n$ tests in Example~\ref{app:examplehighprob} when $1/n < c < 0.5$ \footnote{To see this note that in Example~\ref{app:examplehighprob} each node belongs in all but one edge. Thus for each $v$, $w(V \setminus \{v\})$ is the weight of the only edge containing $v$ which is $1/n$. Thus,  $w(V \setminus \{v\}) < c < 1-c$  for each $v.$ And, the algorithm does not satisfy conditions in lines~\ref{alg:removebig},~\ref{alg:removesmall}, but satisfies the condition in line~\ref{alg:setnotfound} with $S = V$ for each $v.$ Thus, an empty set is tested in line~\ref{alg:setnotfound}, and testing an empty set is considered negative. Thus, every node is tested in line~\ref{alg:finaltest}.} and $d + 2$ tests in Example~\ref{ex:62} when $ 1/d < c < 0.5$ and under the distribution mentioned in the the proof of Claim~\ref{app:ex62}\footnote{Consider the probability distribution in the proof of Claim~\ref{app:ex62}) and $ 1/d < c < 0.5.$ We now argue that Algorithm~\ref{mainalg} on  example \ref{ex:62} bypasses the initial stage and proceeds to test nodes individually in line~\ref{alg:finaltest}, requiring $d+2$ tests overall. To see this note that each edge in $S$ has $d$ nodes, $S$ has $d+1$ nodes, thus each node in $S$ belongs in all but one edge in $S.$ Also, $q_v = 0$ for each node outside $S.$ Thus for each $v \in S$, $w(V \setminus \{v\})$ is the weight of the only edge containing $v$ which is $1/d$. Thus,  $w(V \setminus \{v\}) < c < 1-c$  for each $v\in S$. Also,  $w(V \setminus \{v\}) = 1 > 1-c$ for each $v\in V \setminus S.$ And, the algorithm does not satisfy conditions in line~\ref{alg:removebig} with any $v.$ It removes all nodes in $V \setminus S$ from $V$ in line~\ref{alg:removesmall}, and $S$ consists of the remaining nodes.  Now, the condition in line~\ref{alg:setnotfound} is satisfied and the nodes outside $S$ are tested together. These nodes test negative. Thus, every node in $S$ is tested in line~\ref{alg:finaltest}. Thus, there are $d+2$ tests in all.}. These examples highlight scenarios where conducting individual tests is optimal, leaving no room for { improving the theoretical guarantees in Theorem \ref{theorem:final} and the following corollaries.}
}

These examples raise the question of finding other lower bounds or identifying classes of graphs where entropy is not a tight lower bound. We generalize Example \ref{ex:62} for a larger family of graphs where $H(X)$ is a loose lower bound and at least $\mu$ tests are needed.

Consider  $n = |V|$, and let $m$, $d = |e|$ be given parameters. We build a graph with $n$ nodes and $|E| = m + O(d)$ edges where each edge has size $d$. Let the distribution $\mathcal{D}$ be uniform over the edges. Here, $H(X) = \log(m + O(d))$ and $\mu = d$. The graph is constructed in the following manner, where we start with an empty set for the edges and gradually add edges:
\begin{Example}
Graph $G = (V,E)$ and the distribution $\mathcal{D}$   are constructed as following:

\begin{enumerate}
    \item Initiate $E = \emptyset$.
    \item While ($|E| < m$): find a set $S \subseteq V$ with $|S| = d + 1$ such that {there is a subset of $S' \subset S$ of size $|S'| = d$ that is not currently in $E$. Add all subsets of $S' \subset S$ of size $|S'| = d$ to $E$, i.e.}
    
    $$E = E \bigcup\limits_{S' \subset S, |S'| = d}\{ S' \}.$$ 
    
\label{graphentropy:2}
    \item Set $\mathcal{D}$ to be the uniform distribution over $E$.
    
\end{enumerate}
    
\end{Example}

The first time that $|E|$ exceeds $m$,   $|E| \leq  m - 1 + d+ 1 = m+d.$ We claim that at least $\Omega(d)$ tests are needed to recover the infected set. Suppose $e^*$ is the sampled edge, so at some point in line~\ref{graphentropy:2} $e^*$ was added to $E$. Hence there is a set $S$ of size $d+1$ that contains $e^*$ and there are $d$ edges inside of $S$. Now even if the algorithm knows the set $S$, {using a similar argument as in \cite{gonen2022group} and proof of Claim~\ref{app:ex62}} for Example \ref{ex:62}, 
the algorithm still needs $\Omega(d)$ tests as only one node is negative and it is equally likely to be any of them.

Now if $\mu = d > \log (m+d) \geq H(X)$, then the entropy lower bound is loose and $\mu$ is an order-wise tight lower bound. This happens when $2^d > m + d$.

\subsection{When $\mu$ is not a Lower Bound}\label{sec:loosemu}
In all of the examples provided so far, either $\mu$ or $H(X)$ is a tight lower bound on the number of tests. One might conjecture that $\max (\mu, H(X))$ might be a lower bound on the number of tests. {The following example shows that $\mu$ is not necessarily a lower bound.}

\begin{Example}\label{ex:nested}
    
    {Consider a graph $G=(V,E)$ comprising $n$ nodes and the edges are of the form $e_i = \{v_1, \ldots, v_i\}$ and $p_{e_i} = 1/n$ for each $1\leq i \leq n$.}
\end{Example} 
\begin{figure}
    \centering
    \includegraphics[scale = .25]{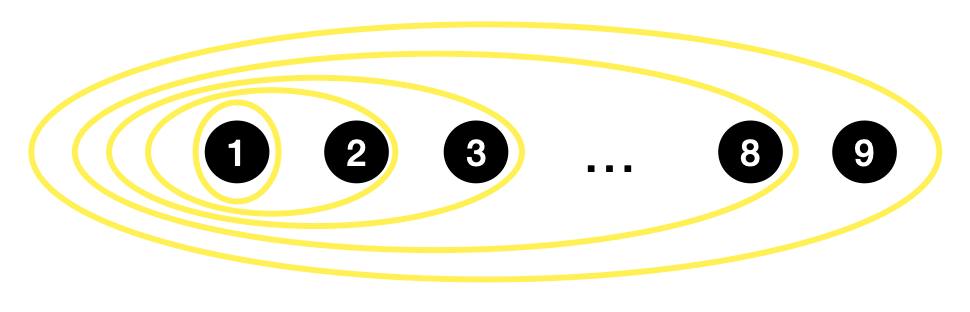}
    \caption{An instance of the graph from 
    Example~\ref{ex:nested} with $n = 9$ nodes.}
    \label{fig:linegraph}
\end{figure}

{Here, we have $\mu = \frac{n+1}{2} = \Omega(n)$. But by a binary search algorithm, we can find the target edge in $O(\log n)$. More precisely, we first test $v_{n/2}$. If the result is positive, it means that the target edge is among $e_i, i\geq n/2$ and the next step is to test $v_{3n/4}$. If the result is negative, then the target edge is among $e_i, i<n/2$, and the next step is to test node $i/4$. By continuing this, we can find the target edge in $O(\log n)$ tests.}

{Note that for the above example, $H(X)$ is a tight lower bound. The following example provides a graph where not only $\mu$ is not a lower bound, $H(X)$ is a loose lower-bound too.}

\begin{Example}\label{ex:biggraph}
    Consider a graph with vertex set $V = V_1\cup V_2 \cdots \cup V_n$ where $|V_i| = n$, so $|V| = n^2$ and we call each $V_i$ a community. Community $V_i$ is composed of nodes $V_{i,j}$ for $j=1,\ldots,n$.  The edge set is composed of ``small edges'' and ``large edges". The set of small edges is defined by $E_1 = \{e \mid \forall j,k: e = V_j \setminus \{V_{j,k} \}  \}$. So each small edge is contained in one of the $V_i$s and is obtained by removing exactly one node from a $V_i$. Thus, each $V_i$ has $n$ small edges in it, each of size $n-1$, there are therefore $n^2$ small edges. The set of large edges is defined  by $E_2 = \{e \mid \forall i: e = V \setminus V_i \}$. Thus, each large edge is obtained by removing all nodes in one $V_i$ from $V.$ Thus, each large edge has $n^2 - n = n(n-1)$ nodes in it, and there are $n$ large edges.  Here, $E = E_1 \cup E_2$. The distribution $\mathcal{D}$ puts half of the mass on each of $E_1, E_2$, and the edges in $E_1$ (respectively $E_2$) have the same probability among themselves. Figure~\ref{fig:graphnbig} shows an example with $n = 4$. 
\end{Example} 

\begin{figure}
    \centering
    \includegraphics[scale = .15]{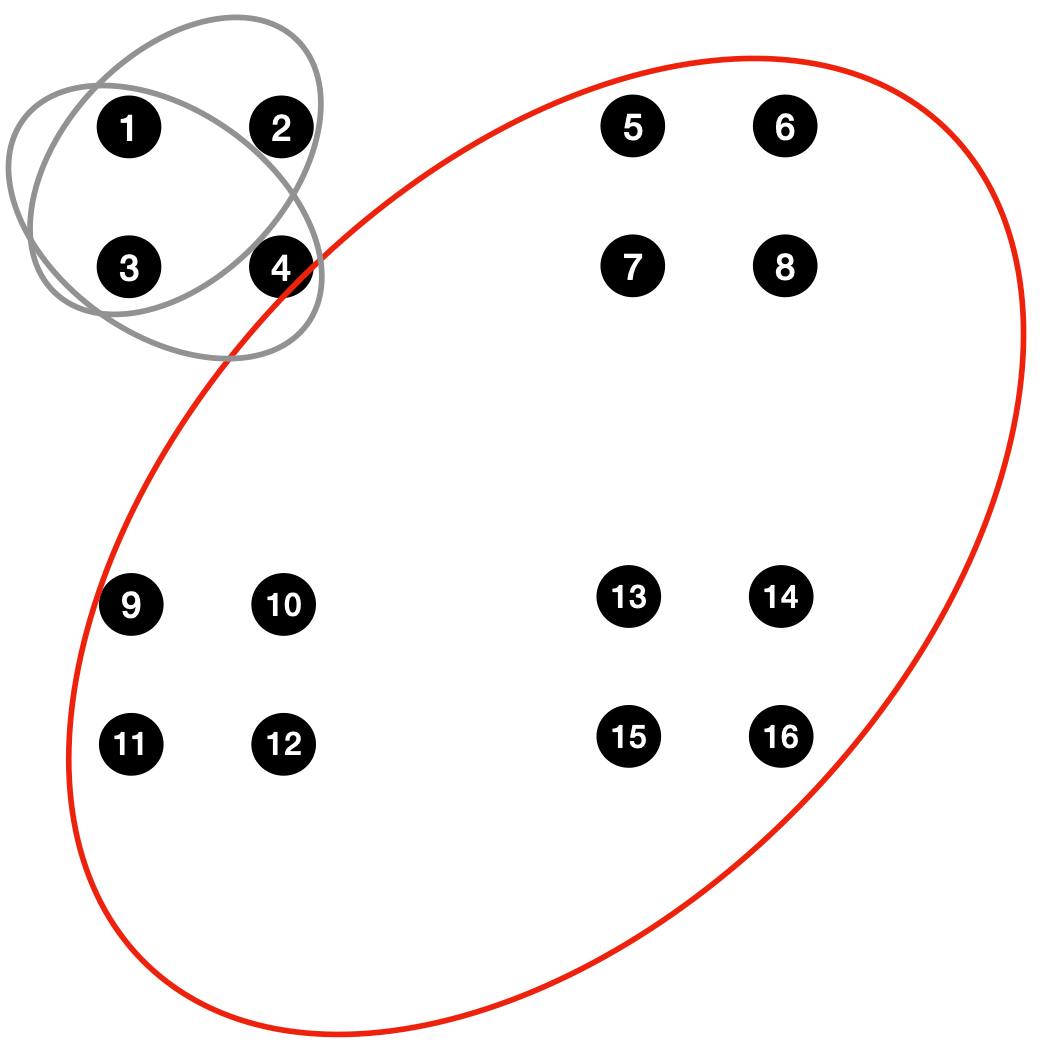}
    \caption{The graph in Example~\ref{ex:biggraph} with $n = 4$, 16 nodes, and 20 edges. Two of the``Small'' edges are shown in gray, one contains $\{1,2,3\}$, and the other contains $\{1,3,4\}$. An example of ``large'' edges is give with red border that contains $\{5,6,\ldots,15,16 \}$. There are 16 small edges and 4 large edges in total.}
    \label{fig:graphnbig}
\end{figure}

In the above example, $\mu = \Theta(n^2)$ and $H(X) = O(\log n)$.
{Following the proof of Claim \ref{app:ex62} for Example \ref{ex:62},   } it is not hard to see that $\Omega(n)$ tests are needed, because even after knowing whether the sampled edge is in $E_1$ or $E_2$, the algorithm needs to find the single negative node (in some community $V_i$) to recover the sampled edge. Thus, $H(X)$ is a loose lower bound. We next show that $O(n)$ tests are sufficient showing that $\mu$ is not a lower bound at all. 

To see that $O(n)$ tests are also sufficient, note that if the infected edge is in $E_2$, all nodes except those in one $V_i$ will be positive, and no node in that $V_i$ is positive. So, one can pick a single node from each $V_1, V_2, V_3, \ldots$ and test the candidates individually. If at least two of them are negative, it means the sampled edge is not in $E_2$, otherwise, it is in $E_2$. If the sampled edge is not in $E_2$, it is in $E_1$. This phase consumes at most $n$ tests. If the sampled edge is in $E_1$, then all nodes are negative, except for those in one $V_i$, and only one node is negative in the exceptional $V_i$ and the target edge is in the exceptional $V_i$. The exception can be determined by identifying $2$ random nodes in each $V_i$, and testing each such pair as a group. Only the pair for the exception will test positive, all others would test negative. Thus, the exception can be identified in additional $n$ tests. Now, the target edge can be identified by testing each node in the exception individually until the negative node is located, and all nodes in the exception except the negative node constitutes the target edge. This requires an additional $n$ tests. Thus, overall, if the target edge is in $E_1$, it can be identified in at most $3n$ tests.   If the target edge is in $E_2$, we can find the negative community using at most $n$ additional tests by individually testing one randomly selected node from each community. Each such test will be positive except for the candidate node from the negative community. Thus, if the target edge is in $E_2$, it can be identified in at most $2n$ tests.  In either case, $O(n)$ tests is sufficient.

The above examples show that the lower bound can not necessarily be obtained by simple properties of the graph, such as $\mu$ and entropy. Finding a tight lower bound constitutes a future direction.

\subsection{When $\mu$ is a Lower Bound for Achieving Zero Probability of Error} \label{sec:tightmu}
We have already established that $H(X)$ is a lower bound. For any graph, if we show that $\mu$ is also a lower bound, then Algorithm~\ref{mainalg} is order-wise optimal. One obvious example of such graphs is when $\mu \leq H(X)$, where Algorithm~\ref{mainalg} would be order-wise optimal. 

{In this section, we provide a family of graphs where $\mu$ is not necessarily smaller than the entropy, but we prove that $\mu$ is a lower bound when the target probability of error is 0. The example is as follows:}

\begin{Example}\label{ex:muopt}
Graph $G$ has this property that for every node $v$ and every realization of states of $V \setminus v$, there is a positive chance that $v$ is infected and there is a negative chance that $v$ is not infected. Thus, for every subset of nodes there is a positive probability that only the nodes in it are positive and the nodes outside are negative, hence, every subset of nodes is a target edge with a positive probability and therefore belongs in $E$.   
\end{Example}

\begin{claim}
Any algorithm on Example~\ref{ex:muopt} with zero probability of error needs at least $\mu$ tests in expectation.
\end{claim}
\begin{proof}
Let $e^*$ be the target edge. Then, for zero probability of error, $e^*$ has to be detected as the target edge. Thus, every node in $e^*$ has to be detected as positive. Thus, every node in $e^*$ has to be tested at least once alone or along with other nodes as its state can not be definitively inferred even if states of all other nodes in $V$ are known, from the condition in the Example. We next argue that every node in $e^*$ has to be tested at least once alone or in a group in which no other node in $e^*$ is present. This is because all other nodes in $e^*$
are positive, and if the considered node in $e^*$ is tested only in company of other positive nodes, the outcome of all such tests would be positive regardless of the state of the considered node. Thus, since each node in $e^*$ is tested at least once without any other node in $e^*$ being tested in the same test, the number of tests must be at least $|e^*|.$ The result follows because $E|e^*| = \mu.$
\end{proof}


\subsection{Random Hypergraphs and Regimes of Optimality}\label{sec:randomgraph}
{Consider $d$-regular random hypergraphs, (ie, where all edges have $d$ nodes, hence $\mu = d$), and each subset $S \subseteq V$ of size $|S| = d$ is included in it with probability $r$, independent of the rest, to make the graph $G_r = (V, E_r)$. 
We analyze two regimes of interest by examining the upper and lower bounds in each case. The first regime is when $r$  is large and the graph is dense (eg, $r$ is $\omega((d/n)^{1/d})$),  the other is when $r$ is small (eg, $O((d/n)^{\Theta(d)})$) and the graph is sparse. } 

\begin{proposition}
    Suppose graph $G_r = (V,E_r)$ is given with the parameter $ r \gg (\frac{d+1}{n})^{1/(d+1)}$. Then there is a distribution $\mathcal{D}$ on $E_r$ so that any algorithm with a constant probability of success needs $\Omega(\mu) = \Omega(d)$ number of tests on average.  
\end{proposition}

The Proposition and its proof are based on the following insight. Since $r$ is large and the hypergraph $G_r$ is dense, enough number of subsets of size $d$ from $V$ make it to $E_r$, such that with a high probability there exists at least one set $S \subset V$ of size $d+1$ such that all subsets of size $d$ of $S$ constitute edges in $E_r.$ Now, Example \ref{ex:62}, Claim \ref{app:ex62} and its proof, show that for such a graph, at least $\Omega(d)$ tests are needed to identify the target edge with a constant probability of success, if the probability distribution for the target edge is uniform over the edges of size $d$
in $S$.

\begin{proof}
    Suppose that $r \gg (\frac{d+1}{n})^{1/(d+1)}$. We partition the nodes into subsets of size $d+1$, so there are $\frac{n}{d+1}$ sets. Each such subset has $d+1$ edges of size $d$ each. The probability that all the edges of size $d$  in a given subset are  included in $G_r$ is $r^{d+1}$, so with probability $(1-r^{d+1})^{\frac{n}{d+1}}$, no such subset has all its edges sampled. Note that  $(1-r^{d+1})^{\frac{n}{d+1}} \leq e^{-r^{d+1}{\frac{n}{d+1}}}.$ Denote $e^{-r^{d+1}{\frac{n}{d+1}}}$ as $\delta$. Hence, with a probability $1-\delta$ there is at least one subset of $d+1$ nodes, say $S$,  where every subset of size $d$ in it is an edge in $G_r$. We first show that $\delta$ is a small positive number.  Let $r' = r^{d+1} \frac{n}{d+1}.$ Then,  $ \delta = e^{-r^{d+1}{\frac{n}{d+1}}} = e^{-r'}$, and $\frac{r}{(\frac{d+1}{n})^{1/(d+1)}}=  r'.$  Since $ r \gg (\frac{d+1}{n})^{1/(d+1)}$, $r' \gg 1.$ Thus, $e^{-r'}$ is a small positive number.

     
    Consider that there is such a subset $S$ as described in the first paragraph. 
    {Call the edges inside $S$ of size $d$, $e_1, e_2, \ldots, e_{d+1}$, where $\forall i: e_i \in E_r$ and $e_i \subset S$ and $|S| = d+1$.} 
    {Now define the distribution $\mathcal{D}$ such that it only puts weights on $e_1, e_2, \ldots, e_{d+1}$, uniformly at random, i.e. $\forall i: \mathcal{D}(e_i) = \frac{1}{d+1}$.} 
    {Following the line of argument in  proof of Claim \ref{app:ex62} for Example \ref{ex:62}, we conclude that any algorithm with a constant probability of success needs  $ \Omega(d)$ number of tests in expectation. Note that this holds if such an $S$ exists which happens with probability $1-\delta$. If such a $S$ does not exist, even allowing that $0$ tests are needed to ensure constant probability of success, the overall expected number of tests is $(1-\delta) \Omega(d) $, which is also $\Omega(d).$ Also, $\mu = d$, so every such algorithm needs $\Omega(\mu)$ tests in expectation.
    } 
\end{proof}

{Note that Algorithm~\ref{mainalg} performs $O(H(X) + \mu) = O(H(X) + d)$ tests in the above regime by Theorem~\ref{theorem:edgesize}, and we already know that $H(X)$ is a lower bound (Theorem \ref{theorem:lowerbound}). The above proposition shows that $d$ is also a lower bound, in the asymptotic order sense,  proving that Algorithm~\ref{mainalg} is order-wise optimal in this regime.}

We now consider the sparse regime with small $r:$
\begin{proposition}
\label{Prop2}
    Suppose the random hypergraph $G_r = (E_r, V)$ is given with the parameter $r$ as $O((d/n)^{c_rd})$, where  $c_r = \frac{2}{c'(1-2c)},$ $c$ is the constant parameter defined in Algorithm~\ref{mainalg}, and $c'$ is any constant positive integer. Then, there is a modified version of Algorithm~\ref{mainalg} that finds the infected set using $O(H(X))$ tests on average when $d$ is  $o(n)$, and $H(X)$ is $\Omega(1)$.
\end{proposition}

Proposition \ref{Prop2} and its proof are based on the following insight. Algorithm~\ref{mainalg} uses $O(H(X) + d)$ tests in expectation for a $d$-regular hypergraph because it tests nodes of $S$ individually in line~\ref{alg:finaltest} and $S$ has $O(\mu)$ nodes in expectation just before line~\ref{alg:finaltest} and for a $d$-regular Hypergraph $\mu = d$. But for a sparse $d$-regular random hypergraph, with a high probability $S$ has only a constant number of edges of $E_r.$ This is because any set of cardinality $O(d)$ has $O((n/d)^{O(d)})$ subsets of cardinality $d$, and if $r$
is $O((d/n)^{O(d)})$, then only a constant number of these subsets are in $E_r$ in expectation.  Thus, $|E(S)|$ is constant in expectation. But not all sets of cardinality $O(d)$ has the target edge in it. Before its line~\ref{alg:finaltest}, Algorithm \ref{mainalg} identifies a set $S$ of cardinality $O(d )$ which has the target edge in it. Thus, instead of testing nodes in $S$ individually, the edges in $E(S)$ can be tested to determine which of them is the target edge. This would need a constant number of tests in expectation in line~\ref{alg:finaltest} instead of $O(d)$ tests therein, if there exists a way to directly test an edge. Note that in general it is not possible to conclude that an edge is a target edge by testing all its nodes together as a group, because the outcome will be positive even if one node in it is positive; thus one group test can not determine if all nodes in it are positive. But this can be accomplished by testing the complement of each edge in $E(S)$ utilizing the fact that $G_r$ is a $d-$regular graph. Thus, $d$ in the upper bound on the number of tests executed by Algorithm~\ref{mainalg} can be replaced by a constant and the Proposition follows as $H(X)$ is at least a constant. 

\begin{proof}
         Consider a modified version of Algorithm~\ref{mainalg} catering specifically to regular hypergraphs (ie, where all edges have equal sizes).  In line~\ref{alg:finaltest}, the modified version first checks  if the number of remaining edges in $S$ is fewer than $|S|$. If it holds, the algorithm tests each $S \setminus e$ for each remaining edge $e \in E(S)$ and outputs $e$ as the target edge the first time one such test is negative. If the condition does not hold, then it tests nodes in $S$ individually and outputs the positive nodes as the target edge.  Note that the target edge has to be one of the remaining edges in $S$ as edges are removed  only when it is established that they can not be the  target edge. Thus since all positive nodes are in the target edge, the test of $S \setminus e$ will be negative for the target edge $e.$ If $S \setminus e$ is negative,  either $e$ or a proper subset thereof is the target edge. But since all remaining edges are of equal size,  the target edge can not be a proper subset of any remaining edge. Hence, the modified algorithm concludes that $e$ is the target edge. Note that the modified algorithm needs at most $\min(|E(S)|, |S|)$ tests in the last step. 
         

        We now specifically consider the $d$-regular random hypergraph $G_r.$ Recall that the original Algorithm \ref{mainalg} performs at most $\frac{\tilde{\mu}}{1-2c}$ tests in line~\ref{alg:finaltest} (Lemma \ref{obs:L2}), where 
    $\tilde{\mu}$ is the expected number of infected nodes at the start of line~\ref{alg:finaltest}. Since each edge consists of $d$ nodes, so does the target edge, and thus, there are exactly $d$ infected nodes in the network, so $\tilde{\mu} = d.$ Also, the original Algorithm \ref{mainalg} performs exactly $|S|$ tests in line~\ref{alg:finaltest}. So $|S| \leq \frac{d}{1-2c}.$ Since the modified algorithm needs at most $\min(|E(S)|, |S|)$ tests in the last step, the expected number of tests in the last step is $O(d)$. So if $d = O(1)$, the expected number of tests in line~\ref{alg:finaltest} is also $O(1).$ We prove that even when  $d$ is $ \omega(1)$, the expected number of tests in line~\ref{alg:finaltest} is  $O(1).$ Then the expected number of tests by the modified algorithm is 1) $O(H(X))$ until before line~\ref{alg:finaltest} (from Corollary \ref{corol:l1} since  the modified and original versions differ only in line~\ref{alg:finaltest}) 2) at most $O(1)$ in line~\ref{alg:finaltest}. Since $H(X) = \Omega(1)$, the total number of tests is  $O(H(X))$ on average. The Proposition follows.

    In the rest of the proof we assume that $d$ is $\omega(1)$, and using the fact that every edge has $d$ nodes, we prove that with probability $1-\epsilon_n$, where $\lim_{n \rightarrow \infty} \epsilon_n = 0$, the  number of edges of $E_r$ in any subset of $V$ of cardinality at most $\frac{d}{1-2c}$ is less than any given constant positive integer $ c'$ when $r$ is small and $d $ is $o(n).$ Thus, with probability $1-\epsilon_n$,  $|E(S)| < c'$, when the original Algorithm \ref{mainalg} gets to line~\ref{alg:finaltest}, and the same holds for the modified version as they differ only in line~\ref{alg:finaltest}. Thus,  with probability $1-\epsilon_n$, $c'$ is an upper bound on the number of tests the modified algorithm needs in this last step. Since the number of tests in this last step is always upper bounded by $|S|$, and $|S| \leq \frac{d}{1-2c}$, if there are $c'$ or more edges in $S$, which happens
    with probability at most $\epsilon_n$, we consider the upper bound to be  $\frac{d}{1-2c}$ instead.  Thus, the expected number of tests in this final step is at most 
    $c' + \epsilon_n (d/(1-2c)).$  We will prove that $\lim_{n \rightarrow \infty} \epsilon_n d = 0$ using the fact that $d$ is $o(n).$ Thus, since $c$ is a constant, in the limit that $n \rightarrow \infty$, the upper bound for the expected number of tests in line~\ref{alg:finaltest} is $c'$, which would complete the proof of the Proposition.

    We now prove that for any constant positive integer $c'$, with probability $1-\epsilon_n$, where $\lim_{n \rightarrow \infty} \epsilon_n = 0$, the  number of edges in any subset of $V$ of cardinality at most $\frac{d}{1-2c}$ is less than  $c'$. 
    First note  that if there does not  exist $S$ of a certain cardinality with at least $c'$ edges, there does not exist $S$ with a lower cardinality with at least $c'$ edges, because a set is always a subset of some other set with higher cardinality. Thus, the event that there does not exist a set with at least $c'$ edges, is a subset of the event that there does not exist a set of higher cardinality with at least $c'$ edges. Thus, the event that there does not exist a set with cardinality less than or equal to $\frac{d}{1-2c}$ with at least $c'$ edges is the event that there does not exist a set with cardinality equal to $\lfloor \frac{d}{1-2c} \rfloor$ with at least $c'$ edges.  We therefore upper bound the probability of the latter. Let $\gamma = 1/(1-2c).$ First, there are ${\lfloor \gamma d \rfloor \choose d}$ subsets of size $d$ in a given set of cardinality $\lfloor \gamma d \rfloor.$ Also, ${\lfloor \gamma d \rfloor \choose d} \leq (\frac{e \lfloor \gamma d\rfloor}{d})^d \leq (e\gamma)^d$ as ${n \choose k} \leq (\frac{en}{k})^k$, for each $n, k \leq n.$ Now, the probability that at least $c'$ such subsets are in $E_r$ is less than $r^{c'} {{\lfloor \gamma d \rfloor \choose d} \choose c'} \leq (re(\gamma e)^d/c')^{c'}$. There are ${n \choose \lfloor \gamma d\rfloor }$ candidates for picking up a subset of size $\lfloor \gamma d \rfloor $ from $V$, and $${n \choose (\lfloor \gamma d \rfloor)} \leq (\frac{en}{\lfloor \gamma d \rfloor})^{\lfloor \gamma d \rfloor} \leq (\frac{en}{d})^{\lfloor \gamma d \rfloor} \leq (\frac{en}{d})^{ \gamma d }$$ , since $d \leq n$, $d$ is $\omega(1)$. Hence by union bound, the probability that there exists $S$ of size $\lfloor \gamma d \rfloor $ with at least $c'$ edges is at most 
    \begin{align*}    
    (\frac{en}{ d})^{\gamma d} (re(\gamma e)^d/c')^{c'} = (\gamma e)^{c'd}(\frac{en}{d})^{\gamma d}(re/c')^{c'} = (e/(1-2c))^{c'd}(\frac{en}{d})^{d/(1-2c)}(re/c')^{c'},
    \end{align*}

    replacing $\gamma$ with $1/(1-2c)$. Let $\kappa = (e/(1-2c))^{c'(1-2c)}$, and $\nu = (e/c')^{c'}$.
    Then, we get $
    (e/(1-2c))^{c'd}(\frac{en}{d})^{d/(1-2c)}(re/c')^{c'} = \nu
    (\kappa \frac{n}{d})^{d /(1-2c)}(r)^{c'}$. 
    When $r \leq (\frac{d}{\kappa n})^{\frac{2d}{c'(1-2c)}}$, the right hand side is at most $\nu  (\frac{d}{\kappa n})^{\frac{d}{c'(1-2c)}}= \epsilon_n $ where $\epsilon_n$ goes to zero as $n \rightarrow \infty$ since $d$ is $o(n)$ and $d$ is $\omega(1)$. Note that the upper bound we considered for $r$ is $ O((d/n)^{c_rd})$, where $c_r = \frac{2}{c'(1-2c)}$.

    We now prove that $\lim_{n \rightarrow \infty} \epsilon_n d = 0$. We have $$\epsilon_n d = \nu  (\frac{d}{\kappa n})^{\frac{d}{c'(1-2c)}} d =  \nu  (\frac{d^{1+\frac{c'(1-2c)}{d}}}{\kappa n})^{\frac{d}{c'(1-2c)}}.$$ 
    As  $d^{1/d} < 2$, $d^{\frac{c'(1-2c)}{d}} < 2^{\frac{c'(1-2c)}{d}} < 2^{c'},$ we get $ \epsilon_n d < \nu (\frac{2^{c'}d}{\kappa n})^{\frac{d}{c'(1-2c)}}$. Since $d$ is $o(n),$ $\lim_{n \rightarrow \infty} \epsilon_n d = 0$.

\end{proof}

The parameter $c_r$ controls the degree of sparsity of the hypergraph considered in Proposition \ref{Prop2}. Specifically,  $\lim_{n \rightarrow \infty} \mathbb{E}[|E_r|]$ can range from infinity to $0$, depending on the value of $c_r$.  To see this note that the number of subsets of $V$ of cardinality $d$, ie, $n \choose d$, is at least $(n/d)^d.$ The expected number of edges in $E_r$ is therefore at least $r(n/d)^d$, which can be $\Theta((n/d)^{(1-c_r)d})$ if $r$ is $\Theta((d/n)^{c_rd})$.  If $0 < c_r < 1,$
 $(n/d)^{(1-c_r)d} \rightarrow \infty$ as $n \rightarrow \infty $, since $d$ is $o(n).$ Thus $\lim_{n \rightarrow \infty} \mathbb{E}[|E_r|] = \infty.$ Since $c_r = \frac{2}{c'(1-2c)},$ $ 0 < c < 1/2$, for all  $c' > 2/(1-2c)$, $0 < c_r < 1.$ Recall that $|E(S)|$ is constant in expectation for the $S$ identified by Algorithm~\ref{mainalg} and its modification before their respective final steps, even though 
 $\lim_{n \rightarrow \infty}  \mathbb{E}[|E_r|] = \infty.$ This can be ensured only because $S$ has $O(d)$ (instead of $n$) nodes, and for each $S \subset V$, with $O(d)$ nodes, $|E(S)|$ is constant in expectation, and before their final steps Algorithm~\ref{mainalg} and its modification identify an $S$ with $O(d)$ nodes which has the target edge, and therefore in the last step the algorithm can solely focus on $S$ and the edges therein. But, if the hypergraph becomes even sparser, ie, $c_r$ higher, leading to lower $r$, then $|E_r|$ itself becomes a constant or goes to zero in expectation. For example, if $c_r > 1$,  (ie, $1 < c' < 2/(1-2c)$), then since there are $O((n/d)^{d})$ subsets of size $d$ of $V$, then $\mathbb{E}[E_r] = O((n/d)^d (d/n)^{c_r d}) = O((d/n)^{(c_r - 1)d}) \rightarrow 0$, as $n \rightarrow \infty$ (since $d$ is $o(n)$, $c_r > 1$). In other words,  the hypergraph is empty, and no node is infected. No algorithm is meaningful in this case.  

Note that the proof of Proposition \ref{Prop2} shows that if  $r$ is $O((d/n)^{c_rd})$, $c_r = \frac{2}{c'(1-2c)},$ $c$ is the constant parameter defined in Algorithm~\ref{mainalg}, then  $|E(S)| < c'$ with a high probability for each $S \subset V$, with $O(d)$ nodes. For large $c'$, ie, $c' > 2/(1-2c)$, the requirement is lax and can be met even when the $d$-regular random hypergraph has a large edge set, ie, $\lim_{n \rightarrow \infty} \mathbb{E}[|E_r|] = \infty.$ Thus, $r$ does not need to be very low in this case, nor $G_r$ extremely sparse, which corresponds to $0 < c_r < 1.$ But for small $c'$, ie, $c' < 2/(1-2c)$, the requirement is stringent and can be met only when the $d$-regular random hypergraph has a vanishing edge set, ie, $\lim_{n \rightarrow \infty} \mathbb{E}[|E_r|] = 0.$ Thus $G_r$ needs to be extremely sparse, actually nearly empty,  to start with, ie, $r$ very low, which happens for high $c_r$, ie, $c_r > 1.$
 
Finally, if $d$ is $\alpha n$, instead of $o(n)$,  then again the total number of subsets of size $d$ of $V$ is small enough so that the hypergraph is empty regardless of the value of $c_r.$ We delve into the details and justify various assumptions made in Proposition \ref{Prop2} next. 

\begin{remark}
    For the assumption in Proposition \ref{Prop2} that 
    $d$ is $o(n)$, 
    if we assume otherwise, i.e., $d$ is $\Theta(n)$, or $d = \alpha n$ for a constant $\alpha \leq 1$, then the random hypergraph becomes so sparse that no node is infected in the asymptotic limit on $n$. This can be seen as follows.  The total number of subsets of size $d$ is ${n \choose d} = {n \choose \alpha n} \leq (e/\alpha)^n$. Hence if $r$ is $O((d/n)^{c_r d})$ $r$ is $O((\frac{\alpha n}{n})^{(c_r \alpha n)}) = O((\alpha^{c_r \alpha})^{n})$. Since $0 < \alpha \leq 1$, as long as $c_r \geq 0$, $\alpha^{c_r \alpha} \leq 1 < e/\alpha$, then $\mathbb{E}(|E_r|)$ is at most $c_1^n$ for a constant $c_1 < 1$. Thus   $\mathbb{E}(|E_r|)$  goes to 0 as $n \rightarrow \infty$. Since $Pr(|E_r| \geq 1) \leq \mathbb{E}(|E_r|)$,  $Pr(|E_r| \geq 1)$ also goes to 0 as $n \rightarrow \infty$. Thus, the random hypergraph becomes empty  in the asymptotic limit on $n$. 

    The assumption $H(X) = \Omega(1)$ is natural, as if $H(X) = o(1)$, it means that the distribution for target edge is concentrated on a single edge and if the probability of error is larger than the probability that this edge is not the target edge, one can output this edge without any test. In other words, the assumption $H(X) = \Omega(1)$, limits considerations to instances where at least one test is needed to identify the target edge. 
\end{remark}

{Proposition \ref{Prop2} and Theorem \ref{theorem:lowerbound} together imply that the modified version of Algorithm~\ref{mainalg} is order-wise optimal for sparse $d$-regular random hypergraphs when $d$ is $o(n)$, and the distribution for the target edge is not excessively concentrated on one edge. Also, the analysis provided in Theorem~\ref{theorem:final} is the worst-case scenario and a slight modification of the algorithm renders the $\mu= d$ term in the upper bound of the theorem as redundant in many cases.}

\section{Towards Non-adaptive Group Testing}\label{sec:nonadapt}

{We have studied adaptive group testing so far. In Section \ref{sec:2} we  proposed SNAGT as a model in between adaptive and non-adaptive testing. To recall, in SNAGT, tests are designed beforehand and do not depend on prior test results, which is similar to non-adaptive. But, similar to adaptive, the tests are done sequentially and we know the result of a test just after performing it. The results of the previous tests can be used to determine if the tests should stop or continue.  We first propose a SNAGT strategy and  analyze its performance. Then we discuss how correlation between states of different nodes complicates the design of fully non-adaptive algorithms.}


We first obtain a SNAGT algorithm for building upon a non adaptive algorithm designed  in \cite{gonen2022group} for $d$- regular hypergraphs. The testing strategy in \cite{gonen2022group} is random and  each node of $V$ is picked to be in a test with probability $1/d$ independent of others. Every such test will be referred to as a random test in this section henceforth. Such tests are repeated and the selections in different tests are independent. After a predetermined number of tests, once the results are obtained, the edges that are not consistent with the result of any test, ie, those which can not be the target edge given the result, are removed. Specifically, if any test is positive then edges which do not have at least one node among the nodes tested in that test are removed; if the test is negative, edges which have at least one  node tested in it are removed. It has been shown in \cite{gonen2022group} that after $O(d (\log |E| + \alpha))$ tests, once the results are processed, the hypergraph contains only the target edge with probability $1 - e^{-\alpha}$. Setting $e^{-\alpha} = 1/(n)^2$ we get $\alpha =2 \log n$. Thus, after $O(d (\log |E| + \log n))$ tests, the hypergraph contains only the target edge with probability $1 - (1/n^2)$ which converges to $1$ as $n \rightarrow \infty.$ Now considering that nodes are selected randomly for tests exactly as in \cite{gonen2022group}, but that results of a test are available right after the test, we explore if the same probability of error can be obtained through a lower number of tests  by utilizing the 1) results to  potentially stop earlier (SNAGT framework) and 2) knowledge of probability distribution on edges specifying the probability that any given edge is a target edge.

{The core idea is as follows. We partition the edges into subgraphs so that the edge probabilities in the same subgraph are close to each other (ie, within a factor of $2$ of each other). This way, each subgraph constitutes a new group testing problem with nearly uniform probabilities where random testing, as discussed in the previous paragraph, is  efficient. This idea is similar to the non-adaptive construction in \cite{li2014group} with two major differences: (i) In \cite{li2014group}, nodes are partitioned, leading to several \textit{smaller} subproblems. In our settings, edges are partitioned and the resulting subproblems are not smaller problems-- each subproblem is still a  hypergraph on $n$ vertices. (ii) the construction in \cite{li2014group} is non-adaptive, and finding the set of infected nodes reduces to finding a set of infected nodes in each of the smaller subproblems. In our setting, a fully non-adaptive strategy requires one to search for the infected set in all of the subgraphs with a minimum number of tests that is governed by the worst-case subproblem and not the average. Nonetheless, a semi-non-adaptive strategy is more powerful because the decoder can perform sequential posterior update after each test result to find the infected edge among all the subproblems. Here, it turns out that the  minimum number of tests  is  governed by the average performance of the involved subproblems.} 

Note that a subgraph can have at most one target edge (as there is only one target edge in the system). Since the result is obtained right after a test, the result can be processed to eliminate edges that are not target edges. The target edge will never be eliminated  as it will be consistent with the result of every test. We refer to the subgraph with the target edge as the target subgraph, the rest as non-target subgraphs.  As edges that are not target edges are being progressively eliminated, and the target edge is never eliminated,  with a high probability, after a certain number of tests, the target subgraph will have only the target edge remaining. So, any  subgraph with only one remaining edge is shortlisted as a candidate subgraph for having the target edge, each such remaining edge is a candidate for being the target edge. But non-target subgraphs  may also be shortlisted as candidates per the above criterion, as they may be reduced to only one edge through the same process of elimination.   The number of candidates is at most the number of subgraphs. As tests continue, some of the non-target subgraphs  would become empty, because of elimination of their remaining edge (which is not a target edge).  If the number of tests exceed a certain threshold,  with a high probability, all non-target candidate subgraphs  become empty; then the non-empty candidate, if any, can be considered the target subgraph and its edge outputted as the target edge. If the number of subgraphs is not too large then the number of candidates won't be either and the number of tests by which only the target subgraph survives (with a high probability) will not be too large either. This is indeed the case if edges with probabilities within a factor of two of each other are grouped together into the same subgraph. But again if the number of subgraphs is small, at least some of them will have a large number of edges increasing the number of tests needed to render them as candidates. But, our partitioning ensures that the edges with high probabilities of being target edges are in subgraphs with few edges each, as there can only be a limited number of high probability edges and only edges with similar probability values are grouped together. And, a subgraph with few edges will be reduced to one edge in a few tests and therefore becomes a candidate in a few tests. In contrast, the subgraphs where edges have low probabilities can reduce to one edge  only in a relatively large number of tests as they can have large number of edges to start with. This quick identification of candidates among high probability edges reduce the overall number of tests, a phenomenon that suggests that high disparities between edge probabilities will yield bigger reductions in number of tests compared to that for the non adaptive algorithm in \cite{gonen2022group}. We will analyze this phenomenon later, but first start with the details of the partition. 

We partition the edges into subgraphs $G_1,G_2,\ldots, G_k$ where $G_i = (V,E_i)$ and $\forall e \in E_i: \frac{1}{2^i} \leq p(e) \leq \frac {1}{2^{i-1}}$. Denote $m_i = |E_i|$ as the number of edges in $G_i$, and $\mathbb{P}(G_i) = \sum_{e \in E_i} p(e)$ the probability that the target edge is in $G_i$. Clearly, $p(e)$ for edges in subgraph $G_i$ decreases with increase in $i$. For  $i > n \log n$, this probability becomes vanishingly small: $p(e) < 2^{-(n \log n)}$. Since there are at most $2^n$ edges overall, $m_i \leq 2^n$ for each $i$. Thus, for $i > n \log n$, $\mathbb{P}(G_i) \leq \frac{2^n}{2^{n \log n}}$, which converges to $0$ as $n \rightarrow \infty.$  We therefore consider that $k \leq n \log n.$ Since $\mathbb{P}(G_i) \leq 1$, and $p(e) \geq \frac{1}{2^i}$ for each $i$,  $m_i \leq 2^i$ for each $i$. Since $p(e) \leq 1/2^{i-1}$ for each $i,$ so $m_i \leq 2/p(e)$ for all $e \in E_i$. Thus, intuitively, size of a group increases with decrease in edge probabilities in the group; as such edges with highest $p(e)$ will be in subgraphs with fewest edges. 

Utilizing the above partition of edges, we design a SNAGT algorithm for $d$-regular hypergraphs next.  The principal challenge is in determining which candidate subgraph is the target subgraph in a reasonable number of tests. 1) We can not determine whether a candidate subgraph is the target subgraph from the timing of its candidature. This is because,  as mentioned before, non-target subgraphs may also become candidates, either before or after the target subgraph becomes.  Even a meaningful upper bound on the timing of candidature of the target subgraph is not known apriori. Note that it follows from  \cite[Theorem 4]{gonen2022group},  that if  $G_i$ is the target subgraph,  after $d (\log m_i + \alpha)$ random tests, $G_i$ will have only one edge left, the target edge, with a probability at least $1-e^{-\alpha}$, which is close to $1$ for large $\alpha.$ 
But, we do not know $i$ apriori, hence we do not know $m_i$ either.   2) For regular hypergraphs, testing the complement of any edge will reveal if it is the target edge or otherwise. But,  the SNAGT framework does not allow us to specifically test the remaining edge or its complement, for the candidate subgraphs. This is because a subgraph becomes a candidate only as a result of outcome of tests, so testing the remaining edge in it or its complement based on the criterion that it is the candidate, would be using the test results to determine who to test. Incidentally, even if we could rule out candidate subgraphs by testing the complement of their remaining edge, we would need up to $k \leq n \log n$ tests, which is very high. 3) As mentioned before, additional random tests can determine, with a high probability, if a candidate subgraph is the target subgraph, by eliminating its remaining edge if it is not. One can in fact prove that the probability that an edge that is not the target edge can survive $d \Omega (1)$ random tests,  approaches $0$ as $ n \rightarrow \infty.$  But, deploying up to $d \Omega(1)$ random tests after each subgraph becomes a candidate may need up to $(n \log n)(d \Omega(1))$ tests, which is  $\omega(dn)$, and the maximum number of tests that can be possibly needed by the non adaptive algorithm in \cite{gonen2022group} is $O(dn).$

We design an algorithm to realize the core idea while overcoming the above challenges. We consider $SG$ as a set of candidate target subgraphs. It is initialized as empty, and candidates are progressively added to $SG$ as they are identified. At any given time, $|SG| \leq k \leq n \log n.$ One can prove that up to $k$ non-target subgraphs empty out by $10 d \log n$ random tests with a probability that approaches $0$ as $ n \rightarrow \infty$; while the target subgraph never empties. Accordingly, we consider the first candidate   to enter $SG$ and survive  $10 d \log n$ further random tests after entering $SG$, as the target subgraph, and its sole remaining edge as the target edge. Thus, although non-target subgraphs can become candidates before the target subgraph, they would likely empty out in $10 d \log n$ tests after they become candidates, and therefore not be deemed the target subgraph. The overall algorithm proceeds as follows. 

\begin{enumerate}
    \item Set $SG = \phi$ and $\text{time}(G_j) = 0$ for $1 \leq j \leq n \log n$. If at any point, the total number of tests exceeds $2d n$, the algorithm outputs error. /*Here, the set $SG$ is the set of candidate subgraphs, ie, those with exactly one edge, and the set is updated after every test, $\text{time}(G_j)$ is the number of random tests  $G_j$ survives in $SG$. */
    \item While $\nexists i: \text{time}(G_i) \geq  10 d \log n$: perform a random test, update the edges.  For each $j \in SG$, update $\text{time}(G_j)$ as $\text{time}(G_j)+ 1$. Then, for all $k, l$ such that $|E_k| = 1$, $|E_l| = 0$,   $SG = (SG \cup G_k) \setminus G_l$ .   
    \label{line2}
    \item From the set $\{G_i | \text{time}(G_i) > 10 d \log n\}$, choose a $G_i$ at random and output its single edge as the target edge.
\end{enumerate}

Note that conforming to the SNAGT framework, the test designs are independent of previous test results (the nodes to be tested are always selected randomly from $V$ with the same probability distribution independent of the entire past and $V$ is not altered throughout the algorithm).  

We first prove that the above algorithm returns the target edge with a probability  that approaches $1$ as $n \rightarrow \infty.$ Note that once the  target subgraph joins $SG$ it will never be removed from it. So the only ways  the algorithm does not return a target edge is when at least one of the following two events happen: 1) the target subgraph does not enter $SG$ by $2 d (n - 5 \log n)$ tests 2) at least one non-target subgraph that enters $SG$ before or at the same test as the target subgraph remains candidate for $10 d \log n$ tests. Note that the first event implies that either the target subgraph does not become a candidate by the end of all tests or $10 d \log n $ tests are not conducted between when it becomes candidate and all tests end. We refer to these events as $err_1, err_2$ and show that their respective probabilities $\delta_n^1, \delta_n^2$ both approach $0$ as $n \rightarrow 1$. The claim made at the start of the paragraph therefore follows due to union bound. 

First, we upper bound the probability of  $err_1$ given that $G_i$ is the target subgraph. First note that $m_i \leq 2^n$ for each $i.$ Thus, $ d\log m_i + 2d\log n < 2d (n - 5 \log n)$, for all large $n$. Thus, 1) the algorithm conducts at least $ d\log m_i + 2d\log n$ tests, and 2) a necessary condition for   $err_1$ to happen is that one or more edges other than the target edge in $G_i$,  survive $d \log m_i + 2d\log n$ tests.   Now, following the analysis in \cite{gonen2022group}, the probability that one or more edges other than the target edge in $G_i$,  survive $d \log m_i + 2d\log n$ tests, is at most $1/n^2$. Since this upper bound does not depend on $i$, it also upper bounds the unconditional probability of $err_1.$ Thus, $\delta_n^1 \leq 1/n^2.$

We now upper bound the probability of $err_2.$ We will show that  the probability that a given non target edge, say  $e$,  survives $10 d \log n$ tests is at most $e^{-2.5 \log n}$. Note that fewer than $k$ non-target subgraphs enter $SG$ before or at the same time as the target subgraph.  Thus, together the non-target candidate subgraphs have fewer than $n \log n$ edges, none of which is the target edge. By union bound, the probability  that at least one of them  survives the $10 d \log n$ tests is upper bounded by $n (\log n) e^{-2.5 log n} \leq 1/n^{1.4}.$ Thus, $\delta_n^2 \leq (1/n^{1.4})$, and the result follows. We now prove the claim at the start of the paragraph. Note that  $e$ must have at least one negative node and it is detected as negative if it is selected in a test and none of the positive nodes are selected for that test. There are $d$ positive nodes overall since the target edge has $d$ nodes in a $d-$ regular hypergraph. Thus the probability that the negative node is detected in a test is $(1/d)(1-(1/d))^d.$ Thus in a random test  $e$ is ruled out as a target edge with a probability at least $(1/d)(1-(1/d))^d.$ The probability that  $e$ survives after $ 10 d \log n$ tests is therefore at most $$(1-(1/d)(1-1/d)^d)^{10 d \log n} \leq e^{-10 (1-1/d)^d \log n},$$  since $1-x \leq e^{-x}$ with $(1/d)(1-(1/d))^d$ as $x.$ Now, noting that $(1-1/d)^d \geq 1/4$ for $d \geq 2$, $ e^{-10 (1-1/d)^d \log n}
\leq e^{-2.5 \log n}$.  The claim has been proven then. 

We now seek to compute the expected number of tests for the above algorithm. We use the following insight. 
{We have shown in the previous paragraph that the target edge is returned with probability  $1-\delta_n^1 - \delta_n^2$. We now upper bound the number of tests in which it is returned with this probability. As noted earlier, given that $G_i$  is the target subgraph,  after $d \log m_i + 2 d\log n $  tests,  $G_i$ enters $SG$ (ie becomes candidate) with probability  $1-\delta_n^1$. Once it enters $SG$, the algorithm conducts exactly  $10 d \log n$ subsequent tests, because $G_i$ survives all of them and the algorithm terminates after at least one candidate survives $10 d \log n$ tests.  So with probability more than $1 - 1/n^2$, it performs $d \log m_i + 2 d\log n + 10d \log n $ tests conditioned on $G_i$ being the target subgraph, which happens with probability $ \mathbb{P}(G_i).$ 
Conditioned on $G_i$ being the target subgraph, with probability at most $\delta_n^1 \leq 1/n^2$, $G_i$ does not enter $SG$ by $d \log m_i + 2 d\log n $ tests, but even so the algorithm can use at most $2dn$  tests. Thus, the expected number of tests, $\mathbb{E}[T],$  is at most

    \[ \sum_{i} \mathbb{P}(G_i) [(1-\delta_n^1)(d \log(m_i) + 2d\log n + 10 d \log n ) + (1/n^2 * 2dn)] \]
    \[\leq  d\sum_{i} \left[(\sum_{e\in G_i} p(e)) \log(m_i)\right] + 12d \log n + O(d/n). \]

Recall that $m_i \leq 2/p(e)$ for all $e \in E_i$. Hence the RHS is

\begin{align}
    &\leq d \sum_i \sum_{e\in G_i} p(e) \log (2/p(e))  + 12 d\log n + O(d/n) \notag\\
    &= d \sum_{e \in E} p(e)(1 + \log(1/p(e)))  + 12d \log n + O(d/n)\notag \\
    &= d + d H(X)  + 12d \log n + O(d/n). \label{eq:proofent}
\end{align}

So, we get $O(d H(X) + d \log n)$ tests in expectation.

We have therefore proven the following theorem for the SNAGT framework for a $d-$regular hypergraph (considering that $u = d$ for any such hypergraph): 

\begin{theorem}\label{theorem:NATAS}
    Consider a graph $G$ with $I$ infected nodes in it. Let  $\Pr(I \geq u)$ be $o(1/n)$ for  $u \geq 2$. There is an algorithm that performs $O(u H(X) + u \log n)$ tests in expectation and recovers the infected nodes with an error probability that goes to $0$ as $n \rightarrow \infty$.
\end{theorem}
In Appendix~\ref{app:proof7} we prove the above theorem when the hypergraph is not regular.

Now,   \cite{gonen2022group} provides a non-adaptive combinatorial group-testing algorithm that finds the target edge in  a $d$-regular hypergraph in $O(d (\log |E| + \log n))$ tests, with probability at least  $1 - (1/n^2).$ Our algorithm is semi non-adaptive, and there is a probability distribution on edges for being the target edge. We obtain the same result as \cite{gonen2022group} in $O(d (H(X) + \log n) )$ tests in the SNAGT framework. When the probability over the edges is uniform, which maximizes the entropy, the number of tests for our algorithm is order-wise the same as \cite{gonen2022group}. It is already known that there is a close connection between combinatorial group testing and probabilistic group testing when the distribution is uniform. When the distribution is not uniform, $H(X)$ can be significantly smaller than $\log |E|$. Particularly, the testing is significantly improved when $|E|$ is $ \Omega(n)$ and the distribution over the edges is far from uniform, ie, $p(e)$ is large for some edges. Note that $m_i$ is small for a subgraph $i$ in which $p(e)$ is large. If such a subgraph is the target subgraph, it becomes a candidate with high probability in $O(d (\log m_i + \log n))$ tests and the algorithm terminates after conducting $10 d \log n$ additional tests. The total number of tests is small then if an edge with high probability of being a target edge is the target edge, because $m_i$ is small for subgraphs containing such edges.

The SNAGT algorithm needs larger, by a multiplicative factor of $d$ or $\log n$, number of tests than the adaptive algorithm. But it is easier to execute. Not only does it not need results of any test to determine who to test next, it does not need knowledge of exact probability that an edge is a target edge; it only needs to know for which $i$ is this probability between $2^{-i}$ and $2^{-(i-1)}.$

Finally,  in Theorem \ref{theorem:NATAS} we have assumed that the number of infected nodes is less than an upper bound $u$, equivalently  $|e^*|\leq u$, with high probability. This would for example hold if the number of infected nodes is concentrated around its mean, $\mu$, then $u = \mu + o(\mu)$. Incidentally, when the testing is done adaptively, we get the best results when the number of infections is concentrated around its mean $\mu$ (eg,  Corollary~\ref{cor:concentr},  Section~\ref{sec:analysis}).  Also, the testing algorithm,  which is described in Appendix~\ref{app:proof7}, and which has been used to prove Theorem \ref{theorem:NATAS}, does not need knowledge of  $\mu$, but only of a stochastic upper bound on the number of infected nodes.  Thus, the algorithm is intrinsically robust to errors in the estimation of the system parameters. Note that the upper bound on the expected number of tests for this Algorithm decreases with decrease in the stochastic upper bound on the number of infected nodes. Thus, if a tighter upper bound on the number of infected nodes can be estimated, then the guarantee on the performance of the Algorithm improves. Clearly, $ u \geq \mu.$ Thus, the guarantee on the expected number of tests decreases at most to  $O(\mu  H(X) + \mu \log n)$.

We have not been able to generalize the above technique to get general upper bounds for non-adaptive group testing. In our algorithm, the encoding stage (i.e., test design) is straightforward, consisting simply of random tests. However, due to the nature of our model, the decoding  is more sophisticated. Specifically, it updates posterior probabilities in real time, tracks potentially candidate subgraphs, and strategically determines when to halt the algorithm. These capabilities distinguish our semi-non-adaptive approach from the non-adaptive setting, where real-time decoding is not possible. 

%
%
Non-adaptive group testing has been extensively studied in the literature for independent group testing  \cite{li2014group,WangGG23,porat2008explicit}, and typically, {state of the art non-adaptive methods differ from the adaptive counterparts only by $O(\log n)$ or $O(d)$ (multiplicative factor)}. For instance, {  $O(d^2 \log n)$\cite{du2000combinatorial} is needed in the classic non-adaptive combinatorial setting (compared to $O(d \log n)$ tests for adaptive) or $O(\log n (H(X) + \mu))$\cite{li2014group} in non-adaptive independent group testing} (compared to $O(H(X) + \mu)$ tests for adaptive). Example~\ref{ex:nested} below shows that this is not necessarily true when states of nodes are correlated and that the gap between adaptive and non-adaptive methods can be large, {highlighting that the problem may be fundamentally more challenging than its counterpart in independent group testing settings.}

\begin{claim}
    {Any non-adaptive algorithm with probability error at most $\frac{k}{2n}, k \in \{0\} \cup \mathbb{N}$, needs to do at least $n-k-1$ tests in Example~\ref{ex:nested}.}
\end{claim}

\begin{proof}
Consider a non adaptive testing policy which conducts $l$ tests, the nodes tested in each test are given in $P = \{ T_1, T_2, \ldots, T_l \}$ where $T_i \subseteq V$ is the set of nodes tested in the $i$th test.  We argue that without loss of generality we can consider that each $T_i$ consists of only one node, otherwise consider that some $T_i$ consists of nodes $p, q$, where $q > p.$ Then if $q$ is positive, so is $p$, and removing $q$ from $T_i$ will yield the same result. Next, if $q$ is negative, removing $q$ from $T_i$ will yield the same result. Thus, $T_i$ will yield the same result if only its left-most node is tested. Thus, we consider that $P = \{i_1,  i_2, \ldots, i_l \}.$ Note that $l \leq n-1$, as node $1$ is positive, since otherwise all nodes are negative and there is no target edge, and thus, it is redundant to test node $1.$


If more than one edge is consistent with the test results,  the algorithm can do no better than to randomly (and uniformly) choose one of them as the target edge.  Let $i \not\in P$, $i \geq 2.$ Note that $e_{i-1} \subset e_i$, and $e_i \setminus e_{i-1} = \{i\}.$ A necessary condition for  $e_{i-1}$ to be determined as the target edge is that  $e_i$ be ruled out while $e_{i-1}$ is not ruled out. This can happen only if $i$ tests negative. If $e_{i-1}$ is the target edge, it can not be ruled out by any test, but then since $i$ is not tested, the status of node $i$ can not be determined, hence $e_i$ can not be ruled out either. 


Since $2 \leq i \leq n$, we considered edges in $\{ e_1, e_2, \ldots, e_{n-1}\}$ in the above argument. The above argument shows that at most for $l$ edges, among $e_1, \ldots, e_{n-1}$, the algorithm can have zero error given that any of these is the target edge, and for the rest $n-1-l$ edges there is at least a probability of $1/2$ for the algorithm to output a wrong edge given that one of them is the target edge.  Given that every edge is equally likely to be the target edge in Example~\ref{ex:nested}, 
the probability of error in this case is $\frac{n-1-l}{2n}$. Substituting $k$ with $n - l$ will prove the result.


\end{proof}

Note that there is an adaptive algorithm for the above example with $\log n$ tests. This shows that the gap between adaptive and non-adaptive algorithms can be as large as $n/\log n$. 


\section{Noisy Group Testing}\label{sec:noisy}


Up until now, we have assumed that the tests are noiseless, i.e. the test is positive iff there is at least one infection in the group that is tested. In this section, we extend our results to the case of noisy tests. {We consider the symmetric noisy model introduced in Section~\ref{sec:model}, where each test result is flipped with a probability $0 < \delta < 1/2$, independent of the other tests.}

If we repeat each test for at least $ \Theta(\frac{\log n}{(1-2\delta)^2})$ times, and then take the majority verdict, using a standard application of Chernoff bound, one can show that we get the correct result for all the tests, with a probability that goes to $1$ as $n$ goes to infinity. So we incur an additional {multiplicative} factor of $\Theta(\frac{ \log n}{(1-2\delta)^2})$ in the expected number of tests, and  the upper bound on the expected number of tests for Algorithm~\ref{mainalg} with each test repeated as above becomes $\Theta(\frac{\log n}{(1-2\delta)^2})O (H(X) + u)$ (Theorem~\ref{theorem:final} when $Pr[|e^*| > u]$ is $O(1/n)$). We show that Algorithm~\ref{mainalg} can be modified to improve the upper bound on the expected number of tests by updating the posterior  probability after each test to incorporate the impact of the noise in the tests and repeating the tests only in some steps: 

\begin{theorem}\label{theorem:noisyadapt}
Suppose that the result of each test is flipped with probability $\delta < 1/2$, and $ 1 - {(1-\delta)^{1-\delta} \delta^\delta} < c < 1/2. $ Also, $Pr[|e^*| > u] = O(1/n)$.  Then, there is an algorithm that performs $ \frac{1}{\log [1/(1-c)] - H(\delta)}H(X) + O(\frac{\log (g(n)u)}{(1-2\delta)^2} \frac{u}{(1-2c)}) + \Theta(\log^2 n/(1-2\delta)^2)$ tests in expectation and returns the correct edge $e^*$ with probability $1-O(1/\min(n, g(n))$. Here, $g(n)$ is any function that goes to infinity as $n \rightarrow \infty$.
\end{theorem}

\begin{remark}
We first compare the guarantee on the number of tests in Theorem~\ref{theorem:noisyadapt} with when noise is handled by repeating every test. 
Consider the case that $\delta$ does not increase with $n$ and $u$ is $o(n).$ Then $\delta$ can be replaced with its initial value in the upper bounds, which is a constant. Then the additional factor $\frac{1}{\log 1/(1-c) - H(\delta)}$ for the first term in the upper bound on the number of tests in Theorem~\ref{theorem:noisyadapt}, which is a constant and is therefore $o(\log n).$ The additional coefficient for $u$ in Theorem~\ref{theorem:noisyadapt} is $\frac{\log [g(n)u]}{(1 - 2\delta)^2}$, which is therefore $o(\log n)$ since $g(n)$ can be any function that increases to infinity as $n$ increases to infinity. So the coefficient of both $H(X)$ and $u$ are asymptotically substantially smaller than the multiplicative factor for the case that tests are merely repeated. But the upper bound in Theorem~\ref{theorem:noisyadapt} has an additional 
additive term of $\Theta(\log^2 n)$ as compared to the bound when the tests are merely repeated. However, the additive term is not an impediment if $H(X)$ is $\Omega(\log n).$ To see this, first consider that $H(X)$ is both $\omega(\log n)$ and $O(\log^2 n).$ The upper bound on the expected number of tests in Theorem~\ref{theorem:noisyadapt} is $\Theta(\log^2 n)$, whereas that when tests are merely repeated is $\omega(\log^2 n).$ Now, consider that $H(X)$ is $\Omega(\log^2 n)$. The upper bound on the expected number of tests in Theorem~\ref{theorem:noisyadapt} is $\Theta(H(X)$, whereas that when tests are merely repeated is $\Theta(\log n) \Theta (H(X)).$ Now, consider the case that $H(X)$ is $\Theta (\log n).$ Then, the upper bounds in Theorem~\ref{theorem:noisyadapt} and when we merely repeat tests are both $\Theta(\log^2 n)$. 
The additional additive term in the upper bound in Theorem~\ref{theorem:noisyadapt} is significant when $H(X)$ is $o(\log n)$, as by repeating each test $\Theta((\log n)/(1-2\delta)^2)$ times, both terms in the total number of tests $\Theta(\frac{\log n}{(1-2\delta)^2}) O(H(X) + u)$ are $o(\log^2 n)$ Thus, the upper bound in Theorem~\ref{theorem:noisyadapt} is  $\Theta(\log^2 n)$ whereas that when we only repeat each test is $o(\log^2 n)$. In this regime, simple well-known testing algorithms  can provide asymptotically vanishing error probability using limited number of tests in the basic noiseless case \footnote{For example, when the entropy is a constant, there are only a few number of edges with significant probability.  Then the complement of each such edge can be tested in decreasing order of size of edges. For the first test that is negative, all edges that are not fully contained in the edge whose complement has been tested can be removed. The same process can be repeated for edges that are fully contained in this edge. The largest remaining edge at the end of this process is the target edge. The number of tests equals the number of edges identified initially as those with significant probability of being a target edge.}, thus even the basic scenario in this case may not be of significant interest any longer. 

The coefficient for $H(X)$ in Theorem~\ref{theorem:noisyadapt} is $\frac{1}{\log 1/(1-c) - H(\delta)}$. It is known that $\frac{1}{1-H(\delta)} H(X)$ is a lower bound on the number of tests and is an optimal multiplicative factor in independent group testing with symmetric noise. Thus, the $\log(1/(1-c))$ is appearing due to the strategies for accommodating correlations in states of nodes. 

The theorem is not contingent upon $\delta$ being a constant, and allows $\delta$ to depend on $n$. In practice, for large communities (large $n$), the error probability can be lower as the local government (eg in big cities as opposed to rural counties) often has more resources and can therefore acquire more reliable test kits. The upper bound on the expected number of tests in Theorem~\ref{theorem:noisyadapt} increases with increase in the probability of error $\delta$ between $0$ and $1/2.$ Specifically,  when $\delta \rightarrow 1/2$, $ {(1-\delta)^{1-\delta} \delta^\delta} \rightarrow 1/2. $  Thus, since  $1 - {(1-\delta)^{1-\delta} \delta^\delta} < c < 1/2$,  $c \rightarrow 1/2$ too.  Then the coefficient of both $H(X)$, which is $\frac{1}{\log 1/(1-c) - H(\delta)}$, and the coefficient of $u$, which is $\frac{\log [g(n)u]}{(1 - 2\delta)^2}$ go to infinity. Thus, there is a tradeoff between error rate and guarantee on the number of tests.

\end{remark}

We now prove Theorem~\ref{theorem:noisyadapt}. 

\begin{proof}
We modify Algorithm~\ref{mainalg} as follows: 1)  each test  in line~\ref{alg:setnotfound} and line~\ref{alg:finaltest} is repeated for $\Theta(\log n / (1-2\delta)^2)$ and  $\Theta(\frac{\log g[(n) u]}{(1-2\delta)^2})$ times respectively.  Each test is deemed positive if the majority is positive, it is negative otherwise (if a test is not repeated the majority is the outcome in the lone instance). The tests  in other lines in the algorithm (ie,  line~\ref{alg:linetest})  are however not repeated as above. The posterior probability of each edge is updated based on the majority outcome of each test to incorporate the impact of the probabilistic flips in results. We also never perform more than $n$ tests, and the algorithm halts if and when that threshold is breached, it can halt even earlier due to the embedded decisions.

\paragraph{\bf Updating the posterior: } We now describe how we update the posterior for a test. Consider an edge $e$ and a test $T \subseteq V$. Let $T^o, T^r$ be the deemed and real results of the test respectively, the outcome of the test reveals $T^o$, but not $T^r.$ Let $T^o = a$, $q(e) = Pr(e = e^* \mid T^o = a)$, where $a \in \{-1, +1\}.$ Thus, $$q(e) = Pr(e = e^* \mid T^o = a, T^r = a)  Pr(T^r = a \mid T^o = a) + Pr(e = e^* \mid T^o = a, T^r = b)  Pr(T^r = b \mid T^o = a).$$ But $[T^r = a \mid T^o = a]$ is equivalent to result not being flipped which has probability $1 - \delta$, and $[T^r = b \mid T^o = a]$ is equivalent to the result being flipped, which has probability $\delta$. Moreover, $Pr(e = e^* \mid T^o, T^r) = Pr(e = e^* \mid T^r)$. Thus, if $T^o = a$, $$q(e) = (1-\delta) Pr(e = e^* \mid T^r = a)  + \delta Pr(e = e^* \mid T^r = b)  . $$ 

Note that $p(e) = Pr(e = e^* \mid T^r = -1) * Pr(T^r = -1)  + Pr(e = e^* \mid T^r = +1) * Pr(T^r = +1).  $ Clearly then $\max(Pr(e = e^* \mid T^r = -1), Pr(e = e^* \mid T^r = +1)) > 0$, as otherwise $p(e) = 0.$ Since $ 0 < \delta < 1,$ it follows that $q(e) > 0,$ if $p(e) > 0.$ Thus, if the posterior for an edge is positive to start with, it remains positive throughout. So no edge is ever eliminated, and the posterior for every edge always remains below $1$  (otherwise the posterior for other edges would reduce to $0$). Thus, line~\ref{alg:finaltest} is the only place that the algorithm actually outputs the infected nodes. 

\paragraph{Upper bounding the number of tests in   line~\ref{alg:linetest}: } We first provide a lower bound on the posterior of the target edge, and then using the lower bound,  we provide an upper bound on the number of tests in line~\ref{alg:linetest}.

\noindent {\bf Lower bounding the posterior: }  First let $e \cap T = \emptyset$, $Pr(e = e^* \mid T^r = -1) = \frac{p(e)}{\sum_{f \in E'}p_f} $, where $E'$ is the set of edges that have no intersection with $T$, and $Pr(e = e^* \mid T^r = +1) = 0.$ Thus,  if $T^o$ is negative, $q(e) =  (1-\delta) \frac{p(e)}{\sum_{f \in E'}p_f} $. Note that  $\sum_{f \in E'}p_f \leq 1-c$:  since $\sum_{f \in E'}p_f = w(V \setminus T)$, and by the algorithm $T$ is tested before line~\ref{alg:setnotfound} if  $c \leq w(V \setminus T) \leq 1-c$. Thus, if $T^o$ is negative, $q(e) \geq (1-\delta) \frac{p(e)}{1-c}.$

Now let $e \cap T \neq \emptyset$. Then,  $Pr(e = e^* \mid T^r = -1) = 0$, and $Pr(e = e^* \mid T^r = +1) = \frac{p(e)}{\sum_{f \in E'}p_f}$, where $E'$ is the set of edges that have a nonempty intersection with $T$. Now, by the algorithm, $c \leq w(V \setminus T) \leq 1-c$. Also, $w(E') = 1 - w(V \setminus T)$, since $E'$ consists of all edges which have at least one node in $T$ and therefore all edges that are not contained in $V \setminus T.$ Thus, since $w(V \setminus T) \geq c$, $w(E') \leq 1-c$. Thus,   if $T^o$ is negative, $q(e) =  \delta \frac{p(e)}{\sum_{f \in E'}p_f} \geq \delta \frac{p(e)}{1-c} $. 

Proceeding similarly, if $T^o$ is positive, 1)  $q(e) \geq \delta \frac{p(e)}{1-c}$ if $e \cap T = \emptyset$, and 2)  $q(e) \geq (1-\delta) \frac{p(e)}{1-c}$ if $e \cap T \neq \emptyset.$

Since $1- \delta \geq \delta,$ $ q(e) \geq \delta \frac{p(e)}{1-c}$, in all instances. 

\noindent{\bf Upper bound: } Suppose $e^*$ is the target edge and $M$ tests are performed in line~\ref{alg:linetest}. Consider the tests for which 1) $e^* \cap T = \emptyset$ and $T^o$ is negative, or 2) $e^* \cap T \neq \emptyset$ and $T^o$ is positive. In each such test, since every node in $e^*$ is positive and no node outside $e^*$ is positive, the test result has not been flipped. In each such test $q(e^*) \geq (1-\delta) \frac{p(e^*)}{1-c}$ by the  analysis in the previous paragraphs. Let the number of tests that have not been flipped be $M_1.$ It follows that  after the $M$ tests, $q_{e^*} \geq p(e^*)(1-\delta)^{ M_1}\delta^{|M| - M_1}(\frac{1}{1-c})^M.$
    

     Since the flips in results are i.i.d., by the Chernoff bound, we have $Pr[M_1 < (1-\Delta)\nu] \leq e^{-\frac{\nu \Delta^2}{2}}$ where $\nu = \mathbf{E}[M] = (1-\delta)M$. Setting $\Delta =  2\sqrt{\frac{\log n}{\nu}}$, we get $$Pr[M_1 < \nu - 2\sqrt{\log (n)\nu}] < e^{2 \log n} = 1/n^2.$$ 
    Hence, with probability more than $1-1/n^2$, we have $M_1 \geq (1-\delta)M - \partial{M_1}$ where $\partial{M_1} = 2\sqrt{\log (n)(1-\delta)M}$. Let the number of tests $M$ be greater than $\log^2 n.$ Then, $\frac{\partial M_1}{M} \rightarrow 0$ as $n \rightarrow \infty$, thus $\partial M_1 $ is $o(M)$ and $\frac{\partial M_1}{M}$ is also $O(1/\sqrt{\log n}).$

    By the new posterior update obtained, we can now lower bound the posterior after $M$  tests with high probability. With probability $1 - 1/n^2$, we have 
    \begin{align*}    
    q_{e^*} &\geq p(e^*) (1-\delta)^{(1-\delta) M - \partial M } \delta^{\delta M + \partial M} (\frac{1}{1-c})^M \\
    &= p(e^*)((\frac{\delta}{1-\delta}^{\frac{\partial M_1}{M}}(1-\delta)^{1-\delta} \delta ^\delta \frac{1}{1-c})^M \\
    &> p(e^*)(\frac{\delta}{1-\delta}^{\frac{\partial M_1}{M}}c')^M,
    \end{align*}
    
    where $c' =  (1-\delta)^{1-\delta} \delta ^\delta \frac{1}{1-c} >1$, and $\lim_{n \rightarrow \infty} \frac{\delta}{1-\delta}^{\frac{\partial M_1}{M}} = 1$. Note that no edge ever gets a probability of $1$, but  after at most $\frac{\log (1/p(e^*))}{\log (c')} + \log^2 n $ tests, $q(e^*) > 1$  with probability at least $1 - (1/n^2).$ Thus, the algorithm must have fewer than  $\log^2 n + \frac{\log (1/p(e^*))}{\log (c')}$ tests in line~\ref{alg:linetest} with probability at least $1 - (1/n^2)$; with the remaining probability it must have $n$ or fewer tests. Thus, given that $e^*$ is the target edge the expected number of tests in line~\ref{alg:linetest} is at most $\log^2 n + \frac{\log (1/p(e^*))}{\log (c')} + (1/n).$

    Since each edge $e$ is the target edge with probability $p(e)$,  the expected number of  tests $M$ in line~\ref{alg:linetest} is at most $\frac{1}{\log c'}H(X) + \log^2 n + (1/n).$ 


    In the above analysis we have implicitly assumed that the algorithm conducts tests in line~\ref{alg:linetest} consecutively, but in reality it may move to line~\ref{alg:setnotfound} from line~\ref{alg:linetest}, and return to line~\ref{alg:linetest} every time the majority outcome in  line~\ref{alg:setnotfound} is positive. But we show later that whenever such return happens, the posterior for the target edge is multiplied by a factor of at least $(1-\delta)/(1-c)$ with a high probability, the probability that this multiplication does not happen over all returns is at most $1/n.$ If the posterior for the target edge is multiplied by a factor of at least $(1-\delta)/(c) > (1-\delta)/(1-c)$ each time the algorithm returns to line~\ref{alg:linetest}, the  bound in the previous paragraph applies for the expected total number of tests in line~\ref{alg:linetest}, even if those are conducted with intermittent detours to line~\ref{alg:setnotfound} in between. If the posterior is not multiplied as above one or more times, the total number of tests is still at most $n.$ Since the probability that the above multiplication does not happen even once is at most $1/n,$ the expected number of  tests in line~\ref{alg:linetest} is at most $\frac{1}{\log c'}H(X) + \log^2 n + (1/n) + 1$, which is $\frac{1}{\log (1/(1-c)) - H(\delta)}H(X) + \Theta(\log^2 n)$ (since $c' =  (1-\delta)^{1-\delta} \delta ^\delta \frac{1}{1-c}$). 
   
   \paragraph{Analysis for repetition of each test in line~\ref{alg:setnotfound} and line~\ref{alg:finaltest}:}
   Let a test be repeated $\ell$  times, where  $\ell$ is $\Theta(\frac{z}{(1-2\delta)^2})$. Let $X = X_1 + \cdots+X_{\ell}$ where $X_i = 1$ if the test is not flipped. Then $\nu = \mathbb{E}[X] = (1-\delta)\ell$ and for $\Delta = 1 - \frac{1}{2(1-\delta)}$, $Pr[X < \ell/2] = Pr[X < (1-\Delta)\nu] \leq e^{-\frac{\nu \Delta^2}{2}}$ by Chernoff bound, since $X_i$s are $i.i.d.$ Replacing $\nu$ and $\Delta$, we have $e^{-\frac{\nu \Delta^2}{2}} = e^{-\frac{(1-\delta)\ell (\frac{1-2\delta}{2(1-\delta)})^2}{2}}$, and when $\ell$ is $\Theta(\frac{z}{(1-2\delta)^2})$, this probability is less than $e^{-2z}$. Hence with probability at least $1 - e^{-2z}$ the majority returns the correct result.

    \paragraph{Upper bounding  the number of tests in line~\ref{alg:setnotfound}:} 
Let $z = \log n$ in the previous paragraph. 
Then, if every test is repeated $\ell$ times, with $\ell$ as $ \Theta(\frac{\log n}{(1-2\delta)^2})$, with probability at least $1 - 1/n^2$, the majority of the tests returns the correct result. The probability that a noiseless test becomes positive, is less than constant $c$. After repetition the verdict is positive either because the group has at least one infected node, which happens with probability at most $c$, or because the majority verdict is in error, which happens with probability at most $(1/n^2).$ By union bound therefore, the majority verdict is positive with probability at most $c + (1/n^2).$ 
The algorithm moves from line~\ref{alg:setnotfound} to line~\ref{alg:linetest} whenever the majority verdict is positive, and there is always a return to line~\ref{alg:setnotfound} after such a move. Thus, the probability that the algorithm returns to line~\ref{alg:setnotfound} at least $(\log n)/(1-c-(1/n^2))$ times is  
\begin{align*}
(c + (1/n^2))^{\log n / (1-c-(1/n^2))} &=
(1 - (1-c-(1/n^2))^{\log n/ (1-c-(1/n^2))}\\ &< ( e^{-(1-c-(1/n^2))})^{(\log n)/(1-c-(1/n^2))}\\
&= e^{-(1-c-(1/n^2))(\log n)/(1-c-(1/n^2))} =  1/n.
\end{align*}

Hence, with probability of at least $1-(1/n)$, the total number of tests in this line is upper bounded by any function that is $\Theta(\frac{\log^2 n}{(1-2\delta)^2(1-c-(1/n^2))})$. Note that $\Theta(\frac{\log^2 n}{(1-2\delta)^2(1-c-(1/n^2))})$ is also $\Theta(\frac{\log^2 n}{(1-2\delta)^2})$ since $c < 1/2$. Even when the upper bound does not hold, the number of tests in this line is at most $n.$ Thus, the expected number of test in this line is upper bounded by $n\times (1/n) + $ any function that is $\Theta(\frac{\log^2 n}{(1-2\delta)^2})$ for all large $n$. Thus, the expected number of tests is upper bounded by any function that is $\Theta(\frac{\log^2 n}{(1-2\delta)^2})$ for all large $n$.

We now analyze the posterior every time the algorithm returns to  line~\ref{alg:linetest} from line~\ref{alg:setnotfound}. Such return is possible only if the majority outcome is positive. Suppose the majority outcome is positive and for an edge $e$, $e \cap T \neq \phi$ where $T$ is the tested set. Then, $ Pr(e = e^* \mid T^r = -1) = 0,$ $Pr(e = e^* \mid T^r = +1) = \frac{p(e)}{\sum_{f \in E'}p_f} $, where $E'$ is the set of edges each of which has at least one node in $T.$ By the algorithm, $w(V \setminus T) \geq 1-c$, so, $\sum_{f \in E'}p_f \leq c.$ Hence $q(e) \geq (1-\delta) p(e)/c \geq (1-\delta)p(e)/(1-c)$ since $ c < 1-c.$ Now, if the majority outcome of testing $T$ is positive, and $e^*$ is the target edge, either  (1) $e^* \cap T \neq \phi$ or (2) the majority outcome is erroneous. In the first case, $q(e^*) \geq (1-\delta) p(e^*)/(1-c).$ The probability of the second event is at most $1/n^2$ for any majority outcome. Since there are at most $n$ tests overall, probability of the second to ever happen is at most $1/n$. Thus, with overall probability at least $1-(1/n)$, every time the algorithm returns to line~\ref{alg:linetest} from line~\ref{alg:setnotfound}, $q(e^*) \geq (1-\delta) p(e^*)/(1-c)$.


    \paragraph{Upper bounding the number of tests in line~\ref{alg:finaltest}:} From the statement of the theorem we are proving here, with probability at most $1-O(1/n)$, the target edge has fewer than $u$ nodes. In this line, we repeat each individual test $\Theta(\frac{\log (g(n)u)}{(1-2\delta)^2})$ times, where $g(n)$ can be any function such that $\lim_{n \rightarrow \infty} g(n) = \infty.$ Thus, in the analysis for repetition of tests, $z = \log (g(n)u).$ 
    From the analysis for repetition of tests, probability that the majority verdict is erroneous is less than $\frac{1}{2g(n)u}$. Note that $|S|$ nodes are tested in this line, for a set $S$ identified in the previous steps. By Lemma~\ref{obs:L2}, $|S| \leq \frac{\tilde{\mu}}{1-2c}$, where $\tilde{\mu}$ is the expected size of the target edge at this stage. If the target edge has $u$ or more nodes, it will still have at most $n$ nodes. Thus, $\tilde{\mu}$ is at most $u$ + $n \times o(1/n)$, which is $O(u).$ Thus, $|S|$ is $O(\frac{u}{1-2c}).$ Thus, 1) the overall number of tests in this line is at most $O(\frac{\log (g(n)u)}{(1-2\delta)^2} \frac{u}{1-2c})$, and 2) by union bound, the probability that the majority verdict is erroneous for at least $1$ node in $S$ is $O(1/g(n))$, which implies that with probability at least $ 1- O(1/g(n)) \rightarrow 1$, the algorithm recovers the states of all nodes it tests in this line.

\paragraph{Upper bounding the overall number of tests and the probability that the target edge is not the output:} Combining the upper bounds for the expected number of tests in lines \ref{alg:linetest}, \ref{alg:setnotfound} and ~\ref{alg:finaltest}, the expected number of overall tests can be upper bounded by  $\frac{1}{\log 1/(1-c) - H(\delta)}H(X) + O( \frac{\log (g(n)u)}{(1-2\delta)^2} \frac{\mu}{(1-2c)} ) + \Theta(\log^2 n /(1-2\delta)^2)$ (any functions that are $\Theta(\cdot)$ of the said functions would suffice). 

As argued before, line~\ref{alg:finaltest} is the only place that the algorithm actually outputs the infectious nodes. The algorithm comes to this line only when the majority verdict of testing $V \setminus S$ is negative in line~\ref{alg:setnotfound}. If the majority verdict is correct, then all positive nodes are in $S$. Once the algorithm reaches line~\ref{alg:finaltest}, it may not  determine the actual set of infected nodes  only if either the majority verdict of testing $V \setminus S$ is erroneous the last time the algorithm is in line~\ref{alg:setnotfound}, or the majority verdict is erroneous for one or more nodes of $S$ in line~\ref{alg:finaltest}. The probability of the first event is at most $1/n^2$, and that of the second is at most  $O(1/g(n))$. Thus, once the algorithm reaches this line, it may not  determine the actual set of infected nodes with probability at most $(1/n^2) + O(1/g(n))$ by union bound. The above probability is $O(1/\min(n^2, g(n)).$ The algorithm may also not output any infected node if it exhausts $n$ tests before it reaches or concludes line~\ref{alg:finaltest}. The probability of this event is $O(1/n) + O(1/g(n))$, which is $O(1/\min(n, g(n)).$ Thus, the probability that the algorithm does not output the target edge is $O(1/\min(n, g(n)).$

\end{proof}

The next theorem extends our upper bound for the SNAGT setting to the noisy setting. 

\begin{theorem}
    Suppose the result of each test is flipped with probability $\delta < 1/2$, and for a graph $G$, the probability distribution $\mathcal{D}$ is such that $\Pr_{e \sim \mathcal{D}}(|e| > u) = o(1/n)$ for any constant $\epsilon > 0$. Then, there is an algorithm in the SNAGT setting that performs $O(\frac{ \log [un]}{(1-2\delta)^2} [u \cdot H(X) + u \log n)])$ tests in expectation and recovers the infected nodes with an error probability that goes to $0$ as $n \rightarrow \infty$.
\end{theorem}

\textbf{Proof Sketch:} We employ the same algorithm described in Theorem~\ref{theorem:NATAS}, but modify it to address noise by incorporating test repetition and majority voting. Specifically, each test is repeated $O(\frac{ \log [dn]}{1-4\delta(1-\delta)})$
    times, and the outcome of the test is determined by the majority vote among these repetitions.

    By choosing the number of repetitions as above, we ensure that the probability of a single test producing an incorrect majority outcome is at most $ \mathrm{err} \ll \frac{1}{un}.$
    Applying a union bound over all tests, the probability that \emph{any} test is incorrect tends to zero as $n$ grows.

    Consequently, with high probability, we may assume that all tests are effectively noiseless. Under this assumption, the situation reduces exactly to the noise-free setting of Theorem~\ref{theorem:NATAS}. Hence, the same performance guarantees and conclusions from that theorem apply, completing the proof.


\section{Conclusion and Future Work}\label{sec:last}

In this paper, we developed a novel framework for group testing with general statistical correlation among nodes.  Our model considers a probability distribution on all possible candidate subsets of nodes, captured by a hypergraph. We developed group testing algorithms that exploit correlation in a variety of settings, namely adaptive, semi-non-adaptive group testing, noisy etc. We provide constructive testing algorithms and upper bound the expected number of tests they need to attain asymptotically vanishing probabilities of error. We investigate special cases of interest, such as random regular hypergraphs, and fine tune our algorithms therein to obtain order-wise optimal number of tests for vanishing probabilities of error for sparse and dense graphs. We also obtain lower bounds on the expected number of tests needed for attaining vanishing probability of error in the general case and several special cases, which in turn reveal if the expected number of tests for our constructive algorithms  can be reduced through improved design. The modifications of the algorithms we provide have several desirable characteristics: 1) ability to schedule tests in advance without waiting for results 2) innate robustness to errors in estimation of parameters 3) tradeoffs between number of tests and accuracy of estimates and accuracy of tests. 



While we discussed conditions under which our algorithm and its corresponding performance upper bound are near-optimal, it remains open whether our algorithm has optimal performance in general. Future work will focus on providing tighter lower and upper bounds for both adaptive and semi-non-adaptive settings, exploring the role of the correlation structure. Additionally, we aim to explore the design of non-adaptive algorithms despite the challenges highlighted in this paper. This includes designing such algorithms for general graphs or specific families of graphs that could significantly reduce the number of tests needed. Furthermore, we will continue to investigate new correlation models that could expand the practical applicability of our framework, ultimately enhancing our ability to perform efficient group testing in diverse, real-world environments.

\section{Acknowledgment}
 The authors would like to acknwledge Dominic Olaguera-Delogu for helpful discussions and contributions to Example~\ref{ex:biggraph}.
The work of Hesam Nikpey and Shirin Saeedi Bidokhti was supported  by the NSF CAREER award 2047482. The work of Saswati Sarkar was supported by NSF grant 2008284, and SIGA. The authors gratefully acknowledge this support.

\bibliographystyle{alpha}
\bibliography{ref}

\appendix









\section{Modeling Statistical Models with Hypergraphs}\label{app:hypmodel}

Here, we give explicit formulas to compute the equivalent hypergraphs that model previous works.

As discussed in Section~\ref{sec:priormodel}, In \cite{nikpey2022group} they consider a simple graph $H = (V_H, E_H)$, and random graph $H_r$ is obtained by keeping each edge with probability $r$. Then, every component is infected with probability $p$, i.e. all the nodes inside are infected, or not infected with probability $1-p$, i.e. all the nodes inside are not infected. In order to compute $\mathcal{D}(S)$, the probability that exactly set $S \subseteq V_H$ is infected, in graph $H_r$, set $S$ should be disconnected from the rest. Let $\mathcal{H}_{S}$ be such graphs. Let $C(H_S)$ be the number of components on graph $H$ induced by $S$, and $E(H_S)$ be the number of edges in $H_S$ induced by $S$. The $\mathcal{D}(S)$ is computed as follows:

$$\mathcal{D}(S) = \sum_{H \in \mathcal{H}_S} p^{C(H_S)} \cdot (1-p)^{C(H_{V \setminus S})} \cdot r^{E(H_S) + E(H_{V \setminus S})} \cdot (1-r)^{E(H_V) - E(H_S) - E(H_{V \setminus S})}.$$

In \cite{arasli2023group}, they have the same model except only one component is infected uniformly at random. Define $\mathcal{H_S}$ be the graphs where $S$ is a connected component, disconnected from the rest of the graphs. Then $\mathcal{D}(S)$ is computed as:
$$\mathcal{D}(S) = \sum_{H \in \mathcal{H}_S} \frac{1}{C(H_V)} \cdot r^{E(H_S) + E(H_{V \setminus S})} \cdot (1-r)^{E(H_V) - E(H_S) - E(H_{V \setminus S})}.$$

For \cite{ahn2021adaptive}, we can compute $\mathcal{D}(S)$ by taking the sum over all possible seeds. Let $N(v)$ be the number of seeds with the same community as $v$ and $n$ be the number of nodes. Then:

$$\mathcal{D}(S) = \sum_{T \subseteq S} q^T \cdot (1-q)^{n-T} \cdot \prod_{v \in T \setminus S} [1-(1-q_1)^{N(v)} (1-q_2)^{|T| - N(v)}].$$


\section{Proof of \cite{ahn2023adaptive} From Section~\ref{sec:5}}\label{app:provecoms}

Let $C_i$ be a random variable that is $1$ when community $i$ is infected and $S = \sum_i C_i$. We show that under the condition $kp \ll 1$, $S$ is unlikely to get values much greater than its mean. We have $\mathbb{E}[S] = m \mathbb{E}[C_1] = m(1-p)^k \simeq me^{-pk} = e^{\log m - pk}$. We know that $C_i$'s are independent of each other, hence by Chernoff bound we have

$$\mathbb{P}(S > (1+\delta)\mathbb{E}[S]) < e^{\frac{\mathbb{E}[S] \delta^2}{2}} = e^{\frac{e^{\log m - pk} \delta^2}{2}}.$$

Now because $pk \ll 1$, the above probability is $O(e^{m\delta^2})$. By having large constant $\delta$, we have $S = O(\mathbb{E}[S])$ with high probability.

Now that the number of infected communities is concentrated, we argue that the number of infections in a community is also concentrated. Note that when $kp \ll m^{-\beta}$, we only see at most a constant number of seeds in each community (with high probability). Since  every non-seed node is exposed to a constant number of seeds, the probability that a non-seed node becomes infected is $O(q)$, and hence the expected number of infections in each community is $O(kq)$ and the expected number of non-seed infections is $O((n-m)q)$ and the non-seed nodes are independent of each other, i.e. one being infected or not does not affect other nodes.
 Now by Chernoff bound the probability that the number of infections is greater than $2(n-m)q = O(kmq) \gg m$ is $o(1)$. This shows that the total number of infections is not more than a constant factor greater than its mean, and hence the condition of Corollary~\ref{cor:concentr} holds. 

 \section{Proof of Theorem~\ref{theorem:NATAS}}\label{app:proof7}

We already proved the theorem when every edge has size $d$. The assumption in the theorem implies that if $e^*$ is the target edge, $\Pr(|e^*| < u) = 1 - o(1/n)$. We are using the same random testing strategy as for the $d-$regular case, with the difference that now each node is selected in any given test with probability $1/u$ (instead of $1/d$). The algorithm is the same as well, except that a preprocessing step is added in the beginning in which any edge $e$ whose size is at least $u$ is removed. 
Note that no node is removed in the preprocessing. 


The only ways  the algorithm does not return a target edge is when at least one of the following events happen: 1) the target subgraph does not enter $SG$ by $2 u (n - 5 \log n)$ tests 2) at least one non-target subgraph that enters $SG$ before or at the same test as the target subgraph remains candidate for $10 u \log n$ tests, 3) the number of infected nodes  is at least $u$  (since all the infected edges together constitute the target edge, then the target edge is removed in the preprocessing). Note that the first event implies that either the target subgraph does not become a candidate by the end of all tests or $10 u \log n $ tests are not conducted between when it becomes candidate and all tests end. We refer to these three events as $err_1, err_2, err_3$ and show that the respective probabilities of $err_1 \cap (err_3)^c, err_2 \cap  (err_3)^c, err_3$,  $\delta_n^1, \delta_n^2, \delta_n^3$ are $o(1/n).$ Then the probability of the union of the three error events is also $o(1/n)$ by union bound.

For upper-bounding $\delta_n^1, \delta_n^2$ we first upper bound the probability that an edge $e$ that is not the target edge survives $T$ successive random tests given that $I < u.$  Suppose  $e^*$ is the target edge. Note that if $e \not\subset e^{*}$, then $e$ must have at least one negative node and it is detected as negative if it is selected in a test and none of the positive nodes are selected for that test. If $e \subset e^*$, then there is a node in $e^*$ that is not in $e$ and this node is positive. Call the node in each of the cases as $v$. Note that if $v$ appears in a test without any other node in $e^*$, then $e$ will be eliminated in both cases (In the second case, $v$ just has to appear in a test without any  node in $e$, but $e \subseteq e^*$; so if $v$ appears without any other node in $e^*$, it also appears without any other node in $e.$). Given that there are at most $u$ positive nodes overall, the probability that $v$ appears in a test without other nodes in $e^*$ is at least $(1/u)(1-(1/u))^u .$ Thus, given that $I < u$, probability that  $e$ is ruled out as a target edge in a random test is  at least $(1/u)(1-(1/u))^u.$ The conditional probability that  $e$ survives after $T$ tests is therefore at most $(1-(1/u)(1-(1/u))^u)^T \leq e^{-(1-1/u)^u T/u}$ since $1-x \leq e^{-x}$ with $(1/u)(1-(1/u))^u$ as $x.$ Now, noting that $(1-1/u)^u \geq 1/4$, for all $u \geq 2$, we  have the conditional probability upper bounded by $e^{-T/(4 u)}.$ Since the probability that $I < u$ is at most $1$, the  probability that edge $e$ that is not the target edge survives $T$ successive random tests and $I < u$  is at most $e^{-T/(4 u)}.$

Now, we upper bound the probability of  the intersection of $err_1$ and that $I <  u$  given that $G_i$ is the target subgraph. First note that $m_i \leq 2^n$ for each $i.$ Thus, $u \log m_i \leq u n$, and $ u \log m_i + 8 u \log n < 2 u (n - 5 \log n)$ for all large $n$. Thus, 1) the algorithm conducts at least $ u \log m_i + 8 u \log n$ tests, and 2) a necessary condition for   $err_1$  to happen is that one or more edges other than the target edge in $G_i$,  survive $u \log m_i + 8 u\log n$ tests.   The probability that a given edge $e$ in $G_i$ that is not the target edge 
survives  $u \log m_i + 8 u\log n$ successive random tests and  $I < u$ is at most $e^{-(\log m_i + 8 \log n)/4 } =  (1/m_i)(1/(n^2)). $ Now, there are $m_i - 1$ edges in $G_i$ that are not target edges. Thus, by union bound, the conditional probability is at most $1/(n^2).$ Since the upper bound does not depend on $i$, this is also the upper bound on $\delta_n^1$, the probability of the intersection of $err_1$ and $I <  u.$


We now upper bound the probability of the intersection of $err_2$ and $I < u.$ 
The probability that edge $e$ that is not the target edge survives $10 u \log n$ successive random tests and $I < u$  is at most $e^{-10 (u \log n)/(4 u)} = e^{-2.5 \log n} = (1/n^{2.5}).$ 
At most $n \log n$ non-target subgraphs enter $SG$ on or before the target subgraph. So the probability that $ I < u$ and at least one of those remain candidate after $10 u \log n$ tests is at most the probability that $I < u$ and at least one among $n \log n$ target edges survive $10 u \log n$ tests. By union bound the probability that $I < u$ and at least one among $n \log n$ target edges survive $10 u \log n$ tests is at most $(n \log n)/n^{2.5} = (\log n)/n^{1.5}.$ This also upper bounds $\delta_n^2$, the probability of the intersection of $err_2$ and $I <  u.$

By assumption in the theorem,  $ \delta_n^3$ is $o(1/n).$

We now analyze the expected number of tests, similar to that for $d$-regular graphs. Note that with probability $1 - o(1/n)$ neither of $err_1,err_2$, or $err_3$ happen, and then when additionally the $i$'th subgraph is the target subgraph, the algorithm performs $u \log m_i + 8 u \log n + 10 u \log n$ tests. Note that when any of the $err_i$'s happens, at most $2u n $ tests are performed. Hence the expected number of tests is at most:

    \[ \sum_{i} \mathbb{P}(G_i) [(1- (\delta_n^1+\delta_n^2+\delta_n^3)) (u \log(m_i) + 8 u \log n + 10 u \log n ) + (\delta_n^1+\delta_n^2+\delta_n^3) 2u n] \]
    
    \[\leq  u\sum_{i} \left[(\sum_{e\in G_i} p(e)) \log(m_i)\right] + 18 u \log n + o(u). \]

Recall that $m_i \leq 2/p(e)$ for all $e \in E_i$. Hence the RHS is

\begin{align}
    &\leq u \sum_i \sum_{e\in G_i} p(e) \log (2/p(e))  + 18 u\log n + o(u) \notag\\
    &= u \sum_{e \in E} p(e)(1 + \log(1/p(e)))  + 18 u \log n + o(u)\notag \\
    &= u + u H(X)  + 18 u \log n + o(u). \label{eq:proofent}
\end{align}

And hence the total number of tests is at most $O(u H(X) + u \log n)$ in expectation.

\end{document}

%% file: headers.tex
\usepackage{amsmath,amsfonts,amsthm,amssymb}
\usepackage{color}
\usepackage[colorlinks=true]{hyperref}
\usepackage[procnumbered,ruled,vlined,linesnumbered]{algorithm2e}
\usepackage{amsfonts}

\newtheorem{theorem}{Theorem}[section]
\newtheorem{corollary}[theorem]{Corollary}
\newtheorem{lemma}[theorem]{Lemma}
\newtheorem{observation}[theorem]{Observation}
\newtheorem{proposition}[theorem]{Proposition}
\newtheorem{claim}[theorem]{Claim}

\newtheorem{definition}[theorem]{Definition}
\newtheorem{remark}[theorem]{Remark}

\newenvironment{fminipage}%
  {\begin{Sbox}\begin{minipage}}%
  {\end{minipage}\end{Sbox}\fbox{\TheSbox}}

\DontPrintSemicolon
\SetKw{KwAnd}{and}
\SetProcNameSty{textsc}
\SetFuncSty{textsc}
\SetKwInOut{Input}{Input}
\SetKwInOut{Output}{Output}

%% file: graph.tex
\begin{tikzpicture}
\node[vertex,label=above:\(v_1\)] (v1) {};
\node[vertex,above right of=v1,label=above:\(v_2\)] (v2) {};
\node[vertex,below right of=v1,label=above:\(v_3\)] (v3) {};
\node[vertex,below left of=v1,label=above:\(v_4\)] (v4) {};
\node[vertex,above left of =v1,label=above:\(v_5\)] (v5) {};

\begin{pgfonlayer}{background}
\draw[edge,color=yellow] (v1) -- (v5);
\begin{scope}[transparency group,opacity=.5]
\draw[edge,opacity=1,color=green] (v1) -- (v2) -- (v3) -- (v1);
\fill[edge,opacity=1,color=green] (v1.center) -- (v2.center) -- (v3.center) -- (v1.center);
\end{scope}
\draw[edge,color=red,line width=40pt] (v4) -- (v5);
\end{pgfonlayer}

\node[elabel,color=green,label=right:\(e_1 \text{, $p_{e_1} = 0.3$}\)]  (e1) at (4,1) {};
\node[elabel,below of=e1,color=yellow,label=right:\(e_2 \text{, $p_{e_2} = 0.2$} \)]  (e2) {};
\node[elabel,below of=e2,color=red,label=right:\(e_3 \text{, $p_{e_3} = 0.5$} \)]  (e3) {};
\end{tikzpicture}

%% file: graph_n-1.tex
\begin{tikzpicture}
\node[vertex,label=above:\(v_1\)] (v1) {};
\node[vertex, right of=v1,label=above:\(v_2\)] (v2) {};
\node[vertex,below  of=v2,label=above:\(v_3\)] (v3) {};
\node[vertex,below of=v1,label=above:\(v_4\)] (v4) {};

\begin{pgfonlayer}{background}
\begin{scope}[transparency group,opacity=.8]

\draw[edge,opacity=.8,color=yellow] (v1) -- (v2) -- (v3) -- (v4) -- (v1);

\draw[edge,opacity=.8,color=green] (v1) -- (v2) -- (v3) -- (v1);

\draw[edge,opacity=8,color=red] (v1) -- (v2) -- (v4) -- (v1);

\draw[edge,opacity=.8,color=black] (v1) -- (v3) -- (v4) -- (v1);

\draw[edge,opacity=.8,color=white] (v2) -- (v3) -- (v4) -- (v2);

\end{scope}
\end{pgfonlayer}

\end{tikzpicture}